\documentclass[11pt]{article}
\usepackage{subcaption}
\usepackage[margin=1in]{geometry}
\usepackage{amsthm}
\usepackage{thmtools,thm-restate}

\newcommand{\doi}[1]{\href{https://doi.org/#1}{\nolinkurl{doi: #1}}}

\newtheorem{lemma}{Lemma}
\newtheorem{theorem}{Theorem}

\usepackage{hyperref}

\usepackage{amssymb}
\usepackage{amsmath}
\usepackage{bm}
\usepackage{cleveref}
\crefname{algocf}{alg.}{algs.}
\Crefname{algocf}{Algorithm}{Algorithms}

\usepackage{mathtools}
\usepackage[dvipsnames]{xcolor}
\usepackage{xspace}
\usepackage{todonotes}

\usepackage[sort&compress,numbers]{natbib}
\usepackage[ruled,vlined]{algorithm2e}
\usepackage{nicefrac}

\newcommand{\algalpha}{\textsc{Alg}_{\alpha\textsc{-tsp}}}
\newcommand{\algtsp}{\textsc{Alg}_{\textsc{tsp}}}
\newcommand{\tsp}{\textsc{tsp}}
\newcommand{\alphatsp}{\alpha\textsc{-tsp}}
\newcommand{\alphapsp}{\alpha\textsc{-psp}}

\title{Approximation Algorithms for Discounted Graph Search\\ with Norm Objectives\thanks{
The work of the first author is funded by the Deutsche Forschungsgemeinschaft (DFG, German Research
Foundation) by the grant Ho 3831/9-1 (project ID: 514505843).
The work of the second author is funded by the Deutsche Forschungsgemeinschaft (DFG, German Research
Foundation) by the grant HO 7562/2-1 (project ID: 573939419).
The work of the third author is funded by the Deutsche Forschungsgemeinschaft (DFG, German Research
Foundation) under Germany´s Excellence Strategy – The Berlin Mathematics
Research Center MATH+ (EXC-2046/1, EXC-2046/2, project ID: 390685689).
We thank Kevin Schewior for fruitful discussions.}}

\author{
Svenja M.\ Griesbach\\
RWTH Aachen University, Department of Computer Science, Germany\\
\texttt{griesbach@algo.rwth-aachen.de}
\and
Felix Hommelsheim\\
University of Cologne, Department of Computer Science, Germany\\
\texttt{hommelsheim@cs.uni-koeln.de}
\and
Max Klimm\\
Technische Universit\"at Berlin, Institute for Mathematics, Germany\\
\texttt{klimm@math.tu-berlin.de}
}
\date{}

\newtheorem{claim}{Claim}
\newtheorem{example}{Example}

\newcommand*{\psp}{discounted graph search problem\xspace}

\newcommand*{\E}{\mathbb{E}}
\newcommand*{\N}{\mathbb{N}}
\newcommand*{\R}{\mathbb{R}}
\newcommand*{\Z}{\mathbb{Z}}
\newcommand*{\alg}{\textsc{Alg}}
\newcommand*{\opt}{\textsc{Opt}}
\newcommand*{\eps}{\varepsilon}
\newcommand{\alphamax}{\alpha_{\max}}
\newcommand{\alphamin}{{\alpha_{\min}}}
\newcommand{\lfrac}[2]{#1/#2}
\tikzstyle{axis}=[thin]
\tikzstyle{state}=
[circle,draw=black,thick,fill=black,minimum size=1.5mm,inner sep=0pt]
\usetikzlibrary{calc,arrows.meta,patterns}
\newcommand{\trianglesize}{1ex}
\newcommand{\meindreieck}{
\begin{tikzpicture}
    \coordinate (A) at (0, 0);
    \coordinate (B) at (\trianglesize, \trianglesize/2);
    \coordinate (C) at (\trianglesize, -\trianglesize/2);
    
    \draw[very thin] (A) -- (B) -- (C) -- cycle;
\end{tikzpicture}
}

\definecolor{myyellow}{cmyk}{0, 0.45, 0.90, 0.17}
\newcommand{\myblue}{MidnightBlue}
\newcommand{\myyellow}{myyellow}
\newcommand{\mygreen}{LimeGreen}
\newcommand{\myred}{Red}
\usepackage{bm,upgreek}
\renewcommand{\vec}[1]{
\begingroup
\let\alpha\upalpha
\let\beta\upbeta
\let\theta\uptheta
\let\lambda\uplambda
\let\mu\upmu
\let\tau\uptau
\let\varphi\upvarphi
\bm{\mathbf{#1}}
\endgroup
}

\newcommand{\edgec}{cost\xspace}
\newcommand{\lat}{C}
\newcommand{\Seq}{\Pi}
\newcommand{\seq}{\pi}

\newcommand*{\factor}{\varrho_1(\alpha)}
\newcommand{\setn}{[n]}

\newcommand{\newparagraph}[1]{\subparagraph*{#1.}}

\begin{document}

\maketitle
\thispagestyle{empty}

\begin{abstract}
We introduce a unified framework for classical search and routing problems, including pathwise search, expanding search, the minimum spanning tree problem, and the traveling salesperson problem. The framework is based on two parameters. The first is a discount factor $\alpha \in [0,1]$: the first traversal of an edge incurs its full \edgec, whereas each subsequent traversal incurs only an $\alpha$-fraction of this \edgec. For a path starting at a designated root vertex, the $\alpha$-latency of a vertex is the discounted \edgec accumulated until the vertex is first visited. The second parameter is a norm parameter $p\geq 1$. The objective is to find a root-starting path that visits all vertices and minimizes the $p$-norm of the resulting vector of $\alpha$-latencies.

The model interpolates between several well-studied objectives. For $p=1$ and $\alpha=1$, it recovers pathwise search; for $p=1$ and $\alpha=0$, it recovers expanding search. As $p$ tends to infinity, the objective converges to a makespan-type criterion. At the endpoints $\alpha=1$ and $\alpha=0$, this limiting objective corresponds to TSP-type and MST-type behavior, respectively. For $p=1$, we give polynomial-time constant-factor approximation algorithms for all $\alpha\in[0,1]$, matching the best known guarantees for expanding search at $\alpha=0$ and pathwise search at $\alpha=1$. For general $p\geq 1$, we obtain a randomized constant-factor approximation algorithm and a derandomized pseudo-polynomial-time algorithm with the same guarantee.
\end{abstract}

\section{Introduction}

Pathwise search, expanding search, the traveling salesperson problem, and the minimum spanning tree problem are four fundamental models for exploring or connecting a network.
At first glance, these problems optimize rather different objectives: pathwise and expanding search minimize sums of discovery times, while the traveling salesperson problem and the minimum spanning tree problem minimize the time until all vertices are reached or connected.
In this paper, we study a common framework that unifies these problems through two parameters.
The first parameter determines how repeated traversals of edges are charged, and the second parameter determines how the individual vertex latencies are aggregated.

We are given an undirected graph $G = (V,E)$ with non-negative edge \edgec~$c_e \in \mathbb{N}$, and a designated start vertex $s$.
A solution is a traversal of the graph, i.e., a path starting in $s$ that may visit edges more than once and eventually visits all vertices.
For a vertex $v$, its latency is the time at which $v$ is visited for the first time.
Classical search problems ask for a traversal that minimizes the sum of these latencies.

The \emph{pathwise search problem} asks for such a traversal when every traversal of an edge $e$ requires~$c_e$ time units.
Thus, the latency of a vertex is equal to the total \edgec of the prefix of the path until the vertex is first visited, and the goal is to minimize the sum of the latencies of all vertices.
The problem, also known as the \emph{traveling repairperson problem}, captures situations in which all traversals of an edge take the same amount of time.
Hence, it has been used as a model for the movement of a repairperson in a road network or for disk heads on a hard drive.
The problem is $\mathsf{NP}$-hard even on weighted trees (\citet{sitters2002minimum}),
so much research has focused on approximation algorithms.
An algorithm is a $\varrho$-approximation algorithm if for any instance it runs in polynomial time and the cost of the solution output by the algorithm is at most $\varrho \cdot \opt$, where $\opt$ denotes the cost of an optimum solution to the respective instance.
The factor $\varrho$ is called the approximation ratio or guarantee.
The best currently known approximation algorithm for general graphs with unit weights for all vertices yields a $3.59$-approximation (\citet{chaudhuri2003paths}).

The \emph{expanding search problem} asks for a sequence of edges with the property that every prefix of edges is a connected subgraph and, without loss of generality, a tree containing $s$.
The latency of a vertex is the total \edgec of the edges added until the vertex first appears in the sequence, and the goal is again to minimize the sum of the latencies of all vertices.
It captures situations in which a tree network needs to be installed, and the \edgec of an edge is interpreted as the time needed to establish a connection between its end vertices.
Hence, it has been used as a model for clearing paths in an area devastated by disasters or for mining.
The problem is $\mathsf{NP}$-hard (\citet{averbakh2012flowtime}) and the best currently known approximation algorithm for general graphs yields a $5.44$-approximation (\citet{GriesbachHKS26}).

We propose and study a general model of graph exploration that we term the \emph{discounted graph search problem}.
As in pathwise search, a solution is a traversal $\seq$ starting in $s$ that may visit edges more than once and eventually visits all vertices.
The first parameter of the model is a discount factor~$\alpha \in [0,1]$.
While the first traversal of an edge $e$ in the sequence requires $c_e$ time units, every further traversal of the same edge requires only $\alpha \cdot c_e$ time units.
We call the \edgec of a path where repeated traversals are discounted in this way the \emph{$\alpha$-\edgec} of the path.
Given a traversal $\seq$, the \emph{$\alpha$-latency} $\lat_{\alpha,v}(\seq)$ of a vertex~$v$ is defined as the $\alpha$-\edgec of the smallest prefix of $\seq$ that visits $v$.

The second parameter is a norm parameter $p\geq 1$, which determines how the individual $\alpha$-latencies are aggregated.
For a traversal $\seq$, its $p$-norm $\alpha$-latency is~$\smash{\left\lVert \vec{\lat}_{\alpha}(\seq) \right\rVert_p
=
\left(
\sum_{v\in V} \lat_{\alpha,v}(\seq)^p
\right)^{1/p}}$.
The goal is to compute a traversal $\seq$ minimizing this quantity.

This two-parameter model contains several classical problems as special cases or limiting cases.
For $p=1$ and $\alpha=1$, every traversal of an edge is charged its full cost, and the objective is the sum of the first-visit times of all vertices.
Thus, we recover the pathwise search problem.
For $p=1$ and~$\alpha=0$, repeated traversals of already used edges are free.
Therefore, only the cost of newly added edges contributes to the discovery time of vertices, and we recover the expanding search problem.

The parameter~$p$ interpolates between sum-of-latencies objectives and makespan-type objectives.
As $p\to\infty$, the $p$-norm objective converges to the maximum $\alpha$-latency of any vertex.
For~$\alpha=1$, this limiting objective asks for a shortest traversal starting at $s$ that visits all vertices, since the maximum latency is exactly the time when the last vertex is first reached.
This is the path version of the traveling salesperson problem.
For $\alpha=0$, the maximum $\alpha$-latency is the total cost of the distinct edges that have been introduced by the time all vertices are reached.
Minimizing this quantity is therefore equivalent to finding a minimum-cost connected subgraph spanning all vertices, and hence to the minimum spanning tree problem.
Intermediate values of $\alpha \in (0,1)$ capture situations where the first traversal of an edge is more time-consuming than later traversals, for example, because of additional delays due to pathfinding or clearing a road, while further traversals still require a non-negligible amount of time.
Intermediate values of $p$ capture settings in which one wants to balance the average discovery time of vertices with the time until the last vertices are reached.
The parameter $p$ can also be interpreted from a fairness perspective.
For $p=1$, the objective minimizes the total latency and thus corresponds to a social-welfare objective, whereas larger values of $p$ put increasing emphasis on vertices with large latency; in the limit $p\to\infty$, the objective becomes an egalitarian objective that minimizes the worst latency.
The special case of $p=2$ and $\alpha=1$ has been studied as the \emph{traveling firefighter problem} in the literature~\cite{farhadi2021traveling}. The authors motivate this particular choice of objective by firefighters looking for an order in which to tackle wildfires whose damage grows quadratically with the elapsed time.

\subsection{Our Results}

We first note that the discounted graph search problem is computationally hard already in the case $p=1$.
In particular, the hardness of pathwise and expanding search carries over to the endpoints $\alpha=1$ and $\alpha=0$, and we also show that the problem remains hard for every fixed intermediate value $\alpha\in(0,1)$; see \Cref{sec:alpha-hardness}.

\begin{restatable}{theorem}{thmhardness}
	\label{thm:hardness-psp}
	For $p=1$ and every constant~$\alpha\in[0,1]$, there exists a constant~$\eps>0$ such that there is no polynomial-time~$(1+\eps)$-approximation algorithm for the \psp with discount factor~$\alpha$, unless~$\mathsf{P}=\mathsf{NP}$.
\end{restatable}

We therefore focus on approximation algorithms.
We first consider the case $p=1$, where the objective is the total $\alpha$-latency.
For this case, we obtain a polynomial-time approximation algorithm with approximation ratio
\begin{align*}
	\factor \coloneqq
 \begin{cases}
2\mathrm{e} & \text{if } \alpha = 0,\\
2 \frac{\alpha}{(1+\alpha)W(\alpha/\mathrm{e})} & \text{otherwise,}
\end{cases}
\end{align*}
where $W:\R_{\geq 0}\rightarrow \R_{\geq 0}$ is the Lambert-$W$ function, which assigns $x\geq 0$ the unique value~$w\geq 0$ such that $w\mathrm{e}^w=x$.
This gives the following result.

\begin{restatable}{theorem}{umain}
\label{thm:u-main}
	For $p=1$ there is a polynomial-time $\factor$-approximation algorithm for the~\psp.
\end{restatable}

\begin{figure}
\begin{subfigure}[t]{0.47\textwidth}
\centering
\resizebox{\linewidth}{!}{%
\begin{tikzpicture}[yscale=0.9,xscale=6]
\draw[-latex,thick] (-0.05,0) -- (1.05,0) node[right] {$\alpha$};
\draw[-latex,thick] (0,0.5) -- (0,3) node[above] {$\factor$};
\draw[dotted,thick] (0,0) -- (0,0.5);
\draw (0.015,2.44) -- (-0.015,2.44) node[left] {$2\mathrm{e}$};
\draw (0.015,0.59) -- (-0.015,0.59) node[left] {$3.59$};
\draw (0.015,1) -- (-0.015,1) node[left] {$4$};
\draw (0.015,2) -- (-0.015,2) node[left] {$5$};

\draw (0,0.1) -- (0,-0.1) node[below] {$0$};
\draw (1/4,0.1) -- (1/4,-0.1) node[below] {$1/4$};
\draw (1/2,0.1) -- (1/2,-0.1) node[below] {$1/2$};
\draw (3/4,0.1) -- (3/4,-0.1) node[below] {$3/4$};
\draw (1,0.1) -- (1,-0.1) node[below] {$1$};
\draw[smooth,very thick,MidnightBlue] plot coordinates {(0,2.44) (0.1,2.12096) (0.2,1.85261) (0.3,1.62121) (0.4,1.41929) (0.5,1.24129) (0.6,1.083) (0.7,0.94116) (0.8,0.8132) (0.9,0.69707) (1,0.59112)};
\node[circle,minimum size=5pt,inner sep=0pt,fill=RedViolet] (chaud) at (1,0.59)  {};
\node[label=left: \textcolor{RedViolet}{\citet{chaudhuri2003paths}}] (chaud1) at (1,0.4)  {};
\node[circle,minimum size=5pt,inner sep=0pt,fill=ForestGreen,label=above right: \textcolor{ForestGreen}{\citet{GriesbachHKS26}}] (gries) at (0,2.44)  {};
\node (this) at (0.5,1.8) {\textcolor{MidnightBlue}{this work}};
\end{tikzpicture}
}
\end{subfigure}
\hfill
\begin{subfigure}[t]{0.47\textwidth}
    \centering
    \includegraphics[width=\linewidth]{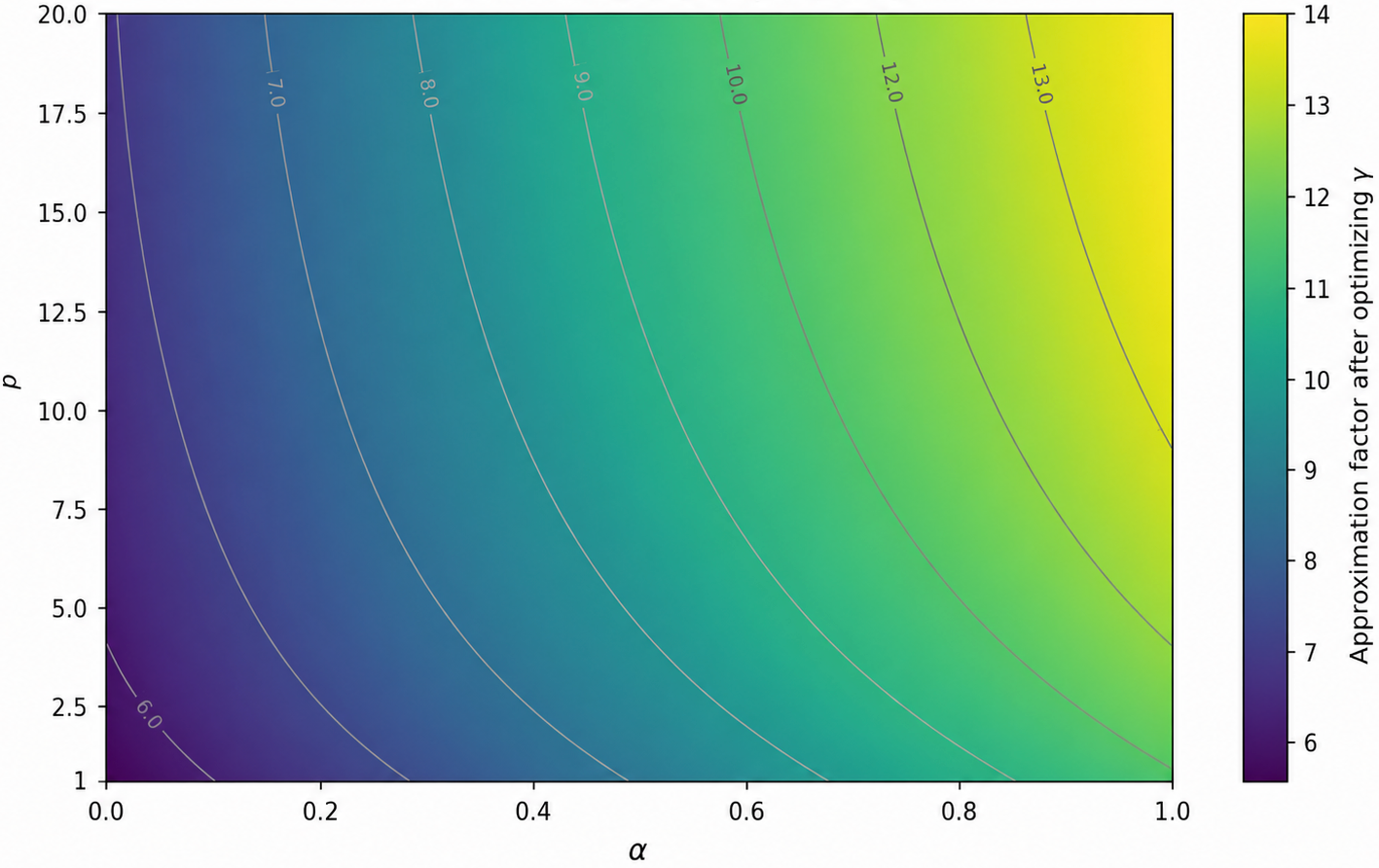}

    \label{fig:p-norm-heatmap}
\end{subfigure}

\caption{Approximation ratios $\varrho(\alpha)$ obtained in this work: (a) for $p=1$ as a function of $\alpha$; (b)~as a function of~$p>1$~and~$\alpha$.\label{fig:approximation-ratio}}
\end{figure}
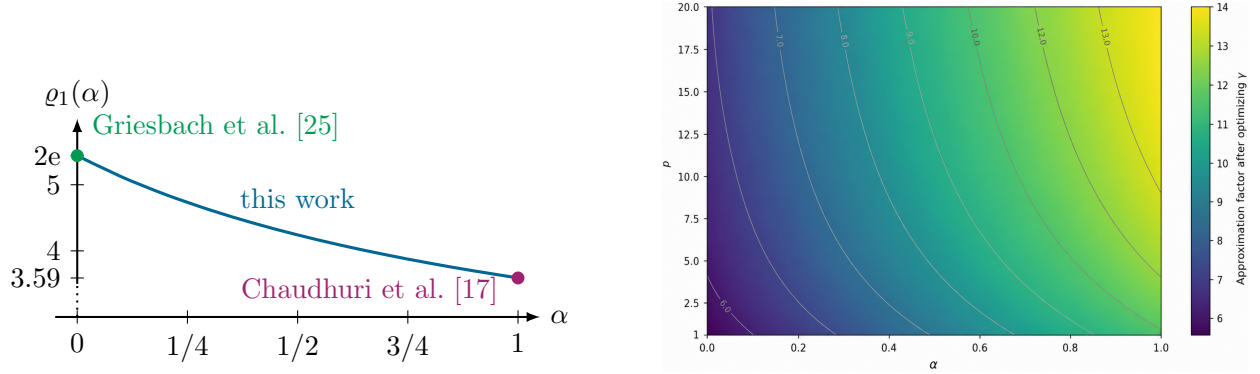

For $\alpha = 0$, the approximation factor is $2\mathrm{e}$, which matches the best currently known approximation guarantee for the expanding search problem (\citet{GriesbachHKS26}).
For $\alpha = 1$, the approximation factor is $3.59$, which matches the best currently known approximation guarantee for the pathwise search problem in general graphs with unit vertex weights (\citet{chaudhuri2003paths}).
For arbitrary values of $\alpha \in [0,1]$, the approximation ratio interpolates smoothly between these two values; see \Cref{fig:approximation-ratio}.

Like the algorithms of \citet{chaudhuri2003paths} and \citet{GriesbachHKS26}, our algorithm is based on a sequence of certain trees of exponentially growing \edgec.
The trees that we consider are related to the concept of \emph{good} trees used by \citet{chaudhuri2003paths} in their $3.59$-approximation for the pathwise search problem.
Roughly speaking, a tree is called \emph{good} if it contains a certain number of vertices and its \edgec is bounded by the cost of an appropriate path visiting the same number of vertices.
We extend this concept to general values of $\alpha$.
Specifically, we define good $k$-trees for $\alpha$ and show that such trees can be computed in polynomial time for the values of $k$ required by our algorithm.
For $\alpha=1$, this recovers the trees used by \citet{chaudhuri2003paths}.
For $\alpha=0$, the good $k$-trees correspond to $k$-MSTs, and the resulting bound matches the best known approximation guarantee for expanding search.

We then turn to the more general case $p\geq 1$, where the objective is the $p$-norm of the vertices' $\alpha$-latencies.
For $p>1$, the approach for the case $p=1$ cannot be applied directly.
The main difficulty is that the objective is no longer linear in the latencies, and hence the auxiliary-graph construction used for $p=1$ no longer directly captures the contribution of each additional tree.
Moreover, the good-tree machinery from the $p=1$ case does not provide suitable trees for every value of~$k$.

We therefore use a randomized construction that is closer in spirit to the geometric tree sequences used for $p=1$, but relies only on 2-approximate $k$-MSTs.
For every $k\in[n]$, we compute such a tree, choose a random geometric offset, and concatenate trees whose costs grow geometrically.
Analyzing the $p$-th moments of the resulting vertex latencies gives the following guarantee.

\begin{restatable}{theorem}{pmain}
\label{thm:u-main-p}
	For every $p\geq 1$, there is a polynomial-time randomized algorithm for the \psp with approximation factor
    \[ \varrho_p(\alpha) =
    \min_{\gamma > 1} \left\{  \frac{2 (1 + \alpha)}{\gamma-1} \biggl[\frac{\gamma^{2p}-\gamma^p}{p \ln \gamma}\biggr]^{1/p} \right\}.
    \]
\end{restatable}

The approximation factor from \Cref{thm:u-main-p} is depicted in \Cref{fig:approximation-ratio}.
For fixed $\alpha$, the optimized factor tends to $8(1+\alpha)$ as $p\to\infty$.
Over the range $\alpha\in[0,1]$, the smallest value is attained at $\alpha=0$ and $p=1$, where the factor is $2\mathrm{e}$, while the largest value is approached for $\alpha=1$ and $p\to\infty$, where the factor tends to $16$.

For the traveling firefighter problem, corresponding to the cases~$\alpha = 1$ and~$p=2$, this yields an approximation factor of~$\approx 11.28$. This is worse by a factor of $2$ than the approximation factor of $\approx 5.65$ claimed by \citet{farhadi2021traveling}. However, as we discuss in \Cref{ssec:related-work}, we do not believe that this discrepancy can be resolved straightforwardly.

Finally, we derandomize the $p$-norm algorithm using an auxiliary graph, as in the case $p=1$.
The difference is that for $p>1$ the state must also record the accumulated cost of the trees traversed so far, since this quantity determines the nonlinear contribution to the $p$-norm objective.
This yields a two-dimensional, time-expanded auxiliary graph and the following deterministic guarantee.

\begin{restatable}{theorem}{pderandomized}
\label{thm:p-derandomized}
For every $p\geq 1$, there is a deterministic pseudo-polynomial-time algorithm for the \psp with approximation factor $\varrho_p(\alpha)$ defined in \Cref{thm:u-main-p}.
Moreover, if the tree costs are integral and polynomially bounded, then the algorithm runs in polynomial time.
\end{restatable}

\subsection{Related Work}
\label{ssec:related-work}
The pathwise search problem is also known as the \emph{traveling repairperson problem}.
This problem was shown to be $\mathsf{NP}$-hard on general graphs by \citet{sahni1976complete}, and to be $\mathsf{NP}$-hard even on weighted trees by \citet{sitters2002minimum}.
In the weighted setting, each vertex has a non-negative weight and the task is to minimize the weighted sum of latencies.
The problem can be solved efficiently on unweighted trees and paths~\cite{afrati1986complexity,garcia2002note,minieka1989delivery}.
The first approximation algorithm for this problem was devised by \citet{blum1994minimum} who gave a 144-approximation for general graphs and an $8$-approximation for weighted trees.
\citet{goemans1998improved} improved the approximation ratio for general graphs to $21.55$ and the one for weighted trees to $3.59$.
In particular, they show that a $\beta$-approximation for the rooted $k$-MST problem of finding a shortest tree containing a given root and $k-1$ further vertices implies an approximation ratio of $3.59\beta$ for the traveling repairperson problem.
Thus, the improved $k$-MST approximations of $\beta=3$ by \citet{garg1996approximation}, of~$\beta=2.5$ by \citet{AryaRamesh1998}, of~$\beta= 2+\eps$ by \citet{arora2006approximation}, and of~$\beta=2$ by \citet{Garg05} immediately yield better approximations of $10.77$, $8.98$, $7.18+\eps$, and $7.18$, respectively. In addition, the work of \citet{Arora98} implies an improved approximation ratio of $3.59 + \eps$ for the Euclidean case.
\citet{archer2008faster} showed that the approximation ratio of $7.18 + \eps$ can be obtained with fewer calls to the $k$-MST subroutine.
\citet{chaudhuri2003paths} improved the approximation ratio for general graphs to match the ratio of $3.59$ for weighted trees. \citet{archer2010improved} further improved the approximation ratio for weighted trees to $3.03$.
\citet{AroraK03} gave a quasi-polynomial-time approximation scheme (QPTAS) both for weighted trees and for Euclidean instances. \citet{sitters2021polynomial} gave a polynomial-time approximation scheme (PTAS) for Euclidean instances, weighted trees, and planar graphs.

The expanding search problem has been shown to be $\mathsf{NP}$-hard by~\citet{averbakh2012flowtime}. \citet{alpern2013mining} proposed an algorithm for the expanding search problem on weighted trees.
The first approximation algorithm was devised by~\citet{HermansLM22} who gave an $8$-approximation.
The approximation ratio was improved by~\citet{GriesbachHKS26} to $2\mathrm{e} \approx 5.44$.

\citet{farhadi2021traveling} study the $L_p$-TSP, where the objective is to minimize the~$p$-norm of the vector of visit times; this problem interpolates between the traveling repairperson problem for $p=1$ and the path variant of TSP for $p \to \infty$ and corresponds to our case of $\alpha=1$.
They claimed to have a universal $8$-approximation that is valid for all $p$ simultaneously. If correct, this would improve over a $16$-approximation by~\citet{golovin2008all}.
\citet{farhadi2021traveling} further claimed a $5.65$-approximation for $p=2$. However, we believe that both their results have a substantial gap.
Their algorithms rely on the ability to compute good $k$-trees for all values $k \in [n]$; Lemma 3.1 in their paper cites \citet{chaudhuri2003paths} for the claim that this can be done in polynomial time. However, this is not what \citet{chaudhuri2003paths} show. They only show ``how to obtain such $k$-trees for some values of $k$'' and ``that
it actually suffices to use the few $k$-trees that we found''~\cite[p.~3]{chaudhuri2003paths}.
The algorithms of \citet{farhadi2021traveling}, on the other hand, crucially require good~$k$-trees for \emph{all} values of $k$ to be computable in polynomial time.
It seems to be an open problem whether this can be done.
A positive answer to this open problem appears to require new techniques, as~\citet[p.~1]{Garg05} writes that ``we cannot argue that the cost of the tree picked is at most
the cost of the best $k$-stroll. As was shown in~\cite{chaudhuri2003paths} such an
argument can be made for certain values of $k$ by exploiting
the same slack in the Goemans-Williamson argument''.
Our algorithms for $p > 1$ avoid this issue by computing $2$-approximate $k$-MSTs for all values of $k$ instead of good $k$-trees; the former can be done with the algorithm of~\citet{Garg05}. This is exactly why we incur an additional loss of a factor of $2$. The discussion of \citeauthor{Garg05} cited above makes it plausible that this additional factor of $2$ cannot be avoided with current techniques.
This route was also taken by~\citet{golovin2008all}, who use $2$-approximate $k$-MSTs for all values of $k$ instead of good $k$-trees.
It is interesting to note that our approximation guarantees for any~$p \geq 1$ and $\alpha \in [0,1]$ are strictly below $16$ and approach this value when $\alpha=1$ and $p \to \infty$.

Norm-based interpolation also appears in ordered optimization problems such as ordered $k$-median, which interpolates between $k$-median and $k$-center~\cite{chakrabarty2018interpolating}.
In scheduling and load balancing, one often minimizes the $L_p$-norm of the machine-load vector, which similarly interpolates between average-load objectives and makespan minimization; see, e.g., the work of \citet{awerbuch1995load} and subsequent work on all-norm and minimum-norm load balancing~\cite{azar2004allnorm, chakrabarty2019minimum, ibrahimpur2021minimum}.
While these problems differ from ours, they illustrate the role of $p$-norm objectives as a natural way to interpolate between sum-type and bottleneck-type criteria.

\section{Preliminaries}
\label{sec:prelim}

Let $G=(V,E)$ be an undirected graph with $|V|=n+1$
and a designated start vertex $s$ and let $\alpha \in [0, 1]$ be the discount factor.
Each edge $e\in E$ has a non-negative \edgec $c_e\in \N$.
We assume that $c_e  > 0$ for every edge incident to $s$ since otherwise these edges can be contracted.
We call a sequence of edges $\seq=(e_1,\dots,e_k)$ a \emph{path} if for all $1\leq i<k$ the end vertex of $e_i$ and the start vertex of $e_{i+1}$ coincide and the sequence starts in $s$.
If the path also ends in $s$, we call it a \emph{tour}.
The set of all tours is denoted by $\Seq$.
For two paths $\seq=(e_1,\dots,e_k)$ and $\seq'=(e'_1,\dots,e'_l)$ where the end vertex of~$\seq$ coincides with the start vertex of~$\seq'$, we write $\seq+\seq'=(e_1,\dots,e_k,e'_1,\dots,e'_l)$ for the concatenation of those two paths.

For a sequence of edges $\seq$, we define the \emph{$\alpha$-\edgec} $c_{\alpha}(\seq)$ of $\seq$ as the sum of all edges in $\seq$, where for the second and further traversals of an edge $e$ its \edgec $c_e$ is multiplied by $\alpha$.
We call a sequence $\seq$ a \emph{tree} if the set of edges in $\seq$ spans a tree in $G$ that contains $s$.
Since a tree contains no edge more than once, the $\alpha$-\edgec of a tree is independent of $\alpha$, and we simply write $c(\seq)$.
The following lemma states the simple fact that every tree can be transformed into a tour by increasing the $\alpha$-\edgec by at most a factor of $1+\alpha$. We state it here for future reference; its proof is deferred to \Cref{app:tree-to-tour}.

\begin{restatable}{proposition}{treetotour}
    \label{prop:tree-to-tour}
    For a tree $T$, there exists a tour $\seq(T)$ that visits all vertices in $T$ and has $\alpha$-\edgec at most $c_{\alpha}(\seq(T))\leq (1+\alpha)c(T)$.
\end{restatable}

For a vertex $v\in V$ visited by a tour $\seq\in \Seq$, we denote by $\seq_v$ the prefix of $\seq$ until $v$ is visited for the first time.
The \emph{$\alpha$-latency} $\lat_{{\alpha},v}(\seq)$ of vertex $v$ in $\seq$ is defined as $\lat_{{\alpha},v}(\seq)\coloneqq c_{\alpha}(\seq_v)$
where we set $\lat_{\alpha,v}(\seq) = \infty$ if $v$ is not visited by~$\seq$ at all.
Note that $\lat_{{\alpha},s}(\seq) = 0$.
The \emph{total $\alpha$-latency} $\lat_{\alpha}(\seq)$ of $\seq$ is then defined as $\lat_{\alpha}(\seq)\coloneqq \sum_{v\in V}\lat_{{\alpha},v}(\seq)$.
In particular, we have $\lat_{\alpha}(\seq) = \infty$ if there is a vertex $v \in V$ that is never visited by~$\seq$.
The set of feasible solutions for the \emph{\psp} is the set $\Seq$.
Further, a tour $\seq^*\in\Seq$ is optimal for the \psp with $\alpha\in[0,1]$ if $\lat_{\alpha}(\seq^*)\leq \lat_{\alpha}(\seq)$ for all tours $\seq\in\Seq$.

A \emph{$k$-tree}, \emph{$k$-path}, or \emph{$k$-tour} is a tree, path, or tour, respectively, that contains $s$ and visits at least $k+1$ distinct vertices of $V$.
For a $k$-tour $\seq\in \Seq$, we denote by $\seq_k$ the prefix of $\seq$ until the $k$-th distinct vertex  of $V \setminus \{ s\}$ is visited for the first time.
A $k$-path $\seq^*$ is \emph{optimal} with respect to~$\alpha$ if
$c_{\alpha}(\seq^*)\leq c_{\alpha}(\seq)$ for all~$k$-paths $\seq$.
If~$\alpha$ is clear from context, we simply say that $\seq^*$ is optimal.
Further, a~$k$-tree $T$ is called \emph{good} if its \edgec is at most
$
	c(T)\leq \frac{2}{1+\alpha} c_{\alpha}(\seq^*),
$
where $\seq^*$ is an optimal $k$-path.

\section{Algorithms for the~$1$-Norm of~$\alpha$-Latencies}
\label{sec:unweighted}

Recall that $\varrho(\alpha) = 2 \frac{\alpha}{(1+\alpha)W(\alpha/\mathrm{e})}$ when $\alpha \in (0,1]$ and $\varrho(0) = 2\mathrm{e}$.
We prove the following result.

\umain*

The intuitive idea behind the algorithm is to concatenate a specific subset of \emph{good trees}.
We construct an initial set of good trees in~\Cref{ssec:u-good-trees-alpha} using an algorithm by \citet{chaudhuri2003paths}.
They introduced this algorithm for the computation of good $k$-trees for a specific subset of values of $k$ and in the special case of $\alpha = 1$.
A more thorough analysis of their approach yields that the obtained trees are indeed good trees for arbitrary choices of $\alpha \in [0,1]$.
Since the algorithm does not return good $k$-trees for all values of $k\in\setn_0$, we use so-called phantom trees for the remaining values of~$k$ in \Cref{ssec:u-phantom-trees}.
Next, we introduce the concatenating algorithm in \Cref{ssec:u-concatenating-algo}, which takes as input a set of good trees and returns a feasible solution to the \psp.
More precisely, it constructs an auxiliary graph~$H$ whose vertices correspond to the good trees.
Within this auxiliary graph, the algorithm computes a shortest path and concatenates the corresponding trees to the desired tour.
The algorithm is thus well-defined only if the input contains only real trees.
We temporarily ignore this constraint and show in \Cref{ssec:u-concatenating-algo-phantom} that, under this relaxation, the returned (phantom) tour has an approximation guarantee of~$\factor$.
The analysis builds on the construction of a randomized path in the auxiliary graph that visits exponentially growing trees.
Then, we show that the same guarantee can be obtained by a path that visits only real trees.
This allows us to restrict the input to good real trees while maintaining the approximation guarantee.
Combining these results yields a well-defined algorithm that returns the desired solution.
The running-time analysis then concludes the proof of \Cref{thm:u-main}.

\subsection{Finding Good Trees}
\label{ssec:u-good-trees-alpha}

To obtain good $k$-trees, we use an algorithm by \citet{chaudhuri2003paths}.
They study the pathwise search problem, i.e., the \psp with $p=\alpha=1$, for which they introduce a primal--dual algorithm that computes good $k$-trees for a subset of values of $k$.
With a more thorough analysis, we can show that these trees are indeed good trees for all choices of $\alpha\in [0,1]$.

A subtle issue is that the primal--dual procedure is not monotone in the vertex-budget parameter.
In contrast to what is sometimes assumed in related Lagrangian approaches for $k$-MST~\cite{arora2006approximation}, increasing the parameter need not produce a larger tree, even before the final pruning step.
This is because the parameter can change the order in which moats merge and become inactive.
An explicit example can be found in \Cref{sec:app:non-monotonicity}.

In \Cref{app:finding-good-trees}, we give a detailed explanation of how the framework of \citet{chaudhuri2003paths} needs to be adapted to be able to compute good $k$-trees for all values of $k \in [n]$.
For this purpose, it is important to distinguish between values of $k$ for which a good $k$-tree can be computed and those for which this is not the case.
In the latter case, we define a phantom tree, which is not an actual tree but a placeholder whose costs are interpolated from the costs of the closest two real $k$-trees.
We later argue that these phantom trees are not used by our approximation algorithm, so their inclusion as placeholders does not introduce any issues down the line.

The main result of \Cref{app:finding-good-trees} in the appendix is the following lemma.

\begin{restatable}{lemma}{ugoodphantomtrees}
	\label{lem:u-good-phantom-trees}
	There is a polynomial-time algorithm that computes a set~$\mathcal{T}$ containing, for each value~$k\in \setn_0$, a real or phantom tree~$T_k\in\mathcal{T}$.
	Each such $T_k$ is a good $k$-tree.
\end{restatable}

\subsection{The Concatenating Algorithm}
\label{ssec:u-concatenating-algo}

We introduce the \emph{concatenating algorithm}, which takes as input a set of good $k$-trees, selects a subset of these, and concatenates them to obtain a solution for the \psp.
To this end, let $\mathcal{I}$ be a set of good $k$-trees for a subset of values $k\in\setn_0$ and let $I\subseteq \setn_0$ be such that $k\in I$ if and only if~$\mathcal{I}$ contains a good $k$-tree denoted by~$T_k$.
Then the concatenating algorithm with input~$\mathcal{I}$ is defined as follows:
\begin{enumerate}
	\item Construct a directed auxiliary graph $H=(V_H,A_H)$ with vertices $V_H=I$, edges $A_H=\{(i,j) : i<j\}$, and edge lengths $\ell_{i,j}=\bigl[(1+\alpha)n-\alpha j-i\bigr]c(T_j)$.
	\item \label{it:u-path} Compute a shortest $0$--$n$-path $P^*=(n_0,n_1,\dots,n_l)$ with $n_0=0$ and $n_l=n$ in $H$.
	\item \label{it:u-tour} Start with the empty sequence $\seq_\alg=(\cdot)$. \\
		For each phase $j=1,\dots,l$, construct the tour $\seq(T_{n_j})$ according to \Cref{prop:tree-to-tour}.
		Let $\seq_\text{f}$ and $\seq_\text{b}$ be the tour obtained from traversing the sequence $\seq(T_{n_j})$ in forward or backward direction, respectively.
		Let $V_{j}\coloneqq V_{n_j} \big\backslash \big(\bigcup_{i=1}^{j-1} V_{n_i}\big)$ be the set of vertices that are contained in $T_{n_j}$ but in no previous tree $T_{n_i}$ for $1\leq i<j$.
		Let $\seq_{\min} \in\{\seq_{\text{f}},\seq_{\text{b}}\}$ be such that
		$\sum_{v\in V_{j}}\lat_{{\alpha},v}(\seq_\alg+\seq_{\min})$ is minimized.
		Set $\seq_\alg\gets \seq_\alg+\seq_{\min}$.
\end{enumerate}

Note that the concatenating algorithm is only well-defined under the following two assumptions.
First, the input must contain a good $0$-tree and a good $n$-tree.
Second, all trees~$T_{n_0},\dots,T_{n_l}$ on the shortest path must be real trees, because otherwise Step~\ref{it:u-tour} is not well defined.
The second constraint is trivially fulfilled if the input~$\mathcal{I}$ only contains real trees.

\subsection{The Concatenating Algorithm with Phantom Trees}
\label{ssec:u-concatenating-algo-phantom}

We want to apply the concatenating algorithm to the set~$\mathcal{T}$ of good trees whose existence is guaranteed by \Cref{lem:u-good-phantom-trees}.
However, this set may contain phantom trees, so the concatenating algorithm is not necessarily well-defined on this input.
We ignore this issue for now by assuming that all phantom trees are real trees and can therefore be transformed into tours.
Under this assumption, we analyze the approximation ratio of the resulting tour~$\seq_\alg$ and then show how to obtain the same guarantee when restricting to real trees.

To this end, let~$\mathcal{T}$ be the set of good $k$-trees whose existence is guaranteed by \Cref{lem:u-good-phantom-trees}.
Note that $\mathcal{T}$ contains a good (real or phantom) $k$-tree \emph{for all} values of $k\in\setn_0$ and that the trees~$T_0$ and~$T_n$ are real trees.
Before applying the concatenating algorithm to~$\mathcal{T}$, we slightly adjust the cost of the phantom trees~$\mathcal{P}\subset\mathcal{T}$.
More precisely, for a phantom tree~$T_k\in\mathcal{P}$, let $k_l<k$ be maximal and $k<k_r$ be minimal such that $T_{k_l},T_{k_r}\in\mathcal{T}$ are two real trees.
We redefine the cost of $T_k$ to

\begin{align}
	\label{eq:u-cost-phantom-tree}
	c(T_k)\coloneqq (1-\mu)c(T_{k_l})+\mu\, c(T_{k_r})
	\quad \text{with }
	\quad \mu\coloneqq\frac{k-k_l}{k_r-k_l}.
\end{align}
Since the phantom tree~$T_k$ was originally constructed by a linear interpolation of two real trees, this redefinition of the cost can only decrease its cost.
See also~\cite{chaudhuri2003paths} for a more detailed discussion.
Hence, the phantom tree~$T_k$ remains a good $k$-tree.
We denote the set~$\mathcal{T}$ with the adjusted costs of the phantom trees by~$\mathcal{T}'$.
We analyze the approximation guarantee of the tour~$\seq_\alg$ obtained by running the concatenating algorithm on input~$\mathcal{T}'$.
In this section, we prove the following lemma.

\begin{restatable}{lemma}{lemufactorwithphantom}
	\label{lem:u-factor-with-phantom}
	If all trees in~$\mathcal{T}'$ were real trees, the tour~$\seq_\alg$ obtained from the concatenating algorithm with input~$\mathcal{T}'$ is a feasible solution for the \psp with approximation guarantee~$\factor$.
\end{restatable}

For the remainder of this section, we denote by $P^*$ the shortest $0$--$n$-path in $H$ computed in Step~\ref{it:u-path} of the concatenating algorithm on input~$\mathcal{T}'$.
Since each vertex $i\in V_H$ corresponds to a unique tree $T_i\in\mathcal{T}'$, we often refer to the vertices on $P^*$ as the chosen trees.
Furthermore, we extend the procedure of Step~\ref{it:u-tour} to an arbitrary $0$--$n$-path $P=(n_0,n_1,\dots,n_{l'})$ in $H$.
More precisely, we introduce a randomized variant of Step~\ref{it:u-tour} that runs as follows.
We start with the empty sequence $\seq_P=(\cdot)$.
For each phase $j=1,\dots,l'$, we construct the tour $\seq(T_{n_j})$ according to \Cref{prop:tree-to-tour}.
Then we pick a traversal direction (forward or backward) of the tour $\seq(T_{n_j})$ uniformly at random and set $\seq_P \gets \seq_P + \seq(T_{n_j})$.

We proceed to analyze the sequence~$\seq_\alg$.
First, observe that $P^*$ contains vertex $n$.
Thus, the corresponding good $n$-tree $T_n$ is traversed and appended to $\seq_\alg$ in Step~\ref{it:u-tour}.
This ensures that the returned tour $\seq_\alg$ is a feasible solution to the \psp with finite total $\alpha$-latency, i.e., $\seq_\alg$ visits all vertices.
It remains to analyze the approximation factor obtained by $\seq_\alg$.
To this end, we first show that for any given $0$--$n$-path $P$ in $H$, the expected total $\alpha$-latency $\E[\lat_{\alpha}(\seq_P)]$ of $\seq_P$ is at most the length $\ell(P)$ of $P$ in $H$.
Next, we use a probabilistic argument to prove the existence of a $0$--$n$-path in $H$ whose length is at most $\factor$ times the total $\alpha$-latency of an optimal tour.
Finally, we explain how these two results imply that the concatenating algorithm on input~$\mathcal{T}'$ indeed computes a tour with an approximation ratio of~$\factor$.

In that direction, we introduce the function~$\psi_P$ such that $\psi_P(k)$ gives an upper bound on the latency of the $k$-th distinct vertex in~$V \setminus \{ s \}$ visited by the tour constructed from path~$P$ in~$H$.
More formally, for a fixed $0$--$n$-path $P=(n_0,n_1,\dots,n_l)$ in $H$, the function $\psi_P: \setn_0 \rightarrow \R_{\geq 0}$ is defined by
\begin{align*}
	\psi_P(k) \coloneqq
	\begin{cases}
		0 & \qquad \text{if } k=0\\
		c(T_{n_{j}})+(1+\alpha)\sum_{i=0}^{j-1}c(T_{n_i}) & \qquad \text{if } n_{j-1}< k \leq n_{j}, \text{for some } j\in \{ 1,\dots,l \}.
	\end{cases}
\end{align*}
This definition is well-defined as $n_0=0$ and $n_l=n$.
We obtain the following lemma whose proof is deferred to \Cref{app:u-expected-latency}.

\begin{restatable}{lemma}{uexpectedlatency}
    \label{lem:u-expected-latency}
    Let $\seq_P$ be the randomized tour obtained from a~$0$--$n$-path~$P$ in~$H$ and let $v\in V$ be the $k$-th distinct vertex of $V\setminus \{ s \}$ in $\seq_P$ for some $k\in \setn_0$.
    Then we have $
    	\E[\lat_{\alpha,v}(\seq_P)]
        \leq \psi_P(k)$.
\end{restatable}

For a fixed $0$--$n$-path $P$ in $H$, \Cref{lem:u-expected-latency} yields
$
	\E[\lat_{\alpha}(\seq_P)]=
	\sum_{v\in V} \E[\lat_{\alpha,v}(\seq_P)]
	\leq \sum_{k=0}^n \psi_P(k)$.
Recall that in Step~\ref{it:u-tour} of the concatenating algorithm, we pick the traversal direction of each subtour such that
$\sum_{v\in V_{j}}\lat_{\alpha,v}(\seq_\alg+\seq_{\min})$ is minimized, where $V_{j}\coloneqq V_{n_j}\big\backslash \big(\bigcup_{i=1}^{j-1} V_{n_i}\big)$.
Hence, for the tour $\seq_\alg$ we obtain the deterministic bound
\begin{align}
	\label{eq:u-expected-latency}
	\lat_{\alpha}(\seq_\alg)
	\leq \sum_{v\in V} \lat_{\alpha,v}(\seq_{P^*})
	\leq \sum_{k=0}^n \psi_{P^*}(k).
\end{align}

Next, we use the following identity, which expresses the length of a path $P$ in $H$ in terms of $\psi_P$; the proof is deferred to \Cref{app:u-path-pi}.

\begin{restatable}{lemma}{upathpi}
    \label{lem:u-path-pi}
    For a $0$--$n$-path $P$ in $H$ we have
    $
    	\sum_{k=0}^n \psi_{P}(k) = \ell(P).
    $
\end{restatable}

The following lemma is the technically most challenging part of the analysis, as it compares the length of a $0$--$n$-path $P$ in $H$ against the total $\alpha$-latency of the optimal solution.
The proof uses similar ideas as in~\cite{chaudhuri2003paths} and~\cite{GriesbachHKS26}.
We consider a randomized sequence of good $k$-trees with exponentially increasing costs where $\gamma \in (1,\infty)$ is a factor determining the rate of the exponential growth.
Then the expected length of that path can be bounded by a constant factor times the optimal total $\alpha$-latency.
A probabilistic argument then yields that this bound is also attained by the shortest $0$--$n$-path $P^*$ in $H$.
Specifically, we obtain the following bound. Its proof can be found in \Cref{app:u-opt-bound}.

\begin{restatable}{lemma}{uoptbound}
   \label{lem:u-opt-bound}
    Let $\seq^* \in \Seq$ be an optimal solution for the \psp.
    Then there exists a $0$--$n$-path $P$ in $H$ such that
    $
    	\ell(P)\leq 2\frac{\gamma+\alpha}{(1+\alpha)\ln\gamma} \lat_{\alpha} (\seq^*)
    $
    for any $\gamma\in (1,\infty)$.
\end{restatable}

\Cref{lem:u-opt-bound} yields an approximation ratio parametrized in $\gamma\in (1,\infty)$.
To obtain the smallest approximation ratio possible, we optimize over $\gamma$.
To this end, we denote by $W : \R_{\geq 0} \to \mathbb{R}_{\geq 0}$ the Lambert-$W$ function that assigns~$x\geq 0$ the unique value~$w\geq 0$ such that~$w\mathrm{e}^w=x$. The proof of the following result is deferred to \Cref{app:min-gamma}.

\begin{restatable}{lemma}{mingamma}
    \label{lem:min-gamma}
    For $\gamma \in (1,\infty)$, let $h$ be defined as
$\smash{h(\gamma) \coloneqq 2 \frac{\gamma + \alpha}{(1+\alpha) \ln(\gamma)}}$.
    Then $h$ is minimized for
    $\smash{\gamma^* = e^{1 + W(\alpha/e)}}$.
    In particular, we have $\smash{\factor=h(\gamma^*)}$.
\end{restatable}

Combining all previous results yields the proof for~\Cref{lem:u-factor-with-phantom}; this proof can be found \Cref{app:lem:u-factor-with-phantom}.

We have shown so far that the concatenating algorithm with input~$\mathcal{T}'$ (the set of trees guaranteed by \Cref{lem:u-good-phantom-trees} with adjusted costs for phantom trees) computes a solution to the \psp with approximation guarantee~$\factor$.
However, we obtained this result only under the assumption that all trees were real trees, i.e., that they could be transformed into tours.
It remains to show that we can obtain the same approximation guarantee when using only real trees as input to the concatenating algorithm.
The issue with phantom trees as input arises only if phantom trees lie on the shortest $0$--$n$-path~$P^*$ computed by the concatenating algorithm, since these trees are the ones traversed to obtain the final tour~$\seq_\alg$.
As a first step towards proving \Cref{thm:u-main}, we thus show that whenever the shortest $0$--$n$-path~$P^*$ on input~$\mathcal{T}'$ contains a phantom tree, there exists another $0$--$n$-path with the same length that visits only real trees.
The proof of the lemma follows the argument of~\cite{archer2008faster,chaudhuri2003paths} and can be found in \Cref{app:u-no-phantom-trees}.

\begin{restatable}{lemma}{unophantomtrees}
\label{lem:u-no-phantom-trees}
	Let $H$ be the auxiliary graph constructed by the concatenating algorithm on input~$\mathcal{T}'$.
	Then there exists a shortest $0$--$n$-path~$P$ in~$H$ that does not visit any vertex whose corresponding tree is a phantom tree.
\end{restatable}

By \Cref{lem:u-no-phantom-trees}, applying the concatenating algorithm to the set~$\mathcal{R}\subseteq\mathcal{T}$ of real trees in $\mathcal{T}$ ensures that the algorithm is well-defined while maintaining an approximation guarantee of~$\factor$.
We can now prove \Cref{thm:u-main}; the proof can be found in \Cref{app:u-main}.

\section{Algorithms for the~$p$-Norm of $\alpha$-Latencies}
\label{sec:p}
We now turn to the case of arbitrary $p\geq 1$.
The case $p=1$ was treated in the previous section, where the linearity of the objective allowed us to use an auxiliary graph whose edge lengths directly represent the additional latency caused by adding a tree.
For $p>1$, this linearity is lost, and the same construction no longer applies directly.

At a high level, we would still like to use sequences of trees of geometrically increasing cost.
However, two difficulties arise.
First, as pointed out by \citet{Garg05}, it is not possible, for every $k \in [n]$, to use the primal--dual algorithm to obtain a $k$-MST with cost no more than the optimum $k$-stroll.
In particular, this issue occurs for those values of $k$ for which the primal--dual procedure does not directly compute a $k$-MST\@.
Second, when $p>1$, the objective is nonlinear in the vertices' $\alpha$-latencies.
Thus, an analogue of the auxiliary-graph cost from the previous section has to keep track not only of how many vertices have already been visited, but also of the accumulated cost of the trees used so far.

We subdivide this section into two parts.
First, we present a randomized algorithm based on a geometric scaling argument.
Then, we show how to derandomize it using a time-expanded auxiliary graph; the resulting running time is pseudo-polynomial in the edge costs.

\subsection{A Randomized Algorithm}
\label{sec:p-randomized}

To avoid the difficulties with the auxiliary-graph construction for the moment, we first give a randomized algorithm with the desired approximation guarantee.

\pmain*

The approximation ratio given in the above theorem is depicted in \Cref{fig:p-norm-heatmap}.
It ranges from $2\mathrm{e}$ for $\alpha = 0$ and $p=1$ to a limiting value of $16$ for $\alpha = 1$ as $p \rightarrow \infty$.
Note that the value $2\mathrm{e}$ matches the bound for $\alpha = 0$ and $p=1$ from the previous section.

Instead of computing a \emph{good} $k$-tree, we use Garg's 2-approximation~\cite{Garg05} to compute a 2-approximate $k$-MST for every $k \in [n]$.
Recall that such a tree contains $s$ and visits $k$ additional vertices.
Let these trees be denoted by $T_1,T_2,\ldots,T_n$.

We next define a randomized sequence of trees, similar to the construction in the previous section.
Fix some $\gamma > 1$, which will be optimized later.
Let $b = \gamma^U$, where $U$ is drawn uniformly at random from $[0,1]$.
For all relevant values of $r$, define
$n_r = \max \{k \in \setn_0 : c(T_k) \leq 2b \gamma^r\}$.
Here, $c(T_k)$ denotes only the total length of the tree and, in particular, does not depend on $p$.

Let $P = (n_0,n_1,\ldots,n_l)$ be the sequence of values obtained by the above procedure, where $n_0=0$.
As in the previous section, we concatenate the sequence of trees induced by $P$ to obtain a tour, which we denote by~$\seq_P$.
We also define a function~$\psi_P$ such that $\psi_P(k)$ gives an upper bound on the $\alpha$-latency of the $k$-th distinct vertex in~$V \setminus \{ s \}$ visited by~$\seq_P$.
Formally, the function $\psi_P: \setn_0 \rightarrow \R_{\geq 0}$ is defined by
\begin{align*}
	\psi_P(k) \coloneqq
	\begin{cases}
		0 & \qquad \text{if } k=0,\\
		(1+\alpha)\sum_{i=0}^{j}c(T_{n_i}) & \qquad \text{if } n_{j-1}< k \leq n_{j}, \text{for some } j\in \{ 1,\dots,l \}.
	\end{cases}
\end{align*}

We obtain the following lemma.

\begin{restatable}{lemma}{uexpectedlatencyp}
    \label{lem:u-expected-latency-p}
    Let $\seq_P$ be the randomized tour obtained from $P$, and let $v\in V$ be the $k$-th distinct vertex of $V\setminus \{ s \}$ in $\seq_P$, for some $k\in \setn_0$.
    Then
    $\lat_{\alpha,v}(\seq_P)
        \leq \psi_P(k)$.
\end{restatable}
\begin{proof}
    This follows directly from the definition of $\psi_P$.
    In the worst case, all trees are nested and vertex $k$ is visited at the very end of the last tree containing it.
    Moreover, the cost of the tour induced by tree $T_{n_i}$ is bounded by $(1+\alpha)c(T_{n_i})$.
\end{proof}

\Cref{lem:u-expected-latency-p} and the randomized sequence of trees imply \Cref{thm:u-main-p}.

\begin{proof}[Proof of \Cref{thm:u-main-p}.]

Let $\seq^*$ be an optimal tour for the $p$-norm $\alpha$-latency objective.
Let $v_i^*$ be the $i$-th vertex visited by $\seq^*$.
Then
\begin{align*}
\mathrm{OPT}
=
\left(
\sum_{i \in [n]}
\bigl(\lat_{\alpha, v_i^*}(\seq^*)\bigr)^p
\right)^{1/p}.
\end{align*}

Fix some arbitrary $i \in [n]$.
Let $r$ be such that
$\lat_{\alpha, v_i^*}(\seq^*) = d\gamma^r$
for some $d \in [1,\gamma)$.
We consider two cases.

\paragraph{First case: $d \leq b$.}
In this case, there is a tree $T_{n_r}$ with cost at most $2b\gamma^r$ and $n_r \geq i$.
Hence,
\begin{align*}
\psi_P(i)
\leq
(1+\alpha) \sum_{k=1}^{r} c(T_{n_k})
\leq
(1+\alpha) \sum_{k=1}^{r} 2b\gamma^k
<
(1+\alpha) \sum_{k=-\infty}^{r} 2b\gamma^k
=
(1+\alpha) 2b \gamma^r \frac{\gamma}{\gamma - 1}.
\end{align*}

\paragraph{Second case: $d > b$.}
In this case, the same argument applied to the next threshold gives
\begin{align*}
\psi_P(i)
\leq
(1+\alpha) \sum_{k=1}^{r+1} c(T_{n_k})
\leq
(1+\alpha) \sum_{k=1}^{r+1} 2b\gamma^k
<
(1+\alpha) \sum_{k=-\infty}^{r+1} 2b\gamma^k
=
(1+\alpha) 2b \gamma^{r+1} \frac{\gamma}{\gamma - 1}.
\end{align*}

Since $U$ is chosen uniformly at random in $[0,1]$ and $b=\gamma^U$, we can average over the two cases and obtain
\begin{align*}
\E_U[\psi_P(i)^p]
\leq{}&
\int_{\log_\gamma d}^1
\biggl( (1+\alpha) 2b\gamma^r \frac{\gamma}{\gamma -1} \biggr)^{p}
\, \mathrm{d} U
+
\int_{0}^{\log_\gamma d}
\biggl((1+\alpha)2b\gamma^{r+1} \frac{\gamma}{\gamma - 1} \biggr)^{p}
\, \mathrm{d}U.
\end{align*}

We now compute the right-hand side.
Let $a \coloneqq \log_\gamma d$.
Since $b=\gamma^U$, we have
\begin{align*}
&\int_{\log_\gamma d}^1
\biggl( (1+\alpha) 2b\gamma^r \frac{\gamma}{\gamma -1} \biggr)^{p}
\, \mathrm{d} U
+
\int_{0}^{\log_\gamma d}
\biggl((1+\alpha)2b\gamma^{r+1} \frac{\gamma}{\gamma - 1} \biggr)^{p}
\, \mathrm{d}U
\\
&=
\biggl((1+\alpha)2\frac{\gamma}{\gamma-1}\biggr)^p
\left[
\gamma^{rp}\int_a^1 \gamma^{pU}\,\mathrm{d}U
+
\gamma^{(r+1)p}\int_0^a \gamma^{pU}\,\mathrm{d}U
\right].
\end{align*}
Using
\[
\int \gamma^{pU}\,\mathrm{d}U
=
\frac{\gamma^{pU}}{p\ln\gamma},
\]
we obtain
\begin{align*}
&\biggl((1+\alpha)2\frac{\gamma}{\gamma-1}\biggr)^p
\left[
\gamma^{rp}\int_a^1 \gamma^{pU}\,\mathrm{d}U
+
\gamma^{(r+1)p}\int_0^a \gamma^{pU}\,\mathrm{d}U
\right]
\\
&=
\biggl((1+\alpha)2\frac{\gamma}{\gamma-1}\biggr)^p
\frac{
\gamma^{rp}(\gamma^p-\gamma^{pa})
+
\gamma^{(r+1)p}(\gamma^{pa}-1)
}{p\ln\gamma}.
\end{align*}
Since $a=\log_\gamma d$, we have $\gamma^{pa}=d^p$.
Hence,
\begin{align*}
\E_U[\psi_P(i)^p]
&\leq
\biggl((1+\alpha)2\frac{\gamma}{\gamma-1}\biggr)^p
\frac{
\gamma^{rp}(\gamma^p-d^p)
+
\gamma^{(r+1)p}(d^p-1)
}{p\ln\gamma}
\\
&=
\biggl((1+\alpha)2\frac{\gamma}{\gamma-1}\biggr)^p
\frac{
d^p\gamma^{rp}(\gamma^p-1)
}{p\ln\gamma}
\\
&=
(1+\alpha)^p
\frac{2^p}{p\ln\gamma}
\left[
\left(\frac{d\gamma^{r+2}}{\gamma-1}\right)^p
-
\left(\frac{d\gamma^{r+1}}{\gamma-1}\right)^p
\right].
\end{align*}

The contribution of $v_i^*$ to the optimum is
\[
\bigl(\lat_{\alpha,v_i^*}(\seq^*)\bigr)^p
=
(d\gamma^r)^p.
\]
Therefore, the previous bound implies

\begin{align}
\label{eq:p-factor}
\frac{\E_U[\psi_P(i)^p]}{(d\gamma^r)^p}
&\leq
(1+\alpha)^p
\frac{2^p}{p\ln\gamma}
\left[
\left(\frac{\gamma^2}{\gamma-1}\right)^p
-
\left(\frac{\gamma}{\gamma-1}\right)^p
\right]
=
(1+\alpha)^p
\frac{2^p}{p\ln\gamma}
\cdot
\frac{\gamma^{2p}-\gamma^p}{(\gamma-1)^p}
=: \eta.
\end{align}

In particular, since $\lat_{\alpha,v_i^*}(\seq^*)=d\gamma^r$, we have
\begin{align}
\label{eq:p-expected-single}
\E_U[\psi_P(i)^p]
\leq
\eta \cdot
\bigl(\lat_{\alpha,v_i^*}(\seq^*)\bigr)^p.
\end{align}

We now use this bound to compare the expected $p$-norm $\alpha$-latency of the randomized tour with the optimum.
By the definition of $\psi_P$, for every realization of $U$ we have
\[
\lat_{\alpha,v_i^*}(\seq_P)
\leq
\psi_P(i)
\qquad \text{for all } i\in[n].
\]
Hence,
\begin{align*}
\E_U\left[
\left(
\sum_{i\in[n]}
\lat_{\alpha,v_i^*}(\seq_P)^p
\right)^{1/p}
\right]
&\leq
\E_U\left[
\left(
\sum_{i\in[n]}
\psi_P(i)^p
\right)^{1/p}
\right].
\end{align*}
We next use Jensen's inequality.
Recall that if $f$ is concave and $X$ is a nonnegative random variable, then
\[
\E[f(X)] \leq f(\E[X]).
\]
In our setting, we apply this with
\[
f(x)=x^{1/p}
\qquad\text{and}\qquad
X=\sum_{i\in[n]}\psi_P(i)^p.
\]
Since $p\geq 1$, the function $f(x)=x^{1/p}$ is concave on $\R_{\geq 0}$.
Moreover, $X\geq 0$ for every realization of the random choice of $U$.
Therefore, Jensen's inequality gives
\begin{align*}
\E_U\left[
\left(
\sum_{i\in[n]}
\psi_P(i)^p
\right)^{1/p}
\right]
\leq
\left(
\E_U\left[
\sum_{i\in[n]}
\psi_P(i)^p
\right]
\right)^{1/p}
=
\left(
\sum_{i\in[n]}
\E_U[\psi_P(i)^p]
\right)^{1/p}.
\end{align*}

Using~\eqref{eq:p-expected-single}, we obtain
\begin{align*}
\E_U\left[
\left(
\sum_{i\in[n]}
\lat_{\alpha,v_i^*}(\seq_P)^p
\right)^{1/p}
\right]
&\leq
\left(
\sum_{i\in[n]}
\eta \cdot
\bigl(\lat_{\alpha,v_i^*}(\seq^*)\bigr)^p
\right)^{1/p}
\\
&=
\eta^{1/p}
\left(
\sum_{i\in[n]}
\bigl(\lat_{\alpha,v_i^*}(\seq^*)\bigr)^p
\right)^{1/p}
\\
&=
\eta^{1/p}\cdot \mathrm{OPT}.
\end{align*}

It remains to simplify $\eta^{1/p}$. By~\eqref{eq:p-factor},
\begin{align*}
\eta^{1/p}
&=
(1+\alpha)
\left[
\frac{2^p}{p\ln\gamma}
\cdot
\frac{\gamma^{2p}-\gamma^p}{(\gamma-1)^p}
\right]^{1/p}\\
&=
(1+\alpha)
\frac{2}{(p\ln\gamma)^{1/p}}
\cdot
\frac{(\gamma^{2p}-\gamma^p)^{1/p}}{\gamma-1}\\
&=
(1+\alpha)
\frac{2}{\gamma-1}
\left[
\frac{\gamma^{2p}-\gamma^p}{p\ln\gamma}
\right]^{1/p}.
\end{align*}

Thus, for every fixed $\gamma>1$, the randomized algorithm has expected approximation factor
\[
(1+\alpha)
\frac{2}{\gamma-1}
\left[
\frac{\gamma^{2p}-\gamma^p}{p\ln\gamma}
\right]^{1/p}.
\]
Optimizing over $\gamma>1$ gives the claimed bound.
\end{proof}

\subsection{Derandomization via a Time-Expanded Auxiliary Graph}
\label{ssec:p-derandomization}

We now describe how to derandomize the randomized algorithm from the previous subsection.
The idea is similar to the derandomization used for the case $p=1$: we construct an auxiliary graph and compare the length of a shortest path in this graph to the expected length of a suitable randomized path.

For $p=1$, the auxiliary graph only needs to keep track of how many vertices have already been visited.
For $p>1$, however, the objective is nonlinear in the latencies.
Thus, the additional cost of traversing a tree depends not only on the number of vertices visited so far, but also on the latency accumulated so far.
We therefore use a two-dimensional auxiliary graph.
One dimension keeps track of the number of vertices that have already been visited, and the other dimension keeps track of the total cost of the trees traversed so far.

\pderandomized*

Let $T_1,\dots,T_n$ be the 2-approximate $k$-MST trees used in the previous subsection, and let $c(T_k)$ denote the cost of tree $T_k$.
We set $T_0$ to be the trivial tree containing only $s$, with $c(T_0)=0$.
Let
\[
Z
\coloneqq
\left\{
\sum_{k\in S} c(T_k) : S\subseteq \setn
\right\}.
\]
Thus, $Z$ is the set of all values that can occur as the accumulated cost of a set of traversed trees.

We construct a directed auxiliary graph
$H=(V_H,A_H)$
as follows:
The vertex set is
$V_H = \bigl(\setn_0 \times Z\bigr) \cup \{\omega\}$,
where $\omega$ is an additional terminal vertex.
A vertex $(i,z)$ represents the state in which $i$ vertices of $V\setminus\{s\}$ have already been visited, and the total cost of the trees traversed so far is $z$.
For every $i,j\in \setn_0$ with $i<j$ and every $z,z'\in Z$ with
$z'=z+c(T_j)$,
we add the directed edge
$\bigl((i,z),(j,z')\bigr)$.
The length of this edge is defined as
\begin{align}
\label{eq:p-aux-edge-length}
\ell_{(i,z),(j,z')}
\coloneqq
(n-i)
\left(
\bigl((1+\alpha)z'\bigr)^p
-
\bigl((1+\alpha)z\bigr)^p
\right).
\end{align}
Finally, for every $z\in Z$, we add an edge $((n,z),\omega)$ of length zero.
The algorithm computes a shortest path from $(0,0)$ to $\omega$ in $H$.
Let
$P^*
=
\bigl((n_0,z_0),(n_1,z_1),\dots,(n_l,z_l),\omega\bigr)$
be such a shortest path, where $n_0=0$, $z_0=0$, and $n_l=n$.
The algorithm then concatenates the trees
$T_{n_1},T_{n_2},\dots,T_{n_l}$
in this order, using the tree-to-tour conversion from \Cref{prop:tree-to-tour}.
We denote the resulting tour by $\seq_\alg$.
We first relate the length of a path in $H$ to the $p$-norm $\alpha$-latency of the corresponding concatenated tour.
\begin{restatable}{lemma}{lemppathlengthbound}
\label{lem:p-path-length-bounds-tour}
Let
$P = \bigl((n_0,z_0),(n_1,z_1),\dots,(n_l,z_l),\omega\bigr)$
be any $(0,0)$--$\omega$ path in $H$, where $n_0=0$, $z_0=0$, and $n_l=n$.
Let $\seq_P$ be the tour obtained by concatenating the trees
$T_{n_1},\dots,T_{n_l}$.
Then
$\smash{\bigl(\sum_{v\in V}
\lat_{\alpha,v}(\seq_P)^p\bigr)^{1/p}
\leq
\ell(P)^{1/p}}$.
\end{restatable}
\begin{proof}
For $h\in\{1,\dots,l\}$, the transition from $(n_{h-1},z_{h-1})$ to $(n_h,z_h)$ corresponds to traversing tree $T_{n_h}$, and by construction
\[
z_h=z_{h-1}+c(T_{n_h}).
\]
After this traversal, the accumulated cost of the traversed trees is $z_h$.
By the tree-to-tour conversion, the $\alpha$-latency of every vertex that is first covered by $T_{n_h}$ is at most
\[
(1+\alpha)z_h.
\]
Hence, if $n_{h-1}<k\leq n_h$, then the $k$-th distinct vertex visited by $\seq_P$ has $\alpha$-latency at most $(1+\alpha)z_h$.

Define
\[
\psi_P(k)
\coloneqq
(1+\alpha)z_h
\qquad
\text{if } n_{h-1}<k\leq n_h.
\]
Then
\[
\lat_{\alpha,v_k}(\seq_P)\leq \psi_P(k)
\]
for the $k$-th distinct vertex $v_k$ of $V\setminus\{s\}$ visited by $\seq_P$.
Therefore,
\[
\sum_{v\in V}\lat_{\alpha,v}(\seq_P)^p
\leq
\sum_{k=1}^n \psi_P(k)^p.
\]

It remains to observe that the right-hand side is exactly the length of $P$.
Indeed, using the definition of the edge lengths in~\eqref{eq:p-aux-edge-length}, we get
\begin{align*}
\ell(P)
&=
\sum_{h=1}^l
(n-n_{h-1})
\left(
\bigl((1+\alpha)z_h\bigr)^p
-
\bigl((1+\alpha)z_{h-1}\bigr)^p
\right)
\\
&=
\sum_{h=1}^l
(n_h-n_{h-1})
\bigl((1+\alpha)z_h\bigr)^p
\\
&=
\sum_{k=1}^n \psi_P(k)^p.
\end{align*}
The second equality follows by telescoping.
Thus,
\[
\sum_{v\in V}\lat_{\alpha,v}(\seq_P)^p
\leq
\ell(P),
\]
which proves the claim.
\end{proof}

We next show that the auxiliary graph contains a path whose length is bounded by the same expression that appeared in the randomized analysis.
For fixed $\gamma>1$, define
\[
\eta_\gamma
\coloneqq
(1+\alpha)^p
\frac{2^p}{p\ln\gamma}
\cdot
\frac{\gamma^{2p}-\gamma^p}{(\gamma-1)^p}.
\]

\begin{restatable}{lemma}{lempexistsshortpath}
\label{lem:p-exists-short-path}
Let $\seq^*$ be an optimal tour for the $p$-norm $\alpha$-latency objective.
Then, for every fixed $\gamma>1$, there exists a $(0,0)$--$\omega$ path $P$ in $H$ such that
\[
\ell(P)
\leq
\eta_\gamma
\sum_{i\in[n]}
\bigl(\lat_{\alpha,v_i^*}(\seq^*)\bigr)^p.
\]
\end{restatable}
\begin{proof}
We use the randomized construction from the previous subsection.
For a random choice of $U\in[0,1]$, let $P(U)$ be the corresponding sequence of trees.
After removing consecutive repetitions, this sequence defines a path in the auxiliary graph $H$.
Indeed, the first coordinate records the number of vertices covered, while the second coordinate records the accumulated cost of the selected trees.

By the definition of the auxiliary graph and the calculation in the proof of \Cref{lem:p-path-length-bounds-tour}, we have
\[
\ell(P(U))
=
\sum_{i\in[n]} \psi_{P(U)}(i)^p.
\]
From the randomized analysis, for every $i\in[n]$ we have
\[
\E_U[\psi_{P(U)}(i)^p]
\leq
\eta_\gamma
\bigl(\lat_{\alpha,v_i^*}(\seq^*)\bigr)^p.
\]
Taking the sum over all $i\in[n]$ and using linearity of expectation gives
\begin{align*}
\E_U[\ell(P(U))] =
\E_U\left[
\sum_{i\in[n]} \psi_{P(U)}(i)^p
\right]
=
\sum_{i\in[n]} \E_U[\psi_{P(U)}(i)^p]
\leq
\eta_\gamma
\sum_{i\in[n]}
\bigl(\lat_{\alpha,v_i^*}(\seq^*)\bigr)^p.
\end{align*}
Therefore, there exists at least one realization of $U$ such that
\[
\ell(P(U))
\leq
\eta_\gamma
\sum_{i\in[n]}
\bigl(\lat_{\alpha,v_i^*}(\seq^*)\bigr)^p.
\]
This proves the lemma.
\end{proof}

Using the above lemmas and analyzing the size of the time-expanded auxiliary graph, we prove \Cref{thm:p-derandomized}.

\begin{proof}[Proof of \Cref{thm:p-derandomized}.]
Let $P^*$ be a shortest $(0,0)$--$\omega$ path in $H$, and let $\seq_\alg$ be the tour obtained by concatenating the trees along $P^*$.
By \Cref{lem:p-path-length-bounds-tour},
\[
\left(
\sum_{v\in V}
\lat_{\alpha,v}(\seq_\alg)^p
\right)^{1/p}
\leq
\ell(P^*)^{1/p}.
\]
Since $P^*$ is a shortest path, for every $\gamma>1$ and for the path $P$ whose existence is guaranteed by \Cref{lem:p-exists-short-path}, we have
\[
\ell(P^*)\leq \ell(P).
\]
Thus,
\[
\ell(P^*)
\leq
\eta_\gamma
\sum_{i\in[n]}
\bigl(\lat_{\alpha,v_i^*}(\seq^*)\bigr)^p
=
\eta_\gamma \cdot \mathrm{OPT}^p.
\]
Taking $p$-th roots gives
\[
\left(
\sum_{v\in V}
\lat_{\alpha,v}(\seq_\alg)^p
\right)^{1/p}
\leq
\eta_\gamma^{1/p}\cdot \mathrm{OPT}.
\]
Finally,
\[
\eta_\gamma^{1/p}
=
(1+\alpha)
\frac{2}{\gamma-1}
\left[
\frac{\gamma^{2p}-\gamma^p}{p\ln\gamma}
\right]^{1/p}.
\]
Since this holds for every $\gamma>1$, we may minimize over $\gamma>1$, which gives the claimed approximation factor.

It remains to discuss the running time.
Let
\[
N_H \coloneqq |V_H|+|A_H|
\]
denote the size of the auxiliary graph.
All edge lengths in $H$ are nonnegative, and hence a shortest $(0,0)$--$\omega$ path can be computed with Dijkstra's algorithm in time
\[
O(N_H\log N_H).
\]
In general, the set $Z$ may be exponentially large, since it consists of all subset sums of the tree costs.
Thus, the algorithm is pseudo-polynomial.
If the tree costs are integral and polynomially bounded, then $|Z|$ is polynomially bounded as well.
In this case, $N_H$ is polynomially bounded, and the algorithm runs in polynomial time.
\end{proof}

\section{Discussion}

We studied the discounted graph search problem, a two-parameter framework that contains pathwise search, expanding search, TSP-type objectives, and MST-type objectives as special or limiting cases.
The first parameter, $\alpha\in[0,1]$, determines how repeated traversals of edges are charged, while the second parameter, $p\geq 1$, determines how the individual vertex latencies are aggregated.
For $p=1$, we obtain a polynomial-time approximation algorithm whose guarantee interpolates between the currently best-known approximation ratios for expanding search and pathwise search.
For general $p\geq 1$, we obtain a randomized constant-factor approximation and a deterministic pseudo-polynomial-time algorithm with the same approximation guarantee.

We stated and analyzed the problem for a uniform discount factor $\alpha$.
The results for~$p=1$ also extend to the case where every edge $e$ has its own discount factor $\alpha_e\in[0,1]$.
In this setting, the first traversal of edge $e$ costs $c_e$, while every further traversal costs $\alpha_e c_e$.
Let
$\alphamin \coloneqq \min_{e\in E}\alpha_e$ and $\alphamax \coloneqq \max_{e\in E}\alpha_e$.
Then the analysis yields the approximation guarantee
\[
	\factor \coloneqq
 \begin{cases}
2\mathrm{e} & \text{if } \alphamax = \alphamin = 0,\\[2mm]
2 \dfrac{\alphamax}{(1+\alphamin)W(\alphamax/\mathrm{e})} & \text{otherwise.}
\end{cases}
\]
The reason is that the analysis uses the largest possible discount factor when comparing the cost of trees to paths, while the conversion of trees into tours benefits from the smallest discount factor.
Thus, the same proof goes through with $\alpha$ replaced by $\alphamax$ in the former part of the analysis and by $\alphamin$ in the latter.

A similar generalization is possible when the cost of an edge depends on how often it has already been used.
Suppose that for each edge~$e$ we are given a sequence
$\smash{c_e^{(1)}, c_e^{(2)}, c_e^{(3)}, \dots}$
such that the $k$-th traversal of edge~$e$ costs $\smash{c_e^{(k)}}$.
If these sequences are non-increasing, i.e.,
$\smash{c_e^{(k)} \geq c_e^{(k+1)}}$ for all $e\in E$ and $k\in\mathbb{N}$,
then our analysis can again be extended.
Define
$\smash{\alphamin \coloneqq
\min_{e\in E}
c_e^{(2)} / c_e^{(1)}}$ and
$\smash{\alphamax \coloneqq
\max_{e\in E}
c_e^{(2)} / c_e^{(1)}}$.
Then the same approximation guarantee as above follows.
Indeed, in the comparison to the optimum, it is sufficient to consider paths in which each edge is used at most twice.
On the other hand, the algorithm may use edges more often, and for non-increasing traversal costs, every traversal after the second one can only become cheaper.
Thus, replacing all costs $\smash{c_e^{(k)}}$ for $k\geq 2$ by $\smash{c_e^{(2)}}$ can only make the algorithm more expensive, while preserving the relevant comparison to the optimum.

By contrast, for arbitrary non-decreasing traversal-cost sequences, one cannot hope for a con\-stant-factor approximation in full generality.
For example, suppose that
$\smash{c_e^{(1)}=0}$ and $\smash{c_e^{(k)}=1}$ for all $k\geq 2$
for every edge $e\in E$.
Then there is a traversal with total latency zero if and only if the graph has a Hamiltonian path starting at $s$.
Otherwise, every feasible traversal must repeat some edge and hence has positive latency.
Since deciding the existence of such a Hamiltonian path is $\mathsf{NP}$-complete, no constant-factor approximation can exist in this setting unless $\mathsf{P}=\mathsf{NP}$.
It would be interesting to understand which structural assumptions on the sequences $c_e^{(1)}, c_e^{(2)}, c_e^{(3)}, \dots$ still allow constant-factor approximations.

Finally, for the case $p>1$, one can also compare our bounds to algorithms for the limiting case $p=\infty$.
For any nonnegative latency vector $x\in\mathbb{R}_{\geq 0}^n$, we have
$\smash{\|x\|_\infty
\leq
\|x\|_p
\leq
n^{1/p}\|x\|_\infty}$.
Thus, any $\beta$-approximation for the maximum-latency objective immediately gives a
$\beta n^{1/p}$-
approximation for the $p$-norm objective.
At the endpoint $\alpha=1$, the maximum-latency objective corresponds to a TSP-path-type problem, while at the endpoint $\alpha=0$ it corresponds to the MST problem.
Therefore, for large values of $p$, one can use approximation algorithms for these limiting problems to obtain alternative approximation guarantees.
In particular, once $p = \Omega(\log n)$, the factor $n^{1/p}$ becomes constant.
For even larger values of $p$, this factor approaches~$1$, and the guarantee approaches that of the corresponding maximum-latency approximation algorithm.
This gives better bounds for sufficiently large $p$ in regimes where the limiting TSP- or MST-type objective can be approximated more accurately than the general bound obtained from the randomized geometric construction.

\bibliography{alpha}

\newpage
\appendix

\section{Deferred Proofs of \Cref{sec:prelim}}

\subsection{Proof of Proposition~\ref{prop:tree-to-tour}}
\label{app:tree-to-tour}
\treetotour*

\begin{proof}
	Let $G_T\subseteq G$ be the subgraph of $G$ that only contains the edges in $T$.
	By definition of a tree, $G_T$ contains vertex $s$.
	For each edge $e\in T$, we add a copy of $e$ with \edgec $\alpha c_e$ and denote the obtained graph by $G_T'$.
	In $G_T'$ every vertex has an even degree and thus, $G_T'$ contains an Euler tour $\seq(T)$ starting in $s$ that visits all vertices of $T$ and has $\alpha$-\edgec $
		c_{\alpha}(\seq(T))=\sum_{e\in G_T} (1+\alpha)c_e = (1+\alpha)c(T)$.
\end{proof}

\section{Computation of good $k$-trees}
\label{app:finding-good-trees}

\newparagraph{The Primal--Dual Subroutine}
The primal--dual subroutine as introduced by \cite{chaudhuri2003paths} and its analysis are based on the primal--dual algorithm for the prize-collecting Steiner tree problem by \citet{goemans1995general} and further inspired by algorithms for the $k$-MST problem used by \citet{blum1996constant} and \citet{garg1996approximation}.
The input consists of a graph $G=(V,E)$, a start vertex $s\in V$, a terminal vertex $t\in V \setminus \{s\}$, and a parameter $\lambda\geq 0$.
Thereupon, the algorithm returns a tree that spans $k_{\lambda,t}$ many vertices, including $s$ and $t$, which is attained as follows.

At the beginning, the initial solution $F$ is given by the empty edge set $F=\emptyset$.
Each vertex $v\in V$ is equipped with a budget $b_v\geq 0$, where $b_s=0, b_t=\infty$, and $b_v=\lambda$ for all $v\in V\setminus \{s,t\}$.
Furthermore, for each vertex set $S\subseteq V$, there exists a dual variable $y_S$ which is initialized with value $0$.
The main part of the algorithm consists of two phases: the \emph{growth phase} and the \emph{delete phase}.

Throughout the growth phase, connected components are divided into \emph{active} and \emph{inactive} sets.
At the beginning, each vertex~$v$ with $b_v>0$ is an active connected component, and only vertex~$s$ is inactive.
Then, the dual variables $y_S$ of all active components $S\subseteq V$ grow continuously at the same speed until one of the following two cases occurs.
Either an edge becomes tight, i.e., $\sum_{S\subset V:e\in \delta(S)}y_S=c_e$, which is referred to as an \emph{edge event};
or the total budget of an active component is depleted, i.e., $\sum_{v\in S}b_v=\sum_{T\subseteq S}y_T$, which is referred to as a \emph{set event}.
In the case of an edge event, the tight edge $e$ is added to the current solution $F$, and the two components $S_1$ and $S_2$ that are linked by $e$ are merged, resulting in a new component $S=S_1\cup S_2$.
If $s\in S$, the new component is inactive, otherwise it is active. 
The previous components $S_1$ and $S_2$ become inactive and hence, their corresponding dual variables are not further increased.
In the case of a set event, the component becomes inactive.
The growth phase terminates when all components are inactive.
Note that at any point in time during the growth phase, the current solution $F$ induces a forest in $G$.
More precisely, if two edges become tight at the same time, the second edge is only added to the set~$F$ if it does not close a cycle in~$F$.
The order in which two simultaneous events are carried out is chosen arbitrarily.
Since vertex $t$ has a budget of $b_t=\infty$, the connected component containing $t$ only becomes inactive when it merges with the component that contains $s$.
Hence, when the growth phase terminates, there exists a tree $T$ in $F$ that contains both vertices $s$ and $t$.
After the growth phase, the delete phase starts.
In this phase, the algorithm iterates over the edges in $T$ and deletes edge $e\in T$ under the following condition:
Let $T_s$ and $T_{-s}$ be the two components obtained from $T\setminus {e}$, such that $s\in T_s$.
If $T_{-s}$ was once inactive during the growth phase, edge $e$ is deleted from $T$, otherwise it is kept.
Since vertex~$t$ was always in an active component until merging with~$s$, the tree $T$ still contains vertices $s$ and $t$ at the end of the growth phase.
Furthermore, after the delete phase, no leaf of tree~$T$ was in an inactive component during the growth phase.
The tree~$T$ is then returned by the algorithm.
\citeauthor{chaudhuri2003paths} argue that this primal--dual subroutine, as a variant of already existing algorithms for the prize-collecting Steiner tree problem, can be implemented in time $\mathcal{O}(n^2)$.

\newparagraph{The Primal--Dual Algorithm}
In contrast to the primal--dual subroutine, the \emph{primal--dual algorithm} does not require a distinct terminal vertex $t$ as input.
In fact, for a given parameter~$\lambda$, the algorithm simply runs the primal--dual subroutine with parameter $\lambda$ and all possible choices of the terminal vertex $t\in V \setminus \{s\}$.
Each subroutine returns a tree~$T_{{\lambda,t}}$ out of which the algorithm returns the one with the lowest \edgec.
Thus, the final output is a tree $T_{\lambda}$, that contains vertex $s$ and $k_\lambda$ many vertices of~$V \setminus \{s\}$.

\newparagraph{Analyzing the Primal--Dual Algorithm}
Our goal is to show that the tree $T_{\lambda}$ returned by the primal--dual algorithm is a good $k_\lambda$-tree for any choice of $\alpha \in [0,1]$.
For the proof, we consider the returned tree~$T_{{\lambda,t}}$ of the primal--dual subroutine for an arbitrary but fixed terminal vertex $t\in V \setminus \{s\}$ and compare it to the $\alpha$-\edgec of an optimal $k_{\lambda,t}$-path that ends in~$t$.
As the primal--dual algorithm chooses the cheapest tree~$T_{{\lambda,t}}$ over all possible choices of $t$, the desired result then follows.

In that direction, we introduce a relaxation $\mathrm{P}_{\alpha,k,t}$ 
of the linear program that computes an optimal $k$-path for $\alpha$ that starts in $s$ and ends in a designated vertex $t\in V \setminus \{s\}$.
Each vertex $v\in V\setminus \{s,t\}$ is associated with a variable $x_v$ equal to $1$ if $v$ is visited by the path, and $0$ otherwise.
Each edge $e\in E$ is associated with two variables $x_e^1$ and $x_e^2$.
The variable $x_e^1$ equals $1$ if $e$ is used at least once in the $s$--$t$-path and $0$ otherwise, while the variable $x_e^2$ equals $1$ if $e$ is used at least twice in the $s$--$t$-path and $0$ otherwise.
The primal linear program $\mathrm{P}_{\alpha,k,t}$ is formally given by
\begin{align*}
\begin{array}{lrcll}
	\text{Min.}
	&\multicolumn{4}{l}{\sum_{e \in E} c_e x_e^1 +\alpha c_e x_e^2}\\
	\text{s.t.}
	&\sum _{e\in \delta(S)}x_e^1 + x_e^2 &\geq &2x_v
	& \text{ for all } S\subseteq V\setminus\{s,t\}, \text{ for all } v\in S\\
	&\sum_{e\in\delta(U)}x_e^1 + x_e^2  &\geq& 1
	& \text{ for all }U\subseteq V :t\in U, s\notin U \\
	&\sum_{v\in V\setminus\{s,t\}}x_v &\geq & k-1
	& \\
	&x_v &\leq &1 
	& \text{ for all } v\in V\setminus\{s,t\}\\
	&x_v &\geq &0 
	& \text{ for all } v\in V\setminus\{s,t\}\\
	&x_e^1 -x_e^2 &\geq &0
	&\text{ for all } e \in E \\
	&x_e^2 &\geq &0
	&\text{ for all } e \in E. 
\end{array}
\end{align*}
Let $P_{k,t}$ be a $k$-path ending in $t$.
Without loss of generality, $P_{k,t}$ does not use any edge more than twice.
Otherwise, we can shortcut $P_{k,t}$ to obtain a $k$-path $P_{k,t}'$ ending in $t$ that visits the same vertices as $P_{k,t}$ and $c_{\alpha}(P_{k,t}') \leq c_{\alpha}(P_{k,t})$.
The path $P_{k,t}$ provides a feasible solution to $\mathrm{P}_{\alpha,k,t}$ with an objective function value equal to its $\alpha$-\edgec $c_{\alpha}(P_{k,t})$.
Hence, the objective function value $\mathrm{P}_{\alpha,k,t}^*$ of an optimal solution to $\mathrm{P}_{\alpha,k,t}$ yields a lower bound for the $\alpha$-\edgec of any $k$-path ending in $t$, i.e.,
\begin{align}
    \label{eq:u-primal-opt-path}
	\mathrm{P}_{\alpha,k,t}^*\leq c_{\alpha}(P_{k,t})
\end{align}
for all $k\in \setn$ and all $k$-paths $P_{k,t}$ ending in $t$.

\citeauthor{chaudhuri2003paths} use a similar linear program to prove that the same primal--dual algorithm returns good trees for the special case of $\alpha=1$.
We refer to their primal linear program by $\mathrm{P}_{1,k,t}$, which is a relaxation of the linear program that computes an optimal $k$-path for $\alpha=1$ that starts in $s$ and ends in $t$.
Thus, any $k$-path from $s$ to $t$ yields a feasible solution to $\mathrm{P}_{1,k,t}$.
In contrast to $\mathrm{P}_{\alpha,k,t}$, in $\mathrm{P}_{1,k,t}$ each edge $e$ is only equipped with one variable $x_e$ indicating how often $e$ is used in the path.
With $\alpha=1$, each traversal of an edge contributes the same to the $\alpha$-\edgec of a path and thus, it is not necessary to distinguish between first and further traversals.
Accordingly, the objective function of $\mathrm{P}_{1,k,t}$ is slightly simpler and given by Min. $\sum_{e\in E}c_e x_e$.
We proceed to analyze and compare optimal solutions of $\mathrm{P}_{1,k,t}$ and $\mathrm{P}_{\alpha,k,t}$.
To this end, let $\left(( x_e^*)_{e\in E},( x_v^*)_{v\in V}\right)\in \R^{|E|+|V|}$ be an optimal solution for $\mathrm{P}_{1,k,t}$.
We define $\left((\hat{ x}_e^1)_{e\in E},(\hat{ x}_e^2)_{e\in E},(\hat{ x}_v)_{v\in V})\right)\in \R^{2|E|+|V|}$ by
\begin{align}
	\hat{x}_e^1  &\coloneqq \lfrac{x_e^*}{2} \text{ for all } e\in E, &
	\hat{x}_e^2  &\coloneqq \lfrac{x_e^*}{2} \text{ for all } e\in E, &&\text{ and } &
	\hat{x}_v  &\coloneqq x_v \text{ for all } v\in V.\label{eq:u-def-solution}
\end{align}
It is easy to verify that this construction yields a feasible solution to $\mathrm{P}_{\alpha,k,t}$.
In the following, we argue that this solution is also optimal.
For better readability, we write $\vec x_e$ for $( x_e)_{e\in E}$ and $\vec x_v$ for $( x_v)_{v\in V}$.

\begin{restatable}{lemma}{uoptimalsolution}
\label{lem:u-optimal-solution}
	For $k\in \N$ and a fixed vertex $t\in V \setminus \{s\}$, let $(\vec x_e^*,\vec x_v^*)\in \R^{|E|+|V|}$ be an optimal solution for the linear program $\mathrm{P}_{1,k,t}$.
	Then $(\hat{\vec x}_e^1,\hat{\vec x}_e^2,\hat{\vec x}_v)\in \R^{2|E|+|V|}$ defined by \eqref{eq:u-def-solution} yields an optimal solution for the linear program $\mathrm{P}_{\alpha,k,t}$.
\end{restatable}

\begin{proof}
    We prove the statement via contradiction.
    To this end, assume $(\vec x_e^*,\vec x_v^*)\in \R^{|E|+|V|}$ is an optimal solution for $\mathrm{P}_{\vec{1},k,t}$, but $(\hat{\ x}_e^1,\hat{\vec x}_e^2,\hat{\vec x}_v)\in \R^{2|E|+|V|}$ defined by \eqref{eq:u-def-solution} is not optimal for $\mathrm{P}_{{\alpha},k,t}$.
    Instead, let $(\tilde{\vec x}_e^1,\tilde{\vec x}_e^2,\tilde{\vec x}_v)\in \R^{2|E|+|V|}$ be an optimal solution for $\mathrm{P}_{\alpha,k,t}$.
    We construct a feasible solution for $\mathrm{P}_{1,k,t}$ from $(\tilde{\vec x}_e^1,\tilde{\vec x}_e^2,\tilde{\vec x}_v)$ that has a smaller objective function value than $(\vec x_e^*,\vec x_v^*)$, contradicting its optimality.
    For that purpose, we first observe that by optimality, we may assume without loss of generality that $\tilde{x}_e^1=\tilde{x}_e^2$ for all $e\in E$.
    Otherwise, one can decrease individual entries of $\tilde{x}_e^1$ while simultaneously increasing $\tilde{x}_e^2$ until equality holds.
    Since $\alpha\in [0,1]$ this can only improve the objective value function.
    Furthermore, let $(\vec x_e,\vec x_v)\in \R^{|E|+|V|}$ be defined by $x_e\coloneqq\tilde{x}_e^1+\tilde{x}_e^2$ and $x_v\coloneqq\tilde{x}_v$ for all $e\in E$ and for all $v\in V$.
    Then, $(\vec x_e,\vec x_v)$ is a feasible solution to $\mathrm{P}_{1,k,t}$ with objective function value
    \begin{align*}
        \mathrm{P}_{1,k,t}((\vec x_e,\vec x_v))
        &=\sum_{e\in E} c_ex_e \\
        &=2\sum_{e\in E} c_e\tilde{x}_e^1 \\
        &=\frac{2}{1+\alpha}\bigg[\sum_{e\in E} c_e\tilde{x}_e^1 + \alpha c_e\tilde{x}_e^2\bigg] \\
        &=\frac{2}{1+\alpha}\mathrm{P}_{\vec{\alpha},k,t}((\tilde{\vec x}_e^1,\tilde{\vec x}_e^2,\tilde{\vec x}_v)) \\
        &<\frac{2}{1+\alpha}\mathrm{P}_{\vec{\alpha},k,t}((\hat{\vec x}_e^1,\hat{\vec x}_e^2,\hat{\vec x}_v)) \\
        &=\frac{2}{1+\alpha}\bigg[\sum_{e\in E} c_e\hat{x}_e^1 + \alpha c_e\hat{x}_e^2\bigg] \\
        &=2\sum_{e\in E} c_e\hat{x}_e^1 \\
        &=\sum_{e\in E} c_ex_e^* \\
        &=\mathrm{P}_{1,k,t}((\vec x_e^*,\vec x_v^*)),
    \end{align*}
    contradicting the optimality of $(\vec x_e^*,\vec x_v^*)$ and, thus, proving that $(\hat{\vec x}_e^1,\hat{\vec x}_e^2,\hat{\vec x}_v)$ defined by \eqref{eq:u-def-solution} is optimal for $\mathrm{P}_{\alpha,k,t}$.
\end{proof}

By \Cref{lem:u-optimal-solution}, an optimal solution for $\mathrm{P}_{1,k,t}$ yields an optimal solution for $\mathrm{P}_{\alpha,k,t}$ by the construction outlined in \eqref{eq:u-def-solution}.
The following lemma compares the optimal objective function values of $\smash{\mathrm{P}_{1,k,t}}$ and $\smash{\mathrm{P}_{\alpha,k,t}}$.

\begin{restatable}{lemma}{uPonePalpha}
\label{lem:u-P1-Palpha}
    For $k\in \N$ and a fixed vertex $t\in V \setminus \{s\}$, let $\mathrm{P}_{1,k,t}^*$ and $\mathrm{P}_{\alpha,k,t}^*$ denote the optimal objective function values of $\mathrm{P}_{1,k,t}$ and $\mathrm{P}_{\alpha,k,t}$, respectively.
    Then, it holds that $
    	\mathrm{P}_{1,k,t}^*= \frac{2}{1+\alpha}\mathrm{P}_{\alpha,k,t}^*.$
\end{restatable}

\begin{proof}
    Let $k\in \N$ and $t\in V \setminus \{ s \}$ be fixed and let $(\vec x_e^*,\vec x_v^*)\in \R^{|E|+|V|}$ be an optimal solution for the linear program $\mathrm{P}_{1,k,t}$.
    Let $(\hat{\vec x}_e^1,\hat{\vec x}_e^2,\hat{\vec x}_v)\in \R^{2|E|+|V|}$ be defined by \eqref{eq:u-def-solution}.
    By \Cref{lem:u-optimal-solution}, $(\hat{\vec x}_e^1,\hat{\vec x}_e^2,\hat{\vec x}_v)$ is an optimal solution for $\mathrm{P}_{\alpha,k,t}$.
    Thus, we obtain
	\begin{align*}
		\mathrm{P}_{1,k,t}^*
        &= \mathrm{P}_{1,k,t}((\vec x_e^*,\vec x_v^*)) \\
        &= \sum_{e\in E} c_ex_e^* \\
        &= 2\sum_{e\in E} c_e\hat{x}_e^1 \\ 
        &= \frac{2}{1+\alpha}\bigg[\sum_{e\in E} c_e\hat{x}_e^1 +\alpha c_e\hat{x}_e^2 \bigg] \\
        &= \frac{2}{1+\alpha} \mathrm{P}_{\alpha,k,t}((\hat{\vec x}_e^1,\hat{\vec x}_e^2,\hat{\vec x}_v)) \\
        &= \frac{2}{1+\alpha}\mathrm{P}_{\alpha,k,t}^* \ ,
	\end{align*}        
    proving the claim.
\end{proof}

With \Cref{lem:u-P1-Palpha} we can now show that the trees returned by the primal--dual algorithm are good trees.
The proof uses the dual of the primal program $\mathrm{P}_{1,k,t}$.

\begin{restatable}{lemma}{ualgogivesgoodtrees}
	\label{lem:u-algo-gives-good-trees}
	Let $T_{\lambda}$ be the tree returned by the primal--dual algorithm when run with parameter~$\lambda\geq 0$ and let $k_\lambda$ be the number of vertices in~$V \setminus \{ s \}$ that $T_{\lambda}$ contains.
	Then $T_{\lambda}$ is a good $k_\lambda$-tree for all choices of~$\alpha \in[0,1]$.
\end{restatable}

\begin{proof}
	Let $k\in \setn$, $t\in V \setminus \{ s \}$, and $\alpha\in[0,1]$ be fixed.
    Further, let $P_{\alpha,k,t}$ and $P_{\vec{\alpha},k,t}$ be two $k$-paths that end in $t$ such that 
	$P_{\alpha,k,t}$ is optimal with respect to $\alpha$ and 
	$P_{\alpha,k,t}$ is optimal with respect to $\alpha$.
	Combining \Cref{lem:u-P1-Palpha} with observation~\eqref{eq:u-primal-opt-path} yields
	\begin{align}
		\label{eq:inequality-chain}
		\begin{split}
			\mathrm{P}_{1,k,t}^*
			&= \frac{2}{1+\alpha}\mathrm{P}_{\alpha,k,t}^*
			\leq \frac{2}{1+\alpha}c_{\alpha}(P_{\alpha,k,t})
		\end{split}
	\end{align}
	for all $k\in \setn$ and $t\in V \setminus \{ s \}$.
	
	Let $T_{{\lambda,t}}$ be the tree returned by the primal--dual subroutine when run on terminal vertex $t$ and parameter $\lambda$ and let $k_{\lambda,t}\in \setn$ be the number of vertices in $V\setminus \{ s \}$ that $T_{{\lambda,t}}$ contains.

    To compare the \edgec of~$T(\lambda,t)$ to~$\mathrm{P}_{1,k(\lambda,t),t}^*$, we use the dual linear program~$\mathrm{D}_{1,k,t}$ of~$\mathrm{P}_{1,k,t}$ introduced in \cite{chaudhuri2003paths} and stated as follows:
	\begin{align*}
	\begin{array}{lrcll}
		\text{Max.}
		&\multicolumn{4}{l}{\!\!\!\!(k-1)p-\sum_{v \in V\setminus\{s,t\}} p_v + \sum_{U:t\in U, s\notin U}y_{t,U}}\\
		\text{s.t.}
		&2\sum _{S \ni v}y_{v,S}+p_v &\geq &p
		&\! \text{ for all } v\in V\setminus\{t\}\\
		&\!\!\!\!\sum_{S:e\in\delta(S)}\sum_{v\in S}y_{v,S}+\sum_{U:t\in U,e\in \delta(U)}y_{t,U}  &\leq & c_e
		&\! \text{ for all }e\in E \\
		&p_v &\geq &0 
		&\! \text{ for all } v\in V\\
		&y_{v,S} &\geq &0 
		&\! \text{ for all } S\subseteq V\setminus\{s,t\}, \text{ for all } v\in S\\
		&y_{t,U} &\geq &0
		&\!\text{ for all } U\subseteq V\setminus\{s\} : t\in U. 
	\end{array}
	\end{align*}
    
    Let~$\mathrm{D}_{1,k,t}^*$ and~$\mathrm{P}_{1,k,t}^*$ be the objective values of optimal solutions for~$\mathrm{D}_{1,k,t}$ and~$\mathrm{P}_{1,k,t}$, respectively.
	By weak duality, they obtain
	\begin{align}
	    \label{eq:u_dual}
	    \mathrm{D}_{1,k,t}^* \leq \mathrm{P}_{1,k,t}^*
	\end{align}
	for all $k\in \setn$ and for all $t\in V\setminus \{ s \}$.
	Furthermore, they prove that the \edgec of $T_{{\lambda,t}}$ is bounded by the optimal objective function value $\mathrm{D}_{1,k_{\lambda,t},t}^*$ of $\mathrm{D}_{1,k_{\lambda,t},t}$.
	Together with \eqref{eq:inequality-chain} and \eqref{eq:u_dual}, this yields
	\begin{align}
		\begin{split}
			\label{eq:u-chaudhuri-total}
		    c(T_{{\lambda,t}})
		    \leq \mathrm{D}_{1,k_{\lambda,t},t}^*
		    \leq \mathrm{P}_{1,k_{\lambda,t},t}^*
		    \leq \frac{2}{1+\alpha}c_{\alpha}(P_{{\alpha},k_{\lambda,t},t})
		\end{split}
	\end{align}
	for all $t\in V\setminus \{ s \}$.
	In particular, let $T_{\lambda}$ be the tree returned by the primal--dual algorithm, i.e., the tree $T_{{\lambda,t}}$ with the lowest \edgec over all possible choices of $t\in V\setminus \{ s \}$ and let $P_{\alpha,k_{\lambda}}$ be an optimal $k_{\lambda}$-path.
	Then, \eqref{eq:u-chaudhuri-total} yields that
	\begin{align*}
		c(T_{{\lambda}})\leq \frac{2}{1+\alpha}c_{\alpha}(P_{\alpha,k_{\lambda}}).
	\end{align*}
	This proves that the trees returned by the primal--dual-algorithm of \cite{chaudhuri2003paths} are \emph{good} trees for all choices of ${\alpha}\in[0,1]$.
\end{proof}

Our goal is to use the primal--dual algorithm to compute good $k$-trees for a sufficiently large set of values of~$k$ such that these trees can be used as input for the concatenating algorithm.
However, for a given value~$k$, there may not exist a parameter~$\lambda$ such that the tree~$T_\lambda$ contains exactly $k$ vertices of~$V \setminus \{s\}$, and even if there was such a parameter, there is no simple formula that computes the required parameter~$\lambda$.
Still, the following lemma shows that we can compute either a parameter~$\lambda$ such that~$k=k_{\lambda,t}$ or two values $\lambda, \lambda'$ such that $k_{\lambda,t}<k<k_{\lambda',t}$ and $k_{\lambda''}\in \{k_\lambda,k_{\lambda'}\}$ for all $\lambda''\in [\lambda,\lambda']$ within polynomial time.
The proof works via induction.
In particular, we start with a sufficiently large interval~$(\lambda_l,\lambda_r)$ such that $k_{\lambda_l} < k< k_{\lambda_r}$.
In each step of the induction, we then partition the interval into smaller subintervals such that in the end we either find a parameter~$\lambda$ with~$k=k_\lambda$ or an interval $(\lambda_l^*,\lambda_r^*)$ such that $k_{\lambda_l^*} < k< k_{\lambda_r^*}$ and the primal--dual subroutine returns the same tree for all~$\lambda\in (\lambda_l^*,\lambda_r^*)$.
We denote by $T_{\lambda,t}$ the tree returned by the primal--dual subroutine when run on parameter~$\lambda$ and terminal vertex~$t$ and by~$k_{\lambda,t}$ the number of vertices in~$V \setminus \{ s \}$ it contains.

\begin{restatable}{lemma}{urunningtimelambda}
	\label{lem:u-running-time-lambda}
	Let $k\in\setn$ and $t\in V \setminus \{s\}$ be fixed.
	If $k_{0,t}\leq k_{\lambda,t}$, then one can compute either a parameter~$\lambda$ such that $k=k_{\lambda,t}$
	or two parameters~$\lambda, \lambda'$ such that $k_{\lambda,t}<k<k_{\lambda',t}$ and $k_{\lambda''}\in \{k_{\lambda,t},k_{\lambda',t}\}$ for all $\lambda''\in [\lambda,\lambda']$ within polynomial time.
\end{restatable}

\begin{proof}
	Let $k\in\setn$ and $t\in V\setminus \{ s \}$ be fixed.
	Throughout the proof, we analyze the behavior of the primal--dual subroutine and its returned trees for the fixed terminal vertex~$t$.
	Thus, for ease of notation, we write $T_{\lambda}\coloneqq T_{{\lambda,t}}$ and $k_\lambda\coloneqq k_{\lambda,t}$ for the remainder of this proof.
	
	The basic idea of the proof is as follows.
	We start with two values $\lambda_l$ and $\lambda_r$ such that $k_{\lambda_l} < k< k_{\lambda_r}$.
	Recall that in the primal--dual subroutine the dual variables~$y_S$ of active components are uniformly increased until either an edge or a set event occurs.
	That is, either an edge enters the set~$F$ or an active component becomes inactive, so its dual variable stops growing.
	We assume that the first $i$ events are the same for all $\lambda \in (\lambda_l,\lambda_r)$.
	Then, we subdivide the interval $(\lambda_l,\lambda_r)$ into smaller subintervals such that within each subinterval the first $i+1$ events are the same.
	We continue until we have either found a parameter $\lambda$ with $k=k_\lambda$ or two parameters $\lambda$ and $\lambda'$ such that $k_\lambda<k<k_{\lambda'}$ and~$k_{\lambda''}\in \{k_\lambda,k_{\lambda'}\}$ for all~$\lambda''\in [\lambda,\lambda']$.
	
	For the proof, we need to consider values of $\lambda$ for which two events occur at the same time during the execution of the primal--dual subroutine.
	These values will subdivide the original interval into a set of open subintervals.
	The proof of the lemma then works via induction.
	To this end, assume we are given two values $\lambda_l$ and $\lambda_r$ such that for all $\lambda\in(\lambda_l,\lambda_r)$ the first $i$ events are the same and $k_{\lambda_l}<k < k_{\lambda_r}$.
	For the start of the induction, i.e., $i=0$, we need to ensure the existence of~$\lambda_l$ and~$\lambda_r$ such that $k_{\lambda_l}<k < k_{\lambda_r}$.
	We argue that this can be achieved by setting $\lambda_l=0$ and $\lambda_r=n\cdot\max\{c_e:{e\in E}\}$.
	In particular, for $\lambda=0$ all components except for $t$ become inactive immediately and only the component of $t$ keeps growing until it merges with vertex~$s$.
	The tree returned after the delete phase is a shortest $s$--$t$-path.
	To see this, note that if $\lambda=0$ the growth phase of the primal--dual subroutine behaves like an algorithm computing a shortest-path-tree for vertex~$t$ as it recursively adds edges to~$F$ that lie on a shortest path starting in~$t$.
	However, the algorithm is interrupted as soon as the tree contains vertex~$s$ which is when the delete phase starts and all leafs of the tree are deleted until only the $s$--$t$-path remains.
	Since by assumption this shortest $s$--$t$-path~$T_0$ contains at least $k$ vertices, we have $k_0=k_{\lambda_l}\leq k$.
	For $\lambda=n\cdot\max\{c_e:{e\in E}\}$, on the other hand, all components remain active until they merge with the component containing $s$ and hence, the returned tree $T_{k_{\lambda_r}}$ contains all vertices. 
	This ensures that $k \leq k_{\lambda_r}$.
	Note that we are already done if either $k=k_{\lambda_l}$ or $k=k_{\lambda_r}$ holds.
	Thus, we may assume that $k_{\lambda_l}<k < k_{\lambda_r}$ and the first $i=0$ events are the same for all $\lambda\in(\lambda_l,\lambda_r)$.
	
	We make one further assumption for the induction that we also need to prove for the induction start.
	In particular, we assume that for all subsets $S\subseteq V$ we can find values $\alpha_S,\beta_S\in \R$ such that $y_S=\alpha_S \lambda +\beta_S$ for all $\lambda\in(\lambda_l,\lambda_r)$.
	For the induction start, this assumption is trivially satisfied as initially, all variables $y_S$ are equal to $0$, and thus, we may set $\alpha_S=0$ and $\beta_S=0$ for all $S\subseteq V$.
	
	We are now ready to analyze the induction step.
	To this end, assume we are given an interval~$(\lambda_l,\lambda_r)$ such that the first~$i$ events are the same for all~$\lambda\in(\lambda_l,\lambda_r)$ and $k_{\lambda_l}<k<k_{\lambda_r}$.
	Further, we are given values $\alpha_S,\beta_S\in \R$ for all~$S\subseteq V$ such that~$y_S=\alpha_S \lambda+\beta_S$ for all $\lambda\in(\lambda_l,\lambda_r)$.
	To find a subinterval within $(\lambda_l,\lambda_r)$ where the first $i+1$ many events are the same, we need to find the point in time after the $i$th event when the next event occurs.
	To this end, let $\mathcal{A}$ denote the set of active components after the first $i$ events.
	Further, let $\mathcal{E}_1$ denote the set of edges with one endpoint in an active set and one endpoint in an inactive set, and let $\mathcal{E}_2$ denote the set of edges with both endpoints in two distinct active sets after the first $i$ events.
	Note that the setS~$\mathcal{A},\mathcal{E}_1$, and $\mathcal{E}_2$ are the same for all $\lambda\in(\lambda_l,\lambda_r)$.
	There are three possible candidates for the $i+1$st event.
	First, an active set $S\in \mathcal{A}$ becomes inactive.
	This happens at time
	\begin{align*}
		t_S(\lambda)
		&\coloneqq \lambda |S| -\sum_{T:T\subset S}y_T
		=\lambda |S| - \sum_{T:T\subset S} \alpha_T\lambda + \beta_T.
		\intertext{Second, an edge $e\in \mathcal{E}_1$ becomes tight.
		This happens at time}
		t_e(\lambda)
		&\coloneqq c_e-\sum_{T:e\in \delta(T)}y_T
		=c_e-\sum_{T:e\in \delta(T)}\alpha_T \lambda +\beta_T.
		\intertext{Third, an edge $e\in \mathcal{E}_2$ becomes tight.
		This happens at time}
		t_e(\lambda)
		&\coloneqq \frac{1}{2}\bigg(c_e-\sum_{T:e\in \delta(T)}y_T\bigg)
		=\frac{1}{2}\bigg(c_e-\sum_{T:e\in \delta(T)}\alpha_T \lambda +\beta_T\bigg).
	\end{align*}
	The functions $t_S$ and $t_e$ are affine in $\lambda$ for all $S\in \mathcal{A}$ and all $e\in \mathcal{E}_1\cup \mathcal{E}_2$.
	Furthermore, the functions $t_S$ have a positive slope for all $S\in \mathcal{A}$ while the functions $t_e$ have a negative slope for all $e\in \mathcal{E}_1\cup \mathcal{E}_2$.
	An illustration is given in \Cref{fig:lambda-subinterval}.
	The minimum of these affine functions is a piecewise affine concave function and the affine function $t_S$ or $t_e$ for which the minimum is attained for some $\lambda$ corresponds to the $i+1$st event for this precise value of $\lambda$.
	We are thus particularly interested in the values of $\lambda$ for which the minimum has a breakpoint, that is two affine functions intersect.
	These values subdivide our interval into smaller open intervals such that within each of these, the first $i+1$ events are the same.	
	
	At this point, it also becomes clear, why we divide $(\lambda_l,\lambda_r)$ into further open intervals instead of closed ones:
	If two events occur simultaneously, we need to distinguish, which event is carried out first.
	If both events are set events, it does not matter in which order the events are carried out.
	In particular, letting one set become inactive does not affect the other set.
	Hence, the second set will still become inactive right after.
	If both events are edge events and the corresponding edges do not connect the same two components, we also carry out both events, no matter which one we carry out first.
	However, if both edges $e$ and $e'$ connect the same two components, only one event will be carried out and the choice of this event may result in a different tree returned after the delete phase.
	If the two corresponding affine functions $t_e$ and $t_{e'}$ have a different slope, considering the open interval to the left and to the right of~$\lambda$ determines a strict order in which the events need to be carried out.
	In the special case, where the slopes of $t_e$ and $t_{e'}$ are the same, we do a deterministic tie-breaking.
	Next, if one event is an edge event for edge $e$ and the other a set event for set $S$, the choice of which event is carried out first may have a strong effect on the further course of the algorithm.
	But, since the slope of the corresponding function $t_S$ is positive and the slope of the function $t_e$ is negative, considering the open interval to the left and to the right of~$\lambda$ again determines a strict order in which the events need to be carried out.
	This approach also takes care of all scenarios, in which more than two events coincide.
	
\begin{figure}[t]
	\begin{center}
	\begin{tikzpicture}[xscale=5,yscale = 4, shorten >= 0pt,
	    shorten <= 0pt,]
	\begin{scope}[xshift=15cm,yshift=-1.5cm]
	
	\draw[thick,Dandelion] (0,1/10) -- (2/3,9/10);
	\draw[thick,Dandelion] (0,4/10) -- (9/10,9/10);
	\draw[thick,MidnightBlue] (0,4/8) -- (1,3/8);
	\draw[thick,MidnightBlue] (0,6/8) -- (1,2/8);
	
	\draw[very thick,Red] (0,1/10) -- (0.302,0.462);
	\draw[very thick,Red] (0.302,0.462) -- (2/3,5/12);
	\draw[very thick,Red] (2/3,5/12) -- (1,2/8);
	
	\node[label=above:{$t_S(\lambda)$}] (s) at (2/3,9/10) {};
	\node[label=above:{$t_{S'}(\lambda)$}] (s) at (9/10,9/10) {};
	\node[label=right:{$t_e(\lambda)$}] (s) at (1,3/8) {};
	\node[label=right:{$t_{e'}(\lambda)$}] (s) at (1,2/8) {};
	
	\draw[thin,->] (-0.1,0) to (-0.1,1) node[above] {}; 
	\draw[thin,->] (-0.1,0) to (1.15,0) node[right] {$\lambda$};
					
	\draw[thin] (1,0.05) -- (1,-0.05) node[below] {$\lambda_r$};
	\draw[thin] (0,0.05) -- (0,-0.05) node[below] {$\lambda_l$};
	\draw[thin] (0.302,0.05) -- (0.302,-0.05) node[below] {$\lambda_1^{\phantom{+}}$}; 
	\draw[thin] (2/3,0.05) -- (2/3,-0.05) node[below] {$\lambda_2^{\phantom{+}}$};
	
	\end{scope}
	\end{tikzpicture}
	\end{center}
	\caption{The affine functions $t(\lambda)$ depicted for iteration $i$.
	The \textcolor{Dandelion}{orange} functions with positive slope determine the time when a set event occurs.
	The \textcolor{MidnightBlue}{blue} functions with negative slope determine the time when an edge event occurs.
	The minimum determines the next event and is colored in \textcolor{Red}{red}.
	The breakpoints of the minimum divide the interval $(\lambda_l,\lambda_r)$ into open subintervals such that within each subinterval the first $i+1$ events are the same.} 
	\label{fig:lambda-subinterval}
	\end{figure}
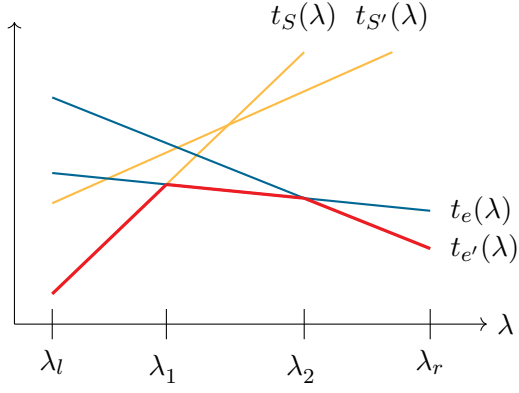

	Let $\lambda_1,\dots, \lambda_j$ be the breakpoints of the minimum of our set of affine functions.
	Then, $\lambda_1,\dots, \lambda_j$ divide our previous interval $(\lambda_l,\lambda_r)$ into subintervals $(\lambda_l,\lambda_1),(\lambda_1,\lambda_2),\dots,(\lambda_j,\lambda_r)$ such that the first $i+1$ events are the same within each subinterval.
	Since the set $\mathcal{A}$ contains at most $n$ active sets and the sets $\mathcal{E}_1\cup \mathcal{E}_2$ contain at most $m$ edges, there are at most $n+m$ affine functions.
	Hence, the minimum of these functions has at most $n+m-1$ many breakpoints, resulting in at most $n+m$ subintervals.
	
	Next, we run the primal--dual subroutine for all parameters $\lambda_1,\dots,\lambda_j$.
	If there exists a parameter $\lambda_*\in \{\lambda_1,\dots,\lambda_j\}$ such that $k=k_{\lambda_*}$ we are done.
	Otherwise, it follows from $k_{\lambda_l}<k< k_{\lambda_r}$, that there exists at least one subinterval $(\lambda_h,\lambda_{h+1})$ for some $h\in \{0,1,\dots,j\}$ where $\lambda_0=\lambda_l$ and $\lambda_{j+1}=\lambda_r$ such that $k_{\lambda_h}<k< k_{\lambda_{h+1}}$.
	This is the interval for the next step of the induction.
	To continue, we need to determine values for $\alpha_S$ and $\beta_S$ for all subsets $S\subseteq V$ such that $y_S=\alpha_S \lambda + \beta_S$ for all $\lambda \in ({\lambda_h},{\lambda_{h+1}})$ after the $i+1$st event.
	To this end, recall that the interval $({\lambda_h},{\lambda_{h+1}})$ corresponds to an affine segment of the minimum of the affine functions $t_S$ and $t_e$.
	Thus, there is a unique affine function $t$ determining the minimum for all $\lambda \in ({\lambda_h},{\lambda_{h+1}})$.
	In particular, for $\lambda \in ({\lambda_h},{\lambda_{h+1}})$ the dual variables $y_S$ of an active component $S\in \mathcal{A}$ have been increased by exactly $t(\lambda)$ since this is the time span between the $i$th and the $i+1$st event.
	Thus, we have $y_S=\alpha_S \lambda +\beta_S + t(\lambda)$ for all $S\in \mathcal{A}$ and $\lambda \in ({\lambda_h},{\lambda_{h+1}})$.
	Since $t(\lambda)$ is an affine function in $\lambda$ we can update $\alpha_S$ and $\beta_S$ appropriately.
	
	We claim that the induction terminates after at most polynomially many iterations.
	To this end, note that the set of tight edges~$F$ is always a forest.
	With $|V|=n+1$, at most $n$ edge events can occur before the algorithm terminates.
	Furthermore, the number of active components cannot increase from one iteration to the next.
	At the beginning of the primal--dual subroutine, there are at most $|V \setminus \{ s \}|=n$ active components.
	Thus, at most $n$ set events can occur before the algorithm terminates.
	In total, we have an upper bound of at most $2n$ different events for a fixed $\lambda$, so there are at most $2n$ steps of the induction.
	Hence, if throughout all induction steps we did not compute a parameter~$\lambda$ such that $k=k_\lambda$, we end up with a final interval~$(\lambda_l,\lambda_r)$ such that $k_{\lambda_l}<k<k_{\lambda_r}$ and all events are the same for all $\lambda\in (\lambda_l,\lambda_r)$.
	More precisely, the primal--dual subroutine returns the same tree for all $\lambda\in (\lambda_l,\lambda_r)$.
	Thus, we may pick an arbitrary parameter $\lambda\in (\lambda_l,\lambda_r)$.
	If $k=k_\lambda$ we are done.
	Otherwise, we have either $k_{\lambda_l}<k<k_{\lambda}$ or $k_{\lambda}<k<k_{\lambda_r}$.
	In the former case, we have found two parameters $\lambda_l$ and $\lambda$ such that $k_{\lambda_l}<k<k_{\lambda}$ and $k_{\lambda'}\in\{k_{\lambda_l},k_{\lambda}\}$ for all $\lambda'\in[\lambda_l,\lambda]$ as desired.
	The latter case is analogous.
	
	It remains to analyze the running time of the entire procedure.
	Note that the most time-consuming part of a single iteration is the execution of the primal--dual subroutine for all parameters $\lambda_1,\dots,\lambda_j$.
	However, since $j$ is at most $n+m$ and a single execution of the subroutine requires $\mathcal{O}(n^2)$ time, the running time of a single iteration lies in $\mathcal{O}(n^4)$, where we used that $m\leq n^2$.
	Altogether, the entire procedure of computing either a value $\lambda$ such that $k=k_\lambda$ or two parameters $\lambda, \lambda'$ such that $k_\lambda<k<k_{\lambda'}$ and $k_{\lambda''}\in\{k_\lambda,k_{\lambda'}\}$ for all $\lambda''\in [\lambda,\lambda']$ requires running time bounded by $\mathcal{O}(n^5)$.
	This finishes the proof.	
\end{proof}

Let $\mathcal{T}^*$ be the set of trees computed by the primal--dual algorithm. By \Cref{lem:u-algo-gives-good-trees}, these are good $k$-trees for some values of~$k$.
We proceed to argue that~$\mathcal{T}^*$ contains a good $0$-tree and a good $n$-tree.
To see this, note that for $\lambda=n\cdot \max\{c_e:{e\in E}\}$ the tree~$T_{\lambda,t}$ returned by the primal--dual subroutine contains all vertices no matter the choice of~$t\in V \setminus \{s\}$.
Thus, also $T_\lambda$ contains all vertices and is a good $n$-tree by~\Cref{lem:u-algo-gives-good-trees}.
The set~$\mathcal{T}^*$, however, may not contain a good $k$-tree for all values of~$k\in\setn_0$.
We, thus, introduce the notion of phantom trees.

\subsection{Phantom Trees}
\label{ssec:u-phantom-trees}

The goal of this section is to design an algorithm that computes good $k$-trees for all values of $k\in\setn_0$.
To do so, we introduce the definition of \emph{phantom trees}.
Such a tree is an artificial object that is attained by a linear interpolation of two real trees.
It, thus, does not really exist as a tree of a graph but functions as a placeholder.
The algorithm is stated in \Cref{alg:u-phantom-tree} and the basic idea is as follows.
For each pair of a value~$k\in\setn$ and a terminal vertex~$t\in V \setminus \{s\}$, the algorithm applies \Cref{lem:u-running-time-lambda}, which returns either one or two parameters~$\lambda,\lambda'\in\R_{\geq 0}$.
In the former case, we run the primal--dual subroutine and obtain a tree~$T_{\lambda,t}$ that contains exactly $k$~vertices.
In the latter case, we run the primal--dual subroutine for each of the two parameters and obtain two trees~$T_{\lambda,t}$ and~$T_{\lambda',t}$.
Afterward, we construct a phantom tree as a linear interpolation of~$T_{\lambda,t}$ and~$T_{\lambda',t}$ and say it contains $k$~vertices.
In total, this gives a set of $n$~real and phantom $k$-trees for each $k\in\setn$ out of which we pick the cheapest and denote it by~$T_k$.
Finally, the set~$\mathcal{T}$ containing all those $k$-trees is returned.
We argue that every tree~$T_k\in\mathcal{T}$ is a good $k$-tree.

\begin{algorithm}[t]
\small
\caption{Construct Real and Phantom Trees\label{alg:u-phantom-tree}}
\DontPrintSemicolon
$T_{0} \gets (\cdot)$; $\mathcal{T} \gets \{ T_{0} \}$\;
\For{$k\in\setn$}{
	$\mathcal{T}_k \gets \emptyset$\;
	\For{$t\in V \setminus \{s\}$}{
		\If{$k_{0,t}\leq k$}{
			\If{\Cref{lem:u-running-time-lambda} finds $\lambda\in \R_{\geq 0}$ such that $k=k_{\lambda,t}$}{
			$\mathcal{T}_k\gets\mathcal{T}_k\cup T_{\lambda,t}$\;
			}
			\Else{
			$\mu \gets \frac{k-k_{\lambda,t}}{k_{\lambda',t}-k_{\lambda,t}}$ where $\lambda,\lambda'\in \R_{\geq 0}$ are as in \Cref{lem:u-running-time-lambda}\;
			define $T_{k,t}$ as phantom $k$-tree with
			$c(T_{k,t}) \gets (1 - \mu)c(T_{\lambda,t}) + \mu \,c(T_{\lambda',t})$\;
			$\mathcal{T}_k\gets \mathcal{T}_k\cup T_{k,t}$\;
			}
		}
	}
	$\mathcal{T}\gets \mathcal{T}\cup T_k$ where $T_k\gets \arg\min\{c(T): T\in \mathcal{T}_k \}$\;
}
\Return{$\mathcal{T}$}\;
\end{algorithm}

\begin{restatable}{lemma}{ugoodphantomtrees2}
	\label{lem:u-good-phantom-trees2}
	Let $\mathcal{T}$ be the set of real and phantom trees returned by~\Cref{alg:u-phantom-tree} and let $T_k\in\mathcal{T}$ for some~$k\in \setn_0$.
	Then $T_k$ is a good $k$-tree.
\end{restatable}

\begin{proof}
	For $k=0$ we have $T_0=(\cdot )$ and the statement is trivial.
	Thus, let $k\in \setn$ and $t\in V\setminus \{ s \}$ be arbitrary but fixed and denote by $T_{k,t}$ the tree constructed during the execution of~\Cref{alg:u-phantom-tree}.
	Note that~$T_{k,t}$ does not exist if $k_{0,t}>k$.
	However, this case only arises, if a shortest $s$--$t$-path in $G$ already contains more than $k$ vertices and hence choosing $T_{k_0,t}$ instead of~$T_{k,t}$ decreases the cost while increasing the number of visited vertices.
	Thus, such a tree~$T_{k,t}$ would be redundant.
	Hence, we may assume for now that $k_{0,t}\leq k$.
	If $T_{k,t}$ is a real tree, \eqref{eq:u-chaudhuri-total} yields
	\begin{align*}
		c(T_{k,t})
		\leq \mathrm{D}_{1,k,t}^*
	    \leq \mathrm{P}_{1,k,t}^*
	    \leq \frac{2}{1+\alpha}c_{\alpha}(P_{\alpha,k,t}),
	\end{align*}
	where $P_{\alpha,k,t}$ is an $\alpha$-optimal $k$-path that ends in~$t$.
	If $T_{k,t}$ is a phantom tree, it is constructed by a linear interpolation of two real trees.
	\citet{chaudhuri2003paths} show that also in this setting, we have $c(T_{k,t})\leq \mathrm{D}_{1,k,t}^*$.
	So, \eqref{eq:u-chaudhuri-total} again yields
	$c(T_{k,t})\leq \frac{2}{1+\alpha}c_{\vec{\alpha}}(P_{{\alpha},k,t})$.
	Since tree~$T_k$ is set to the cheapest $k$-tree among all such real and phantom $k$-trees~$T_{k,t}$,
	we obtain $c(T_{k})\leq \frac{2}{1+\alpha}c_{\alpha}(P_{\alpha,k,t})$ for all choices of a terminal vertex~$t\in V\setminus \{ s \}$.
	In particular, we obtain $c(T_{k})\leq \frac{2}{1+\alpha}c_{\alpha}(P_{{\vec{\alpha}},k})$, where $P_{\alpha,k}$ is an $\alpha$-optimal $k$-path with no restriction on which vertex to end on.
	Hence, $T_k$ is a good $k$-tree.
\end{proof}

\subsection{Non-Monotonicity of trees in the Budget Parameter}
\label{sec:app:non-monotonicity}
A technical point that will be important in our use of the primal--dual framework is that the trees returned by the procedure are not monotone in the budget parameter.
More precisely, suppose that the primal--dual algorithm is run with a parameter that can be interpreted as a budget, or penalty, per vertex.
One might expect that increasing this parameter can only lead to larger trees, in the sense that the set of vertices spanned by the resulting tree is monotone.
This monotonicity property is often implicit in related uses of prize-collecting or Lagrangian-relaxation based algorithms for variants of $k$-MST; see, e.g., the algorithms~\cite{arora2006approximation}.
However, this monotonicity does not hold for the primal--dual procedure we use here.
In fact, the failure already occurs before the final pruning step: a larger value of the parameter may change the order in which moats merge and are deactivated, and the resulting intermediate tree need not contain the tree obtained for a smaller parameter.
We give an explicit example below.

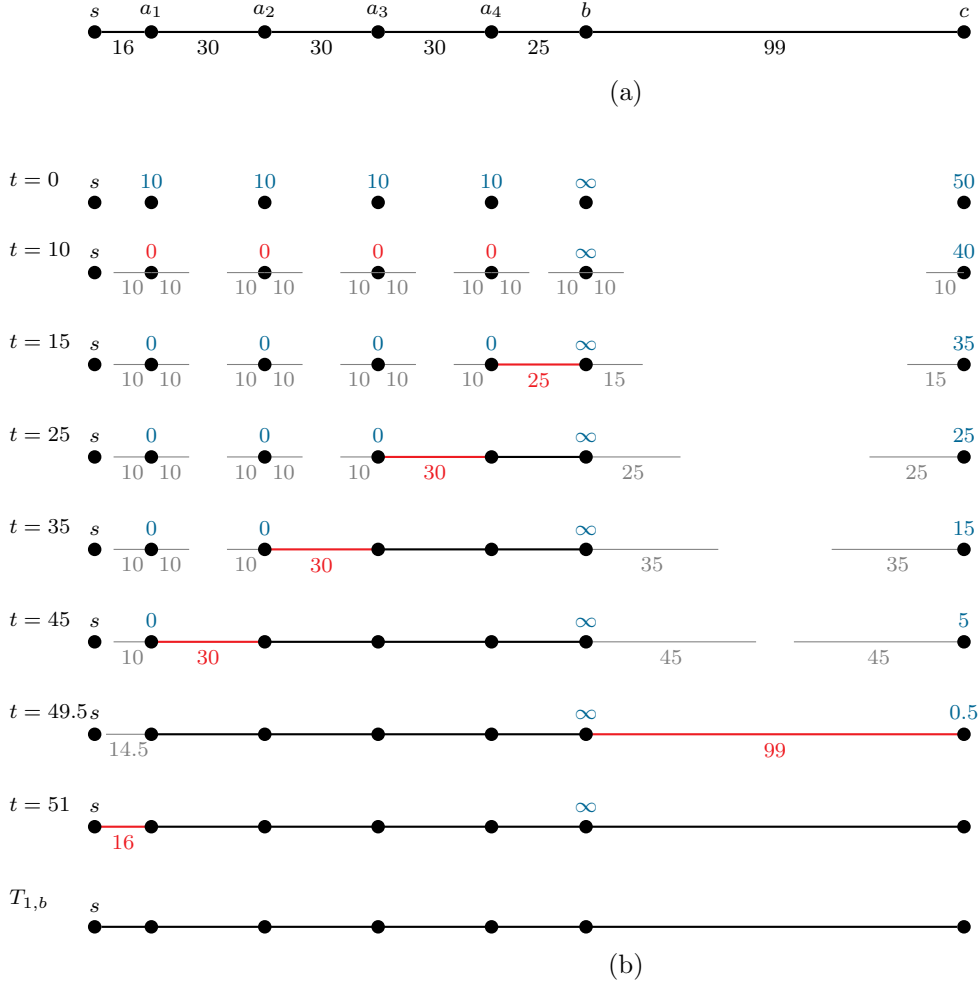
\begin{figure}[t]
\begin{subfigure}[t]{\textwidth}
\scriptsize
\begin{tikzpicture}[xscale=0.5,shorten > = 0pt]
	\node[state,label=above:{$s$}] (s) at (0,0) {};
	\node[state,label=above:{$a_1$}] (v1) at (1.5,0) {};
	\node[state,label=above:{$a_2$}] (v2) at (4.5,0) {};
	\node[state,label=above:{$a_3$}] (v3) at (7.5,0) {};
	\node[state,label=above:{$a_4$}] (v4) at (10.5,0) {};
	\node[state,label=above:{$b$}] (v5) at (13,0) {};
	\node[state,label=above:{$c$}] (v6) at (23,0) {};
	\node[label=above:{\makebox[\widthof{$t=00000$}][l]{}}] (t) at (-1,0) {};
	
	\draw[thick] (s) to node[below]{$16$} (v1); 
	\draw[thick] (v1) to node[below]{$30$} (v2); 
	\draw[thick] (v2) to node[below]{$30$} (v3); 
	\draw[thick] (v3) to node[below]{$30$} (v4); 
	\draw[thick] (v4) to node[below]{$25$} (v5); 
	\draw[thick] (v5) to node[below]{$99$} (v6); 
\end{tikzpicture}
\caption{}
\label{fig:exmp-cgrt-not-monotone}
\end{subfigure}

\vspace{0.7cm}
\begin{subfigure}[t]{\textwidth}
\scriptsize
\begin{tikzpicture}[xscale=0.5,shorten > = 0pt]
	\node[state,label=above:{$s$}] (s) at (0,0) {};
	\node[state,label=above:{\textcolor{\myblue}{$10$}}] (v1) at (1.5,0) {};
	\node[state,label=above:{\textcolor{\myblue}{$10$}}] (v2) at (4.5,0) {};
	\node[state,label=above:{\textcolor{\myblue}{$10$}}] (v3) at (7.5,0) {};
	\node[state,label=above:{\textcolor{\myblue}{$10$}}] (v4) at (10.5,0) {};
	\node[state,label=above:{\textcolor{\myblue}{$\infty$}}] (v5) at (13,0) {};
	\node[state,label=above:{\textcolor{\myblue}{$50$}}] (v6) at (23,0) {};
	\node[label=above:{\makebox[\widthof{$t=00000$}][l]{$t=0$}}] (t) at (-1,0) {};
	
\end{tikzpicture}

\vspace{0.3cm}
\begin{tikzpicture}[xscale=0.5,shorten > = 0pt]
	\node[state,label=above:{$s$}] (s) at (0,0) {};
	\node[state,label=above:{\textcolor{\myred}{$0$}}] (v1) at (1.5,0) {};
	\node[state,label=above:{\textcolor{\myred}{$0$}}] (v2) at (4.5,0) {};
	\node[state,label=above:{\textcolor{\myred}{$0$}}] (v3) at (7.5,0) {};
	\node[state,label=above:{\textcolor{\myred}{$0$}}] (v4) at (10.5,0) {};
	\node[state,label=above:{\textcolor{\myblue}{$\infty$}}] (v5) at (13,0) {};
	\node[state,label=above:{\textcolor{\myblue}{$40$}}] (v6) at (23,0) {};
	\node[label=above:{\makebox[\widthof{$t=00000$}][l]{$t=10$}}] (t) at (-1,0) {};
	
	\draw[-,gray] (0.5,0) to node[below]{$10$} (1.5,0);
	\draw[-,gray] (1.5,0) to node[below]{$10$} (2.5,0);
	\draw[-,gray] (3.5,0) to node[below]{$10$} (4.5,0);
	\draw[-,gray] (4.5,0) to node[below]{$10$} (5.5,0);
	\draw[-,gray] (6.5,0) to node[below]{$10$} (7.5,0);
	\draw[-,gray] (7.5,0) to node[below]{$10$} (8.5,0);
	\draw[-,gray] (9.5,0) to node[below]{$10$} (10.5,0);
	\draw[-,gray] (10.5,0) to node[below]{$10$} (11.5,0);
	\draw[-,gray] (12,0) to node[below]{$10$} (13,0);
	\draw[-,gray] (13,0) to node[below]{$10$} (14,0);
	\draw[-,gray] (22,0) to node[below]{$10$} (23,0);
\end{tikzpicture}

\vspace{0.3cm}
\begin{tikzpicture}[xscale=0.5,shorten > = 0pt]
	\draw[-,gray] (0.5,0) to node[below]{$10$} (1.5,0);
	\draw[-,gray] (1.5,0) to node[below]{$10$} (2.5,0);
	\draw[-,gray] (3.5,0) to node[below]{$10$} (4.5,0);
	\draw[-,gray] (4.5,0) to node[below]{$10$} (5.5,0);
	\draw[-,gray] (6.5,0) to node[below]{$10$} (7.5,0);
	\draw[-,gray] (7.5,0) to node[below]{$10$} (8.5,0);
	\draw[-,gray] (9.5,0) to node[below]{$10$} (10.5,0);
	\draw[-,\myred,thick] (10.5,0) to node[below]{$25$} (13,0);
	\draw[-,gray] (13,0) to node[below]{$15$} (14.5,0);
	\draw[-,gray] (21.5,0) to node[below]{$15$} (23,0);
	
	\node[state,label=above:{$s$}] (s) at (0,0) {};
	\node[state,label=above:{\textcolor{\myblue}{$0$}}] (v1) at (1.5,0) {};
	\node[state,label=above:{\textcolor{\myblue}{$0$}}] (v2) at (4.5,0) {};
	\node[state,label=above:{\textcolor{\myblue}{$0$}}] (v3) at (7.5,0) {};
	\node[state,label=above:{\textcolor{\myblue}{$0$}}] (v4) at (10.5,0) {};
	\node[state,label=above:{\textcolor{\myblue}{$\infty$}}] (v5) at (13,0) {};
	\node[state,label=above:{\textcolor{\myblue}{$35$}}] (v6) at (23,0) {};
	\node[label=above:{\makebox[\widthof{$t=00000$}][l]{$t=15$}}] (t) at (-1,0) {};
\end{tikzpicture}

\vspace{0.3cm}
\begin{tikzpicture}[xscale=0.5,shorten > = 0pt]
	\draw[-,gray] (0.5,0) to node[below]{$10$} (1.5,0);
	\draw[-,gray] (1.5,0) to node[below]{$10$} (2.5,0);
	\draw[-,gray] (3.5,0) to node[below]{$10$} (4.5,0);
	\draw[-,gray] (4.5,0) to node[below]{$10$} (5.5,0);
	\draw[-,gray] (6.5,0) to node[below]{$10$} (7.5,0);
	\draw[-,\myred,thick] (7.5,0) to node[below]{$30$} (10.5,0);
	\draw[thick] (10.5,0) to (13,0);
	\draw[-,gray] (13,0) to node[below]{$25$} (15.5,0);
	\draw[-,gray] (20.5,0) to node[below]{$25$} (23,0);
	
	\node[state,label=above:{$s$}] (s) at (0,0) {};
	\node[state,label=above:{\textcolor{\myblue}{$0$}}] (v1) at (1.5,0) {};
	\node[state,label=above:{\textcolor{\myblue}{$0$}}] (v2) at (4.5,0) {};
	\node[state,label=above:{\textcolor{\myblue}{$0$}}] (v3) at (7.5,0) {};
	\node[state] (v4) at (10.5,0) {};
	\node[state,label=above:{\textcolor{\myblue}{$\infty$}}] (v5) at (13,0) {};
	\node[state,label=above:{\textcolor{\myblue}{$25$}}] (v6) at (23,0) {};
	\node[label=above:{\makebox[\widthof{$t=00000$}][l]{$t=25$}}] (t) at (-1,0) {};
\end{tikzpicture}

\vspace{0.3cm}
\begin{tikzpicture}[xscale=0.5,shorten > = 0pt]
	\draw[-,gray] (0.5,0) to node[below]{$10$} (1.5,0);
	\draw[-,gray] (1.5,0) to node[below]{$10$} (2.5,0);
	\draw[-,gray] (3.5,0) to node[below]{$10$} (4.5,0);
	\draw[-,\myred,thick] (4.5,0) to node[below]{$30$} (7.5,0);
	\draw[thick] (7.5,0) to (10.5,0);
	\draw[thick] (10.5,0) to (13,0);
	\draw[-,gray] (13,0) to node[below]{$35$} (16.5,0);
	\draw[-,gray] (19.5,0) to node[below]{$35$} (23,0);
	
	\node[state,label=above:{$s$}] (s) at (0,0) {};
	\node[state,label=above:{\textcolor{\myblue}{$0$}}] (v1) at (1.5,0) {};
	\node[state,label=above:{\textcolor{\myblue}{$0$}}] (v2) at (4.5,0) {};
	\node[state] (v3) at (7.5,0) {};
	\node[state] (v4) at (10.5,0) {};
	\node[state,label=above:{\textcolor{\myblue}{$\infty$}}] (v5) at (13,0) {};
	\node[state,label=above:{\textcolor{\myblue}{$15$}}] (v6) at (23,0) {};
	\node[label=above:{\makebox[\widthof{$t=00000$}][l]{$t=35$}}] (t) at (-1,0) {};
\end{tikzpicture}

\vspace{0.3cm}
\begin{tikzpicture}[xscale=0.5,shorten > = 0pt]
	\draw[-,gray] (0.5,0) to node[below]{$10$} (1.5,0);
	\draw[-,\myred,thick] (1.5,0) to node[below]{$30$} (4.5,0);
	\draw[thick] (4.5,0) to  (7.5,0);
	\draw[thick] (7.5,0) to (10.5,0);
	\draw[thick] (10.5,0) to (13,0);
	\draw[-,gray] (13,0) to node[below]{$45$} (17.5,0);
	\draw[-,gray] (18.5,0) to node[below]{$45$} (23,0);
	
	\node[state,label=above:{$s$}] (s) at (0,0) {};
	\node[state,label=above:{\textcolor{\myblue}{$0$}}] (v1) at (1.5,0) {};
	\node[state] (v2) at (4.5,0) {};
	\node[state] (v3) at (7.5,0) {};
	\node[state] (v4) at (10.5,0) {};
	\node[state,label=above:{\textcolor{\myblue}{$\infty$}}] (v5) at (13,0) {};
	\node[state,label=above:{\textcolor{\myblue}{$5$}}] (v6) at (23,0) {};
	\node[label=above:{\makebox[\widthof{$t=00000$}][l]{$t=45$}}] (t) at (-1,0) {};
\end{tikzpicture}

\vspace{0.3cm}
\begin{tikzpicture}[xscale=0.5,shorten > = 0pt]
	\draw[-,gray] (0.3,0) to node[below]{$14.5$} (1.5,0);
	\draw[thick] (1.5,0) to (4.5,0);
	\draw[thick] (4.5,0) to (7.5,0);
	\draw[thick] (7.5,0) to (10.5,0);
	\draw[thick] (10.5,0) to (13,0);
	\draw[-,\myred,thick] (13,0) to node[below]{$99$} (23,0);
	
	\node[state,label=above:{$s$}] (s) at (0,0) {};
	\node[state] (v1) at (1.5,0) {};
	\node[state] (v2) at (4.5,0) {};
	\node[state] (v3) at (7.5,0) {};
	\node[state] (v4) at (10.5,0) {};
	\node[state,label=above:{\textcolor{\myblue}{$\infty$}}] (v5) at (13,0) {};
	\node[state,label=above:{\textcolor{\myblue}{$0.5$}}] (v6) at (23,0) {};
	\node[label=above:{\makebox[\widthof{$t=00000$}][l]{$t=49.5$}}] (t) at (-1,0) {};
\end{tikzpicture}

\vspace{0.3cm}
\begin{tikzpicture}[xscale=0.5,shorten > = 0pt]
	\draw[-,\myred,thick] (0,0) to node[below]{$16$} (1.5,0);
	\draw[thick] (1.5,0) to (4.5,0);
	\draw[thick] (4.5,0) to (7.5,0);
	\draw[thick] (7.5,0) to (10.5,0);
	\draw[thick] (10.5,0) to (13,0);
	\draw[thick] (13,0) to (23,0);
	
	\node[state,label=above:{$s$}] (s) at (0,0) {};
	\node[state] (v1) at (1.5,0) {};
	\node[state] (v2) at (4.5,0) {};
	\node[state] (v3) at (7.5,0) {};
	\node[state] (v4) at (10.5,0) {};
	\node[state,label=above:{\textcolor{\myblue}{$\infty$}}] (v5) at (13,0) {};
	\node[state] (v6) at (23,0) {};
	\node[label=above:{\makebox[\widthof{$t=00000$}][l]{$t=51$}}] (t) at (-1,0) {};
\end{tikzpicture}

\vspace{0.3cm}
\begin{tikzpicture}[xscale=0.5,shorten > = 0pt]
	\node[state,label=above:{$s$}] (s) at (0,0) {};
	\node[state] (v1) at (1.5,0) {};
	\node[state] (v2) at (4.5,0) {};
	\node[state] (v3) at (7.5,0) {};
	\node[state] (v4) at (10.5,0) {};
	\node[state] (v5) at (13,0) {};
	\node[state] (v6) at (23,0) {};
	\node[label=above:{\makebox[\widthof{$t=00000$}][l]{$T_{1,b}$}}] (t) at (-1,0) {};
	
	\draw[thick] (s) to (v6);
\end{tikzpicture}
\caption{}
\label{fig:exmp-cgrt-not-monotone-1}
\end{subfigure}
\caption{Illustrations of \Cref{exmp:cgrt-not-monotone}.
(a)~The input graph~$G$.
(b)~The primal--dual subroutine with parameter~$\lambda=1$ and terminal vertex~$b$.}
\end{figure}

\begin{example}
\label{exmp:cgrt-not-monotone}
	Consider the graph~$G=(V,E)$ illustrated in \Cref{fig:exmp-cgrt-not-monotone}.
	Each vertex~$a_1,a_2,a_3$, and~$a_4$ is part of a clique of~$9$ further vertices not drawn in the figure.
	Similarly, vertex~$c$ is part of a clique of~$49$ further vertices.
	The edges within each clique have \edgec~$0$, so as soon as the primal--dual subroutine starts for any input~$\lambda$ and~$t$, all clique edges become tight immediately, and it suffices to consider the thus obtained connected components.
	We select vertex~$b$ as the designated terminal vertex and run the primal--dual subroutine for parameters~$\lambda=1,\lambda=1.6$, and~$\lambda=2$.
	\Cref{fig:exmp-cgrt-not-monotone-1}, \Cref{fig:exmp-cgrt-not-monotone-1-6}, and \Cref{fig:exmp-cgrt-not-monotone-2} show multiple snapshots of the subroutine after different time steps~$t$.
	The remaining budgets of connected components are shown in blue above the components.
	Edge and set events are emphasized in red.
	A grey half-edge incident to a vertex~$v$ indicates the sum of the dual variables~$y_S$ for which~$v\in S$.
	For~$\lambda=1$, the algorithm terminates after~$51$ time steps.
	Since the clique at vertex~$c$ has always belonged to an active component, no edge is deleted in the delete phase, and the returned tree is~$T(1,b)=G$.
	For~$\lambda=1.6$, the algorithm terminates after~$80$ time steps.
	Again, no edge is deleted in the delete phase, but in contrast to~$\lambda=1$, the returned tree does not contain the clique at vertex~$c$.
	Finally, for~$\lambda=2$, the algorithm terminates after~$84$ time steps.
	Again, the clique at vertex~$c$ has always belonged to an active component, and thus the returned tree is~$T(2,b)=G$.
	In total, we obtain~$k(1.6,b)=41<91=k(1,b)=k(2,b)$ demonstrating that~$k(\lambda,t)$ is not monotone in~$\lambda$.
	\hfill~$\meindreieck$
\end{example}

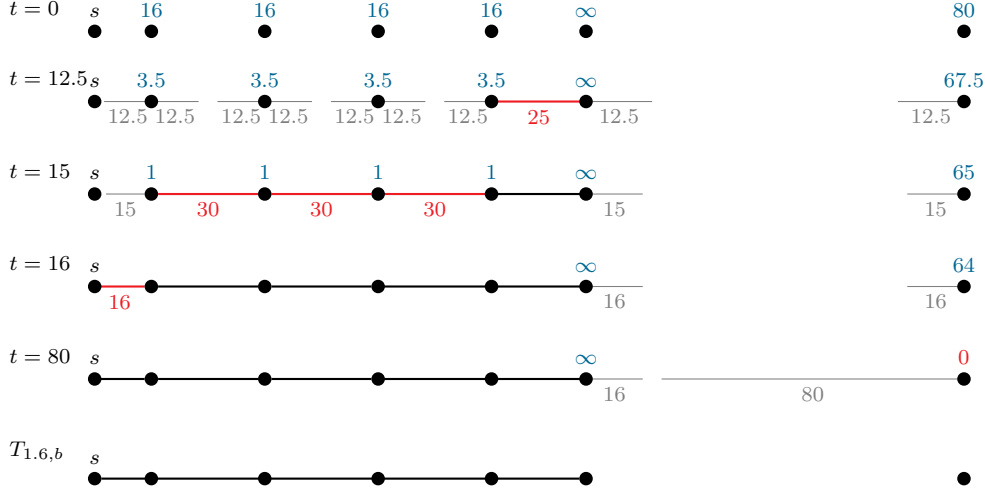
\begin{figure}[t]
\scriptsize
\begin{subfigure}[t]{\textwidth}
\begin{tikzpicture}[xscale=0.5,shorten > = 0pt]
	\node[state,label=above:{$s$}] (s) at (0,0) {};
	\node[state,label=above:{\textcolor{\myblue}{$16$}}] (v1) at (1.5,0) {};
	\node[state,label=above:{\textcolor{\myblue}{$16$}}] (v2) at (4.5,0) {};
	\node[state,label=above:{\textcolor{\myblue}{$16$}}] (v3) at (7.5,0) {};
	\node[state,label=above:{\textcolor{\myblue}{$16$}}] (v4) at (10.5,0) {};
	\node[state,label=above:{\textcolor{\myblue}{$\infty$}}] (v5) at (13,0) {};
	\node[state,label=above:{\textcolor{\myblue}{$80$}}] (v6) at (23,0) {};
	\node[label=above:{\makebox[\widthof{$t=00000$}][l]{$t=0$}}] (t) at (-1,0) {};
\end{tikzpicture}

\vspace{0.3cm}
\begin{tikzpicture}[xscale=0.5,shorten > = 0pt]
	\draw[-,gray] (0.25,0) to node[below]{$12.5$} (1.5,0);
	\draw[-,gray] (1.5,0) to node[below]{$12.5$} (2.75,0);
	\draw[-,gray] (3.25,0) to node[below]{$12.5$} (4.5,0);
	\draw[-,gray] (4.5,0) to node[below]{$12.5$} (5.75,0);
	\draw[-,gray] (6.25,0) to node[below]{$12.5$} (7.5,0);
	\draw[-,gray] (7.5,0) to node[below]{$12.5$} (8.75,0);
	\draw[-,gray] (9.25,0) to node[below]{$12.5$} (10.5,0);
	\draw[-,\myred,thick] (10.5,0) to node[below]{$25$} (13,0);
	\draw[-,gray] (13,0) to node[below]{$12.5$} (14.75,0);
	\draw[-,gray] (21.25,0) to node[below]{$12.5$} (23,0);
	
	\node[state,label=above:{$s$}] (s) at (0,0) {};
	\node[state,label=above:{\textcolor{\myblue}{$3.5$}}] (v1) at (1.5,0) {};
	\node[state,label=above:{\textcolor{\myblue}{$3.5$}}] (v2) at (4.5,0) {};
	\node[state,label=above:{\textcolor{\myblue}{$3.5$}}] (v3) at (7.5,0) {};
	\node[state,label=above:{\textcolor{\myblue}{$3.5$}}] (v4) at (10.5,0) {};
	\node[state,label=above:{\textcolor{\myblue}{$\infty$}}] (v5) at (13,0) {};
	\node[state,label=above:{\textcolor{\myblue}{$67.5$}}] (v6) at (23,0) {};
	\node[label=above:{\makebox[\widthof{$t=00000$}][l]{$t=12.5$}}] (t) at (-1,0) {};
\end{tikzpicture}

\vspace{0.3cm}
\begin{tikzpicture}[xscale=0.5,shorten > = 0pt]
	\draw[-,gray] (0.3,0) to node[below]{$15$} (v1);
	\draw[-,\myred,thick] (v1) to node[below]{$30$} (v2);
	\draw[-,\myred,thick] (v2) to node[below]{$30$} (v3);
	\draw[-,\myred,thick] (7.5,0) to node[below]{$30$} (10.5,0);
	\draw[thick] (10.5,0) to (13,0);
	\draw[-,gray] (13,0) to node[below]{$15$} (14.5,0);
	\draw[-,gray] (21.5,0) to node[below]{$15$} (23,0);
	
	\node[state,label=above:{$s$}] (s) at (0,0) {};
	\node[state,label=above:{\textcolor{\myblue}{$1$}}] (v1) at (1.5,0) {};
	\node[state,label=above:{\textcolor{\myblue}{$1$}}] (v2) at (4.5,0) {};
	\node[state,label=above:{\textcolor{\myblue}{$1$}}] (v3) at (7.5,0) {};
	\node[state,label=above:{\textcolor{\myblue}{$1$}}] (v4) at (10.5,0) {};
	\node[state,label=above:{\textcolor{\myblue}{$\infty$}}] (v5) at (13,0) {};
	\node[state,label=above:{\textcolor{\myblue}{$65$}}] (v6) at (23,0) {};
	\node[label=above:{\makebox[\widthof{$t=00000$}][l]{$t=15$}}] (t) at (-1,0) {};
\end{tikzpicture}

\vspace{0.3cm}
\begin{tikzpicture}[xscale=0.5,shorten > = 0pt]
	\draw[-,thick,\myred] (0,0) to node[below]{$16$} (v1);
	\draw[thick] (v1) to (v2);
	\draw[thick] (v2) to (v3);
	\draw[thick] (7.5,0) to(10.5,0);
	\draw[thick] (10.5,0) to (13,0);
	\draw[-,gray] (13,0) to node[below]{$16$} (14.5,0);
	\draw[-,gray] (21.5,0) to node[below]{$16$} (23,0);
	
	\node[state,label=above:{$s$}] (s) at (0,0) {};
	\node[state] (v1) at (1.5,0) {};
	\node[state] (v2) at (4.5,0) {};
	\node[state] (v3) at (7.5,0) {};
	\node[state] (v4) at (10.5,0) {};
	\node[state,label=above:{\textcolor{\myblue}{$\infty$}}] (v5) at (13,0) {};
	\node[state,label=above:{\textcolor{\myblue}{$64$}}] (v6) at (23,0) {};
	\node[label=above:{\makebox[\widthof{$t=00000$}][l]{$t=16$}}] (t) at (-1,0) {};
\end{tikzpicture}

\vspace{0.3cm}
\begin{tikzpicture}[xscale=0.5,shorten > = 0pt]
	\draw[thick] (0,0) to (v1);
	\draw[thick] (v1) to (v2);
	\draw[thick] (v2) to (v3);
	\draw[thick] (7.5,0) to(10.5,0);
	\draw[thick] (10.5,0) to (13,0);
	\draw[-,gray] (13,0) to node[below]{$16$} (14.5,0);
	\draw[-,gray] (15,0) to node[below]{$80$} (23,0);
	
	\node[state,label=above:{$s$}] (s) at (0,0) {};
	\node[state] (v1) at (1.5,0) {};
	\node[state] (v2) at (4.5,0) {};
	\node[state] (v3) at (7.5,0) {};
	\node[state] (v4) at (10.5,0) {};
	\node[state,label=above:{\textcolor{\myblue}{$\infty$}}] (v5) at (13,0) {};
	\node[state,label=above:{\textcolor{\myred}{$0$}}] (v6) at (23,0) {};
	\node[label=above:{\makebox[\widthof{$t=00000$}][l]{$t=80$}}] (t) at (-1,0) {};
\end{tikzpicture}

\vspace{0.3cm}
\begin{tikzpicture}[xscale=0.5,shorten > = 0pt]
	\node[state,label=above:{$s$}] (s) at (0,0) {};
	\node[state] (v1) at (1.5,0) {};
	\node[state] (v2) at (4.5,0) {};
	\node[state] (v3) at (7.5,0) {};
	\node[state] (v4) at (10.5,0) {};
	\node[state] (v5) at (13,0) {};
	\node[state] (v6) at (23,0) {};
	\node[label=above:{\makebox[\widthof{$t=00000$}][l]{$T_{1.6,b}$}}] (t) at (-1,0) {};
	
	\draw[thick] (s) to (v5);
\end{tikzpicture}
\end{subfigure}
\caption{Illustrations of \Cref{exmp:cgrt-not-monotone}. The primal--dual subroutine with parameter~$\lambda=1.6$ and terminal vertex~$b$.
\label{fig:exmp-cgrt-not-monotone-1-6}}
\end{figure}

\begin{figure}[t]
\scriptsize
\begin{subfigure}[t]{\textwidth}
\scriptsize
\begin{tikzpicture}[xscale=0.5,shorten > = 0pt]
	\node[state,label=above:{$s$}] (s) at (0,0) {};
	\node[state,label=above:{\textcolor{\myblue}{$20$}}] (v1) at (1.5,0) {};
	\node[state,label=above:{\textcolor{\myblue}{$20$}}] (v2) at (4.5,0) {};
	\node[state,label=above:{\textcolor{\myblue}{$20$}}] (v3) at (7.5,0) {};
	\node[state,label=above:{\textcolor{\myblue}{$20$}}] (v4) at (10.5,0) {};
	\node[state,label=above:{\textcolor{\myblue}{$\infty$}}] (v5) at (13,0) {};
	\node[state,label=above:{\textcolor{\myblue}{$100$}}] (v6) at (23,0) {};
	\node[label=above:{\makebox[\widthof{$t=00000$}][l]{$t=0$}}] (t) at (-1,0) {};
\end{tikzpicture}

\vspace{0.3cm}
\begin{tikzpicture}[xscale=0.5,shorten > = 0pt]
	\draw[-,gray] (0.25,0) to node[below]{$12.5$} (1.5,0);
	\draw[-,gray] (1.5,0) to node[below]{$12.5$} (2.75,0);
	\draw[-,gray] (3.25,0) to node[below]{$12.5$} (4.5,0);
	\draw[-,gray] (4.5,0) to node[below]{$12.5$} (5.75,0);
	\draw[-,gray] (6.25,0) to node[below]{$12.5$} (7.5,0);
	\draw[-,gray] (7.5,0) to node[below]{$12.5$} (8.75,0);
	\draw[-,gray] (9.25,0) to node[below]{$12.5$} (10.5,0);
	\draw[-,\myred,thick] (10.5,0) to node[below]{$25$} (13,0);
	\draw[-,gray] (13,0) to node[below]{$12.5$} (14.75,0);
	\draw[-,gray] (21.25,0) to node[below]{$12.5$} (23,0);
	
	\node[state,label=above:{$s$}] (s) at (0,0) {};
	\node[state,label=above:{\textcolor{\myblue}{$7.5$}}] (v1) at (1.5,0) {};
	\node[state,label=above:{\textcolor{\myblue}{$7.5$}}] (v2) at (4.5,0) {};
	\node[state,label=above:{\textcolor{\myblue}{$7.5$}}] (v3) at (7.5,0) {};
	\node[state,label=above:{\textcolor{\myblue}{$7.5$}}] (v4) at (10.5,0) {};
	\node[state,label=above:{\textcolor{\myblue}{$\infty$}}] (v5) at (13,0) {};
	\node[state,label=above:{\textcolor{\myblue}{$87.5$}}] (v6) at (23,0) {};
	\node[label=above:{\makebox[\widthof{$t=00000$}][l]{$t=12.5$}}] (t) at (-1,0) {};
\end{tikzpicture}

\vspace{0.3cm}
\begin{tikzpicture}[xscale=0.5,shorten > = 0pt]
	\draw[-,gray] (0.3,0) to node[below]{$15$} (v1);
	\draw[-,\myred,thick] (v1) to node[below]{$30$} (v2);
	\draw[-,\myred,thick] (v2) to node[below]{$30$} (v3);
	\draw[-,\myred,thick] (7.5,0) to node[below]{$30$} (10.5,0);
	\draw[thick] (10.5,0) to (13,0);
	\draw[-,gray] (13,0) to node[below]{$15$} (14.5,0);
	\draw[-,gray] (21.5,0) to node[below]{$15$} (23,0);
	
	\node[state,label=above:{$s$}] (s) at (0,0) {};
	\node[state,label=above:{\textcolor{\myblue}{$5$}}] (v1) at (1.5,0) {};
	\node[state,label=above:{\textcolor{\myblue}{$5$}}] (v2) at (4.5,0) {};
	\node[state,label=above:{\textcolor{\myblue}{$5$}}] (v3) at (7.5,0) {};
	\node[state,label=above:{\textcolor{\myblue}{$5$}}] (v4) at (10.5,0) {};
	\node[state,label=above:{\textcolor{\myblue}{$\infty$}}] (v5) at (13,0) {};
	\node[state,label=above:{\textcolor{\myblue}{$85$}}] (v6) at (23,0) {};
	\node[label=above:{\makebox[\widthof{$t=00000$}][l]{$t=15$}}] (t) at (-1,0) {};
\end{tikzpicture}

\vspace{0.3cm}
\begin{tikzpicture}[xscale=0.5,shorten > = 0pt]
	\draw[-,thick,\myred] (0,0) to node[below]{$16$} (v1);
	\draw[thick] (v1) to (v2);
	\draw[thick] (v2) to (v3);
	\draw[thick] (7.5,0) to(10.5,0);
	\draw[thick] (10.5,0) to (13,0);
	\draw[-,gray] (13,0) to node[below]{$16$} (14.5,0);
	\draw[-,gray] (21.5,0) to node[below]{$16$} (23,0);
	
	\node[state,label=above:{$s$}] (s) at (0,0) {};
	\node[state] (v1) at (1.5,0) {};
	\node[state] (v2) at (4.5,0) {};
	\node[state] (v3) at (7.5,0) {};
	\node[state] (v4) at (10.5,0) {};
	\node[state,label=above:{\textcolor{\myblue}{$\infty$}}] (v5) at (13,0) {};
	\node[state,label=above:{\textcolor{\myblue}{$84$}}] (v6) at (23,0) {};
	\node[label=above:{\makebox[\widthof{$t=00000$}][l]{$t=16$}}] (t) at (-1,0) {};
\end{tikzpicture}

\vspace{0.3cm}
\begin{tikzpicture}[xscale=0.5,shorten > = 0pt]
	\draw[thick] (0,0) to (v1);
	\draw[thick] (v1) to (v2);
	\draw[thick] (v2) to (v3);
	\draw[thick] (7.5,0) to(10.5,0);
	\draw[thick] (10.5,0) to (13,0);
	\draw[-,thick,\myred] (13,0) to node[below]{$99$} (23,0);
	
	\node[state,label=above:{$s$}] (s) at (0,0) {};
	\node[state] (v1) at (1.5,0) {};
	\node[state] (v2) at (4.5,0) {};
	\node[state] (v3) at (7.5,0) {};
	\node[state] (v4) at (10.5,0) {};
	\node[state,label=above:{\textcolor{\myblue}{$\infty$}}] (v5) at (13,0) {};
	\node[state,label=above:{\textcolor{\myblue}{$16$}}] (v6) at (23,0) {};
	\node[label=above:{\makebox[\widthof{$t=00000$}][l]{$t=84$}}] (t) at (-1,0) {};
\end{tikzpicture}

\vspace{0.3cm}
\begin{tikzpicture}[xscale=0.5,shorten > = 0pt]
	\node[state,label=above:{$s$}] (s) at (0,0) {};
	\node[state] (v1) at (1.5,0) {};
	\node[state] (v2) at (4.5,0) {};
	\node[state] (v3) at (7.5,0) {};
	\node[state] (v4) at (10.5,0) {};
	\node[state] (v5) at (13,0) {};
	\node[state] (v6) at (23,0) {};
	\node[label=above:{\makebox[\widthof{$t=00000$}][l]{$T_{2,b}$}}] (t) at (-1,0) {};
	
	\draw[thick] (s) to (v6);
\end{tikzpicture}
\end{subfigure}
\caption{Illustrations of \Cref{exmp:cgrt-not-monotone}. The primal--dual subroutine with parameter~$\lambda=2$ and terminal vertex~$b$.\label{fig:exmp-cgrt-not-monotone-2}}

\end{figure}

\section{Deferred Proofs of \Cref{sec:unweighted}}

\subsection{Proof of \Cref{lem:u-expected-latency}}
\label{app:u-expected-latency}

\uexpectedlatency*

\begin{proof}
	Let $P=(n_0,n_1,\dots,n_l)$ be a $0$--$n$-path $P$ in $H$.
	First, note that for $k=0$, we have $\E[\lat_{alpha,s}(\seq_P)]=0=\psi_P(0)$.
	Thus, consider the case $k\in \setn$.
	Let $j\in [l]$ be such that $n_{j-1}< k \leq n_{j}$.
	Let~$v$ be the $k$-th distinct vertex in~$V\setminus \{ s \}$ visited by~$\seq_P$.
	The $\alpha$-latency of~$v$ is maximal if~$v$ is first visited during the traversal of~$T_{n_j}$.
	This happens in particular if the trees $T_{n_0},T_{n_1},\dots,T_{n_{j}}$ are nested.
	Thus, the~$\alpha$-latency of~$v$ in~$\seq_P$ can be bounded from above by the cost of all tours~$\seq(T_{n_i})$ with $i\in[j-1]_0$ and the additional expected latency of vertex~$v$ in the tour~$\seq(T_{n_j})$.
	Since the direction of the traversal of $\seq(T_{n_{j}})$ is picked at random, we have
	\begin{align*}
		\E[\lat_{\alpha,v}(\seq(T_{n_j}))]
		=\frac{1}{2}\lat_{{\alpha},v}(\seq_f(T_{n_j}))+\frac{1}{2}\lat_{\alpha,v}(\seq_b(T_{n_j}))
		\leq c(T_{n_j}),
	\end{align*}
	where $\seq_f(T_{n_j})$ and $\seq_b(T_{n_j})$ are the forward and backward traversals of tour $\seq(T_{n_j})$.
	In total, we conclude that
	\begin{align*}
		\E[\lat_{{\alpha},v}(\seq_P)]
		\leq c(T_{n_{j}})+(1+\alpha)\sum_{i=0}^{j-1}c(T_{n_i})
		= \psi_P(k).
	\end{align*}
 This completes the proof.
\end{proof}

\subsection{Proof of \Cref{lem:u-path-pi}}
\label{app:u-path-pi}

\upathpi*

\begin{proof}
	For some $k\in \setn$, let $j(k)\in[l]$ be such that $n_{j(k)-1}< k \leq n_{j(k)}$.
	Then, we obtain
	\begin{align*}
		\sum_{k=0}^n \psi_{P}(k)
		& = \sum_{j=1}^l \sum_{v\in V_{j}} \left( c(T_{n_j}) + (1+\alpha) \sum_{i=1}^{j-1}c(T_{n_i}) \right)\\
		& = \sum_{j=1}^l (n_j-n_{j-1}) \left(c(T_{n_j})+(1+\alpha)\sum_{i=1}^{j-1}c(T_{n_i})\right)\\
		& = \sum_{j=1}^l (n_j-n_{j-1})c(T_{n_j})+(1+\alpha)(n-n_j)c(T_{n_j})\\
		& = \sum_{j=1}^l \bigl((1+\alpha)n-\alpha n_j-n_{j-1}\bigr)c(T_{n_j})\\
		& = \ell(P).
	\end{align*}
This completes the proof.
\end{proof}

\subsection{Proof of \Cref{lem:u-opt-bound}}
\label{app:u-opt-bound}

\uoptbound*

\begin{proof}
	We first show how to construct a randomized $0$--$n$-path $P$ in $H$ which fulfills the claimed upper bound on its length in expectation.
	To obtain a deterministic algorithm we will then pick a shortest $0$--$n$-path $P^*$ in $H$ which has length at most $\E[\ell(P)]$.
	In this regard, we write $\seq^*(k)$ for an optimal $k$-path for some $k\in \setn$.
	Note that the $\alpha$-\edgec of an optimal $k$-path gives a lower bound on the latency of $v_k$ in the optimal tour~$\seq^*\in\Seq$, i.e., $c_{\alpha}(\seq^*(k))\leq \lat_{{\alpha},v_k}(\seq^*)$ and thus
	\begin{align}
		\label{eq:u-opt-path-opt-tour}
		\sum_{k=0}^n c_{\alpha}(\seq^*(k)) \leq \lat_{\alpha}(\seq^*).
	\end{align}
	Since the trees~$T_k$ are good $k$-trees for all $k\in \setn_0$, we obtain $c(T_k)\leq \frac{2}{1+{\alpha}} c_{\alpha}(\seq^*(k))$.
	Let  $\gamma>1$ be fixed.
	For the construction of the desired path $P$, we set $b=\gamma^U$, where $U$ is a random variable distributed uniformly in $[0,1)$.
	In particular, this yields $b\in [1,\gamma)$.
	We denote by ${\check{\jmath}}\in \Z$ the smallest number such that $c(T_n)\leq \frac{1}{2}b\gamma^{\check{\jmath}}$, and by $\hat{\jmath}\in \Z$ the largest number such that $\min\{c_e :{e\in \delta(s)}\}>\frac{2}{1+\alpha}b\gamma^{\hat{\jmath}}$.
	Then, we define for all $j\in\{{\hat{\jmath}},\dots,{\check{\jmath}}\}$
	\begin{align*}
		n_j \coloneqq \max\left\{k\in\setn_0  : c(T_k)\leq \frac{2}{1+\alpha}b\gamma^j \right\},
	\end{align*}
	i.e., $n_j$ is the largest number of vertices that can be visited by one of the good $k$-trees $T_0, T_1,\dots, T_n$ such that the \edgec of that tree is bounded by $\frac{2}{1+\alpha}b\gamma^j$.
	By the choice of ${\hat{\jmath}},{\check{\jmath}}\in\Z$, we have $n_{\hat{\jmath}}=0$ and $n_{\check{\jmath}}=n$.
	Together with $c(T_0)=0$, the values $n_j$ are well-defined for all $j\in\{{\hat{\jmath}},\dots,{\check{\jmath}}\}$.
	Consider the sequence $n_{\hat{\jmath}},n_{{\hat{\jmath}}+1},\dots,n_{\check{\jmath}}$.
	This sequence is non-decreasing.
	Further, we may assume without loss of generality, that it is strictly increasing, as otherwise we simply choose its inclusion-wise maximal, strictly increasing subsequence.
	Thus, the sequence corresponds to a randomized $0$--$n$-path $P=(n_{\hat{\jmath}},n_{{\hat{\jmath}}+1},\dots,n_{\check{\jmath}})$ in $H$. 
	We continue by showing that the expected length of this path is bounded by $\E[\ell(P)]\leq 2\frac{\gamma+\alpha}{(1+\alpha)\ln\gamma} \lat_{\alpha}(\seq^*)$.
	For $k\in \setn$, we set $j\in\{{\hat{\jmath}},\dots,{\check{\jmath}}\}$ and $d\in[1,\gamma)$ such that $c_{{\vec{\alpha}}}(\seq^*(k))=d\gamma^j$.
	We briefly argue that two such values always exist.
	First, note that with~$k>0$ we have 
	\begin{align*}
		c_{\alpha}\big(\seq^*(k)\big)\geq \min\{c_e:{e\in \delta(s)}\}>\frac{2}{1+\alpha}b\gamma^{\hat{\jmath}}\geq \gamma^{\hat{\jmath}}.
	\end{align*}
	Additionally, we can give an upper bound on~$c_{\alpha}\big(\seq^*(k)\big)$ by
	\begin{align*}
		c_{\alpha}\big(\seq^*(k)\big)\leq 2 c(T_n)\leq b\gamma^{\check{\jmath}},
	\end{align*}
	where~$b\in [1,\gamma)$.
	Thus, assuming $c_{\alpha}\big(\seq^*(k)\big)=d\gamma^j$ for some $d\in[1,\gamma)$ is possible by our choices of ${\hat{\jmath}}$ and ${\check{\jmath}}$.
    This yields the upper bound on $c(T_{k})$
    \begin{align*}
        c(T_{k})
        &\leq \frac{2}{1+\alpha}c_{\alpha}(\seq^*(k))
        \leq \frac{2}{1+\alpha}d \gamma^j.
    \end{align*}
    We distinguish two cases based on the relationship of the values~$d$ and~$b$.
    
    \newparagraph{Case 1: $d\leq b$}
    With $c_{\alpha}(\seq^*(k))=d\gamma^j$ and $d\leq b$ we obtain
    \begin{align*}
        c(T_{k})
        \leq & \frac{2}{1+\alpha}d\gamma^j
        \leq  \frac{2}{1+\alpha}b\gamma^j.
    \end{align*}
    In particular, this shows that $n_j\geq k$
     and we can give an upper bound on $\psi_P(k)$ by
    \begin{align*}
        \psi_P(k) 
        &\leq c(T_{n_j})+(1+\alpha)\sum_{i=1}^{j-1}c(T_{n_i})\\
        &\leq c(T_{n_j})+(1+\alpha)\sum_{i=-\infty}^{j-1}c(T_{n_i})\\
        &\leq \frac{2b\gamma^j}{1+\alpha}+(1+\alpha)\sum_{i=-\infty}^{j-1} \frac{2b\gamma^i}{1+\alpha}\\
        &= \frac{2b\gamma^j}{1+\alpha} + \frac{1+\alpha}{1+\alpha}\cdot\frac{2b\gamma^j}{\gamma-1}\\
        &= \frac{2b\gamma^j(\gamma+\alpha)}{(1+\alpha)(\gamma-1)}.
    \end{align*}
    
    \newparagraph{Case 2: $d>b$}
    With $d<\gamma$ and $b\geq 1$ we have $d<\gamma\leq b\gamma$.
    Together with $c_{\alpha}(\seq^*(k))=d\gamma^j$ we obtain
    \begin{align*}
        c(T_{k})
        \leq & \frac{2}{1+\alpha}d\gamma^j
        \leq  \frac{2}{1+\alpha}b\gamma^{j+1}.
    \end{align*}
    In particular, this shows that $n_{j+1}\geq k$
     and we can give an upper bound on $\psi_P(k)$ by
    \begin{align*}
        \psi_P(k) 
        &\leq c(T_{n_{j+1}})+(1+\alpha)\sum_{i=1}^{j}c(T_{n_i})\\
        &\leq c(T_{n_{j+1}})+(1+\alpha)\sum_{i=-\infty}^{j}c(T_{n_i})\\
        &\leq \frac{2b\gamma^{j+1}}{1+\alpha}+(1+\alpha)\sum_{i=-\infty}^{j} \frac{2b\gamma^i}{1+\alpha}\\
        &= \frac{2b\gamma^{j+1}}{1+\alpha} + \frac{1+\alpha}{1+\alpha}\cdot\frac{2b\gamma^{j+1}}{\gamma-1}\\
        &= \frac{2b\gamma^{j+1}(\gamma+\alpha)}{(1+\alpha)(\gamma-1)}.
    \end{align*}
    
    Case 1
	arises whenever $U\in[\log_\gamma d,1]$ while 
	Case 2
	arises whenever $U\in[0,\log_\gamma d)$.
    Taking the expectation over the random variable $U$ yields
    \begin{align*}
        \E_U[\psi_P(k)]
        \leq & \int_{\log_\gamma d}^1 \frac{2b\gamma^j(\gamma+\alpha)}{(1+\alpha)(\gamma-1)} \,\mathrm{d}U \\
        & + \int_0^{\log_\gamma d} \frac{2b\gamma^{j+1}(\gamma+\alpha)}{(1+\alpha)(\gamma-1)} \,\mathrm{d}U \\
        =& \frac{2b\gamma^j(\gamma+\alpha)}{(1+\alpha)(\gamma-1)}
        \Biggl(\int_{\log_\gamma d}^1 \gamma^U \,\mathrm{d}U + \gamma \int_0^{\log_\gamma d} \gamma^U \,\mathrm{d}U\Biggr) \\
        =& \frac{2b\gamma^j(\gamma+\alpha)}{(1+\alpha)(\gamma-1)}
        \Biggl(\frac{\gamma-d}{\ln \gamma}+\gamma\frac{d-1}{\ln \gamma} \Biggr)\\
        =& 2d\gamma^j\frac{\gamma+\alpha}{(1+\alpha)\ln\gamma}\\
        =& 2c_{{\vec{\alpha}}}(\seq^*(k))\frac{\gamma+\alpha}{(1+\alpha)\ln\gamma}.
    \end{align*}
    With \Cref{lem:u-path-pi} and \eqref{eq:u-opt-path-opt-tour},
    we obtain
    \begin{align*}
        \frac{\E[\ell(P)]}{\lat_{\alpha}(\seq^*)}
        = \frac{\E\bigl[\sum_{k=0}^n\psi_P(k)\bigr]}{\lat_{\alpha}(\seq^*)}
        \leq \frac{2\frac{\gamma+\alpha}{(1+\alpha)\ln\gamma}\sum_{k=1}^n c_{\alpha}(\seq^*(k))}{\lat_{\vec{\alpha}}(\seq^*)}
        = 2\frac{\gamma+\alpha}{(1+\alpha)\ln\gamma}.
    \end{align*}
    Thus, the randomized $0$--$n$-path $P$ has expected length of at most $\E[\ell(P)]\leq 2\frac{\gamma+\alpha}{(1+\alpha)\ln\gamma} \lat_{\alpha}(\seq^*)$.
    In particular, a shortest $0$--$n$-path $P^*$ has length at most $\ell(P^*)\leq \E[\ell(P)]\leq 2\frac{\gamma+\alpha}{(1+\alpha)\ln\gamma} \lat_{\alpha}(\seq^*)$.
\end{proof}

\subsection{Proof of \Cref{lem:min-gamma}}
\label{app:min-gamma}

\mingamma*

\begin{proof}
    We first note that
    $\lim_{\gamma \to 1} h(\gamma) = \infty$ and $\lim_{\gamma \to \infty} h(\gamma) = \infty$.
    Thus, the minimum is attained for $\gamma^* \in (1,\infty)$.
    The first derivative of $h$ is given by
    \begin{align*}
        h'(\gamma) &=
        2\frac{(1+\alpha)\ln(\gamma) - (\gamma+ \alpha)\frac{1 + \alpha}{\gamma}}{(1+\alpha)^2\ln^2(\gamma)}.
    \end{align*}
    We then have $h'(\gamma^*) = 0$ which yields the equation
    \begin{align*}
        \gamma^*\ln(\gamma^*)(1+\alpha) &= (\gamma^*+ \alpha)(1 + \alpha)
        \intertext{which is equivalent to}
        \gamma^*\ln(\gamma^*)&= (\gamma^*+ \alpha).
    \end{align*}
    We thus obtain $\gamma^*(\ln(\gamma^*)-1) = \alpha$.
    Substituting $\gamma^* = \mathrm{e}^z$, we obtain the equation $\mathrm{e}^z(z-1) = \alpha$ which yields $\mathrm{e}^{z-1}(z-1) = \frac{\alpha}{\mathrm{e}}$.
    The solution to the latter equation is given by $z-1 = W(\frac{\alpha}{\mathrm{e}})$. Substituting back, we obtain
    \begin{align*}
        \gamma^* = \mathrm{e}^z = \mathrm{e}^{1 + W(\alpha/\mathrm{e})},
    \end{align*}
    so that the minimum of $h$ is attained at the claimed value. We obtain
    \begin{align*}
        h(\gamma^*) = 2 \frac{\alpha + \mathrm{e}^{1+W(\alpha/\mathrm{e})}}{(1+\alpha)(1+W({\alpha}/{\mathrm{e}}))}.
    \end{align*}
    For $\alpha= 0$, we obtain $h(\gamma^*) = \frac{2\mathrm{e}}{1+\alpha}$.
    For $\alpha > 0$ we can simplify this expression using the identity $W(\frac{\alpha}{\mathrm{e}}) \mathrm{e}^{W(\alpha/\mathrm{e})} = \frac{\alpha}{\mathrm{e}}$
    \begin{align*}
        h(\gamma^*) &= 2 \frac{\alpha + \mathrm{e}\cdot \mathrm{e}^{W({\alpha}/{\mathrm{e}})}}{(1+\alpha)(1+W({\alpha}/{\mathrm{e}}))}\\
        &=2 \frac{\alpha + \mathrm{e}\frac{\alpha/\mathrm{e}}{W({\alpha}/{\mathrm{e}})}}{(1+\alpha)(1+W({\alpha}/{\mathrm{e}}))}\\
        &=2 \frac{\alpha(1+ \frac{1}{W(\alpha/\mathrm{e})})}{(1+\alpha)(1+W({\alpha}/{\mathrm{e}}))} \\
        &= 2 \frac{\alpha}{(1+\alpha)W({\alpha}/{\mathrm{e}})}.
    \end{align*}
    Overall, we obtain 
    \begin{align*}
        h(\gamma^*)=
        \begin{cases}
            2\mathrm{e} & \text{ if } \alpha = \alpha = 0,\\
            2 \frac{\alpha}{(1+\alpha)W({\alpha}/{\mathrm{e}})} & \text{ otherwise.}
        \end{cases}    
    \end{align*}
\end{proof}

\subsection{Proof of \Cref{lem:u-factor-with-phantom}}
\label{app:lem:u-factor-with-phantom}

\lemufactorwithphantom*

\begin{proof}
	Together, \Cref{lem:u-expected-latency,lem:u-path-pi,lem:u-opt-bound} combined with \eqref{eq:u-expected-latency} and \Cref{lem:min-gamma} imply
	\begin{align*}
		\lat_{\alpha}(\seq_\alg)
		\leq \sum_{k=0}^n \psi_{P^*}(k)
		=\ell(P^*)
		\leq \factor \lat_{\alpha} (\seq^*),
	\end{align*}
	proving the approximation ratio claimed in \Cref{lem:u-factor-with-phantom}.
\end{proof}

\subsection{Proof of \Cref{lem:u-no-phantom-trees}}
\label{app:u-no-phantom-trees}

\unophantomtrees*

\begin{proof}
	Recall that a phantom tree's \edgec is defined as a linear interpolation of two real trees in $\mathcal{T}$.
	Further, we may assume without loss of generality that $c(T_k)\leq c(T_{k+1})$ for all $k\in [n-1]_0$, as we can set $T_k$ to $T_{k+1}$ if $c(T_k)> c(T_{k+1})$.
	Let $P^*$ be the shortest $0$--$n$-path in the auxiliary graph~$H$ computed by the concatenating algorithm on input~$\mathcal{T}'$ and let $a<b<c$ be three consecutive vertices on that path with corresponding trees $T_a, T_b$, and $T_c$.
	Assume that $T_b$ is a phantom tree that is obtained by a linear interpolation of the real trees $T_{b_0}$ and $T_{b_1}$, i.e., $b=(1-\mu)b_0+\mu b_1$ for some~$\mu\in(0,1)$.
	
	In what follows, we show that by setting $b$ to either $\max\{a, b_0\}$ or $\min\{c,b_1\}$, we can construct a $0$--$n$-path $P$ in $H$ with $\ell(P)= \ell(P^*)$.
	Thus, we obtain a $0$--$n$-path with the same length but fewer phantom trees.
	Iterating this procedure then yields the desired $0$--$n$-path in~$H$.
	
	In that sense, recall that by the definition of $H$, the subpath~$((a,b),(b,c))$ of~$P^*$ has length
	\begin{align*}
		\ell_{a,b}+\ell_{b,c}
		&=\big[(1+\alpha)n-\alpha b-a\big]c(T_b)
		+\big[(1+\alpha)n-\alpha c-b\big]c(T_c) \\
		&=\big[(1+\alpha)n-\alpha ((1-\mu)b_0+\mu b_1)-a\big]\big[(1-\mu)c(T_{b_0})+\mu c(T_{b_1})\big]\\
		&\phantom{ = [} +\big[(1+\alpha)n-\alpha c-((1-\mu)b_0+\mu b_1)\big]c(T_c).
	\end{align*}
	This function is quadratic in $\mu$, where $\mu^2$ has the coefficient 
	\begin{align*}
		-\alpha(b_1-b_0)(c(T_{b_1})-c(T_{b_0}))\leq 0.
	\end{align*}
	Hence, a local minimum is attained when replacing $b$ by either $\max\{a,b_0\}<b$ or $\min\{c,b_1\}>b$.
	This gives a new path $P$ where we delete any self-loop at $a$ or $c$ if $b$ was set to either $a$ or $c$.
	Then, the new $0$--$n$-path $P$ has length at most $\ell(P^*)$ and contains fewer vertices that correspond to phantom trees.
	We can repeat this process until the new $0$--$n$-path $P$ contains no phantom trees anymore.
\end{proof}

\subsection{Proof of \Cref{thm:u-main}}
\label{app:u-main}

\umain*

\begin{proof}[Proof of \Cref{thm:u-main}]
	Consider the following algorithm:
\begin{enumerate}
	\item Run \Cref{alg:u-phantom-tree} and obtain the set~$\mathcal{T}$.
	\item Denote the set of real trees in~$\mathcal{T}$ by~$\mathcal{R}$.
	\item Run the concatenating algorithm on input~$\mathcal{R}$ and obtain~$\seq_\alg$.
\end{enumerate}
We claim that this algorithm runs in polynomial time and the returned solution~$\seq_\alg$ guarantees an approximation ratio of~$\factor$ proving \Cref{thm:u-main}.

The claimed approximation guarantee of~$\seq_\alg$ follows immediately by \Cref{lem:u-factor-with-phantom,lem:u-no-phantom-trees}.
It remains to prove that the algorithm runs in polynomial time.
To this end, we first analyze the running time of \Cref{alg:u-phantom-tree}.
For every terminal vertex~$t\in V \setminus \{ s \}$, the primal--dual subroutine is executed once with parameter~$\lambda=0$ and terminal vertex~$t$.
By \citeauthor{chaudhuri2003paths} one call of the subroutine runs in time~$\mathcal{O}(n^2)$, so we need time~$\mathcal{O}(n^3)$ to compute all values~$k_{0,t}$.
Afterward, we apply \Cref{lem:u-running-time-lambda} to each pair of a value~$k\in\setn_0$ and a terminal vertex~$t\in V\setminus \{ s \}$.
By the proof of \Cref{lem:u-running-time-lambda}, each call takes time at most~$\mathcal{O}(n^5)$.
Since the remaining steps of \Cref{alg:u-phantom-tree} can be executed in constant time, we obtain an upper bound on the running time by~$\mathcal{O}(n^7)$.
Next, consider the execution of the concatenating algorithm on input~$\mathcal{R}$.
Dijkstra's algorithm gives a trivial upper bound of $\mathcal{O}(n^2)$ for the computation of the shortest $0$--$n$-path.
Afterward, the chosen trees are turned into tours and then concatenated.
Each tour can be computed in linear time by finding an Euler tour as mentioned in \Cref{prop:tree-to-tour}.
In total, the concatenating algorithm on input~$\mathcal{R}$ runs in time~$\mathcal{O}(n^2)$.

Thus, the complete algorithm as stated above runs in polynomial time and computes a solution~$\seq_\alg$ for the \psp with approximation guarantee~$\factor$.
This completes the proof of \Cref{thm:u-main}. 
\end{proof}

\section{Hardness of Approximation}
\label{sec:alpha-hardness}

In this section, we show that the \psp is hard to approximate for $p=1$ and any $\alpha \in [0, 1]$.
In particular, we show the following theorem.

\thmhardness*

The proof of this theorem is split into two parts.
First, we focus on the special case of~$\alpha=0$.
Afterward, we consider all remaining values~$\alpha\in (0,1]$.

\subsection{Hardness of the Expanding Search Problem ($\alpha = 0)$}

For~$\alpha=0$, the \psp with discount factor~$\alpha$ corresponds to the expanding search problem.
\citet{GriesbachHKS26} proved that the weighted version of the expanding search problem, in which the goal is to minimize the weighted sum of latencies and $p=1$, is $\mathsf{NP}$-hard.
In particular, they show this for the $\{0, 1\}$-weighted version, in which each vertex has a weight either $0$ or $1$. They show the following.

\begin{theorem}[Theorem~6.1 in~\cite{GriesbachHKS26}]
\label{thm:hardness-esp}
    If~$\alpha= 0$, there exists a constant~$\eps>0$ such that there is no polynomial-time~$(1+\eps)$-approximation algorithm for the $\{0, 1\}$-weighted \psp with discount factor~$\alpha$, unless~$\mathsf{P}=\mathsf{NP}$.
\end{theorem}

We note that this result also implies the same result for the unweighted version of \psp by replacing each vertex with weight $1$ with a clique of polynomial size, where each edge of the clique has cost $0$. 
This directly gives the desired result for $\alpha=0$ and $p=1$ in our setting.

\subsection{Hardness of the Discounted Graph Search Problem ($\alpha \in (0, 1]$)}

To complete the proof of \Cref{thm:hardness-psp}, it remains to consider values~$\alpha\in(0,1]$ for the discount factor~$\alpha$.
Thus, this subsection is dedicated to proving the following theorem.
We note that in the proof of \Cref{thm:hardness-psp-not-0} we consider the weighted version but give every vertex the same non-negative weight $w > 0$. 
Up to scaling, this is equivalent so the unweighted setting.

\begin{theorem}
	\label{thm:hardness-psp-not-0}
	For every constant~$\alpha\in(0,1]$, there exists a constant~$\eps>0$ such that there is no polynomial-time~$(1+\eps)$-approximation algorithm for the \psp with discount factor~$\alpha$, unless~$\mathsf{P}=\mathsf{NP}$.
\end{theorem}

The proof idea is as follows.
First, we introduce a variant of the traveling salesperson problem with edge costs either~$1$ or~$2$.
In the original version, one is given an undirected complete graph~$G$ on~$n$ vertices where every edge has a cost of either~$1$ or~$2$.
The goal is to find a tour in~$G$ that visits all vertices and has minimal cost.
We call this problem \textsc{TSP}(1,2).
In the variant we consider, we are additionally given a discount factor~$\alpha\in(0,1]$, and we want to find a tour in~$G$ that visits all vertices and has minimal~$\alpha$-cost.
We call this problem~$\alpha$-\textsc{TSP}(1,2).
We show that for~$\alpha$-\textsc{TSP}(1,2), there exists a constant~$\eps>0$ such that there is no polynomial-time~$(1+\eps)$-approximation algorithm, unless~$\mathsf{P}=\mathsf{NP}$.
This hardness of approximation is then used to prove \Cref{thm:hardness-psp-not-0}.

\citet{papadimitriou1993traveling} proved the following result for \textsc{TSP}(1,2).

\begin{theorem}[Theorem 3 in \cite{papadimitriou1993traveling}] \label{thm:tsp12}
    \textsc{TSP(1,2)} is~$\mathsf{MaxSNP}$-hard.
\end{theorem}

Recall from the previous subsection that \citet{arora1998proof} showed that there exists no polynomial-time approximation scheme for any~$\mathsf{MaxSNP}$-hard problem, unless~$\mathsf{P}=\mathsf{NP}$.
Hence, there exists a constant~$\varrho>0$ such that there is no polynomial-time~$(1+\varrho)$-approximation algorithm for \textsc{TSP}(1,2), unless~$\mathsf{P}=\mathsf{NP}$.
We use the hardness result for TSP(1,2) to show that for every constant~$\alpha\in(0,1]$, there exists a constant~$\eps>0$ such that there is no polynomial-time~$(1+\eps)$-approximation algorithm for~$\alpha$-\textsc{TSP(1,2)}, unless~$\mathsf{P}=\mathsf{NP}$.
In particular, we assume that there was a polynomial-time~$(1+\eps)$-approximation algorithm for~$\alpha$-\textsc{TSP(1,2)} for every~$\eps>0$ and show that this either contradicts \Cref{thm:tsp12} or implies~$\mathsf{P}=\mathsf{NP}$.
Note that for~$\alpha=0$, an optimal solution of~$\alpha$-\textsc{TSP(1,2)} is given by a minimum-cost spanning tree of~$G$.
Hence, an optimal solution for~$0$-\textsc{TSP(1,2)} can be computed in polynomial time.

\begin{lemma}
	\label{lem:hardness-alpha-tsp}
	For every constant~$\alpha\in(0,1]$, there exists a constant~$\eps>0$ such that there is no polynomial-time~$(1+\eps)$-approximation algorithm for~$\alpha$-\textsc{TSP(1,2)}, unless~$\mathsf{P}=\mathsf{NP}$.
\end{lemma}
\begin{proof}
	Let~$\alpha\in(0,1]$ be fixed.
	The basic idea of the proof is as follows.
	Let~$\varrho>0$ be a constant such that there is no polynomial-time~$(1+\varrho)$-approximation algorithm for \textsc{TSP}(1,2), unless~$\mathsf{P}=\mathsf{NP}$.
	This constant exists due to \cite{arora1998proof} and \cite{papadimitriou1993traveling}.
	We assume that for every~$\eps>0$, there exists a polynomial-time~$(1+\eps)$-approximation algorithm for~$\alpha$-\textsc{TSP}(1,2) and conclude that this yields a~polynomial-time~$\gamma$-approximation algorithm for \textsc{TSP}(1,2) with~$\gamma<(1+\varrho)$.
	Given a \textsc{TSP}(1,2) instance, we construct an instance for~$\alpha$-\textsc{TSP}(1,2) as follows.
	
	\paragraph{Construction of the~$\alpha$-\textsc{TSP}(1,2) Instance.}
	Let~$I_{\textsc{tsp}}=(G=(V,E),(c_e)_{e\in E})$ be an instance of \textsc{TSP}(1,2).
	Without loss of generality, we assume that the number of vertices~$|V|=n$ is even.
	For fixed~$\alpha\in(0,1]$, we construct the~$\alpha$-\textsc{TSP}(1,2) instance~$I_{\alphatsp}=(G'=(V',E'),(c_e)_{e\in E'})$ as follows.
	The graph~$G'=(V',E')$ is obtained by two copies of~$G$, referred to as~$G_a$ and~$G_b$.
	Let~$v_a\in V'$ be a vertex of~$G'$ in copy~$G_a$, and let~$v_b\in V'$ be the corresponding copy in~$G_b$.
	We connect vertices~$v_a$ and~$v_b$ by a path consisting of~$k+1\in \N$ edges of cost~$1$ and~$k$ intermediate vertices where we set~$k\coloneqq \left\lceil \frac{3}{\alpha} \right\rceil$.
	Every edge~$e\in E'$ is equipped with the same discount factor~$\alpha$.
	An illustration of the obtained~$\alpha$-\textsc{TSP}(1,2) instance is shown in \Cref{fig:construction-tsp}.
	
	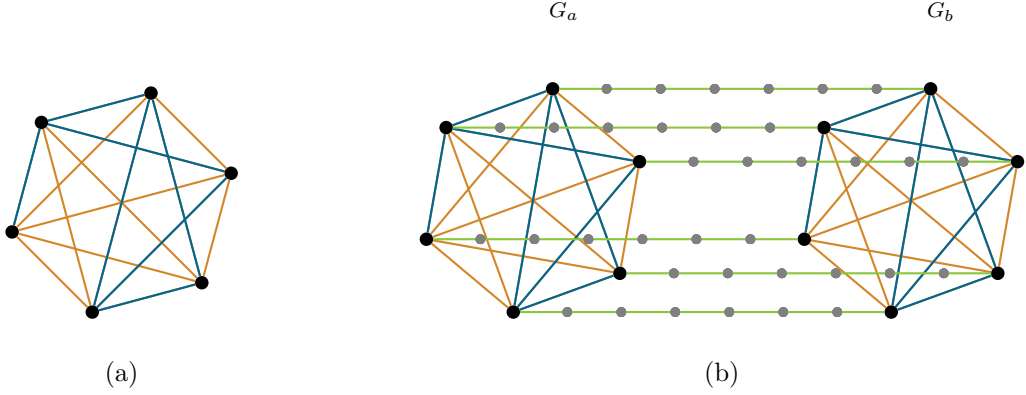
\begin{figure}[t]
	\scriptsize
	\centering
	\begin{subfigure}[b]{.3\textwidth}
	\begin{center}
	\begin{tikzpicture}[xscale=0.5,yscale=0.5,shorten > = 0pt]
		\foreach \i in {1,...,6} {
	    \node[state] (n\i) at ({(360/6) * (\i - 1) +15}:3) {};
		}
		\foreach \i in {1,...,6} {
		    \foreach \j in {\i,...,6} {
		        \ifnum\i<\j
		            \draw[thick, \myyellow] (n\i) -- (n\j);
		        \fi
		    }
		}

		\draw[thick, \myblue] (n1) -- (n3);
		\draw[thick, \myblue] (n2) -- (n3);
		\draw[thick, \myblue] (n4) -- (n3);
		\draw[thick, \myblue] (n5) -- (n2);
		\draw[thick, \myblue] (n5) -- (n1);
		\draw[thick, \myblue] (n6) -- (n2);
		\draw[thick, \myblue] (n6) -- (n5);
	\end{tikzpicture}
	\end{center}
	\caption{}
	\label{fig:construction-tsp-1}
	\end{subfigure}
	\begin{subfigure}[b]{.65\textwidth}
	\begin{center}
	\begin{tikzpicture}[xscale=0.5,yscale=0.5,shorten > = 0pt]
		\begin{scope}[xshift=10cm]
			\foreach \i in {1,...,6} {
		    \node[state] (m\i) at ({(360/6) * (\i - 1) +20}:3) {};
			}
			\foreach \i in {1,...,6} {
			    \foreach \j in {\i,...,6} {
			        \ifnum\i<\j
			            \draw[thick, \myyellow] (m\i) -- (m\j);
			        \fi
			    }
			}
			\draw[thick, \myblue] (m1) -- (m3);
			\draw[thick, \myblue] (m2) -- (m3);
			\draw[thick, \myblue] (m4) -- (m3);
			\draw[thick, \myblue] (m5) -- (m2);
			\draw[thick, \myblue] (m5) -- (m1);
			\draw[thick, \myblue] (m6) -- (m2);
			\draw[thick, \myblue] (m6) -- (m5);
			\node[label=right:{$G_b$}] (G1) at(0,5) {};
		\end{scope}
		\foreach \i in {1,...,6} {
	    \node[state] (n\i) at ({(360/6) * (\i - 1) +20}:3) {};
		}
		\foreach \i in {1,...,6} {
		    \foreach \j in {\i,...,6} {
		        \ifnum\i<\j
		            \draw[thick, \myyellow] (n\i) -- (n\j);
		        \fi
		    }
		}
		\draw[thick, \myblue] (n1) -- (n3);
		\draw[thick, \myblue] (n2) -- (n3);
		\draw[thick, \myblue] (n4) -- (n3);
		\draw[thick, \myblue] (n5) -- (n2);
		\draw[thick, \myblue] (n5) -- (n1);
		\draw[thick, \myblue] (n6) -- (n2);
		\draw[thick, \myblue] (n6) -- (n5);
		\node[label=right:{$G_a$}] (G1) at (0,5) {};
		
		\foreach \i in {1,...,6} {
		\draw[thick, \mygreen] (n\i) -- (m\i);
		
		\foreach \j in {1,...,6} {
		\foreach \i in {1,...,6} {
	    \node[circle,draw=gray,thick,fill=gray,minimum size=1mm,inner sep=0pt] (a\i) at ($ (n\j)!{\i/7}!(m\j)$) {};
	  	}
	  	}
		}
	\end{tikzpicture}
	\end{center}
	\caption{}
	\label{fig:construction-tsp-2}
	\end{subfigure}
	\caption{(a)~Instance~$I_{\textsc{tsp}}$ for \textsc{TSP}(1,2).
	\textcolor{\myblue}{Blue} edges have cost~$1$ and \textcolor{\myyellow}{orange} edges have cost~$2$.
	(b)~Instance~$I_{\alphatsp}$ for~$\alpha$-TSP(1,2) for~$\smash{\alpha=\frac{1}{2}}$ constructed from~$I_{\textsc{tsp}}$.
	On each \textcolor{\mygreen}{green} path from~$G_a$ to~$G_b$, there are~$k=6$ intermediate vertices shown in gray.
	\textcolor{\myblue}{Blue} and \textcolor{\mygreen}{green} edges have cost~$1$, and \textcolor{\myyellow}{orange} edges have cost~$2$.}
	\label{fig:construction-tsp}
	\end{figure}
	
	Assume the polynomial-time~$(1+\eps)$-approximation algorithm~$\algalpha$ for~$\alpha$-\textsc{TSP(1,2)} returns solution~$\seq_{\alphatsp}$ when applied on instance~$I_{\alphatsp}$.
	We want to use~$\seq_{\alphatsp}$ to construct a solution~$\seq_{\textsc{tsp}}$ for \textsc{TSP}(1,2) on instance~$I_{\textsc{tsp}}$.
	In particular, we want to transform~$\seq_{\alphatsp}$ such that the edges in~$G_a$ and~$G_b$ correspond to a simple TSP-tour on~$G$ and every edge in one of the copies~$G_a$ or~$G_b$ is followed by a~$v_a$--$v_b$-path (or~$v_b$--$v_a$-path) to the other component.
	In the following, we show that this transformation of~$\seq_{\alphatsp}$ can be done in polynomial time and does not increase the~$\alpha$-cost of the tour.
	
	To this end, note that it is without loss of generality to assume that all edges in~$\seq_{\alphatsp}$ appear at most twice in~$\seq_{\alphatsp}$ as otherwise the tour can be shortened.
	Next, we show that all edges of a~$v_a$--$v_b$-path connecting the two copies~$G_a$ and~$G_b$ can be assumed to be consecutive in~$\seq_{\alphatsp}$.	

	\begin{claim}\label{cl:full-paths}
	    Without loss of generality, we can assume that all edges of a~$v_a$--$v_b$-path are consecutive in~$\seq_{\alphatsp}$.
	\end{claim}
	\begin{proof}[Proof of \Cref{cl:full-paths}.]
		Assume there is a~$v_a$--$v_b$-path~$P$ in~$G'$ whose edges are not consecutive in~$\seq_{\alphatsp}$.
	    Let~$v_a,v_1,v_2,\dots,v_k,v_b$ be the vertices on that~$v_a$--$v_b$-path.
	    Since every edge of~$G'$ is contained at most twice in~$\seq_{\alphatsp}$ and all intermediate vertices~$v_1,v_2,\dots,v_k$ are visited by~$\seq_{\alphatsp}$, one of the following two cases holds:
	    (i)~the sequence~$\seq_{\alphatsp}$ contains the two subpaths~$(v_a,v_1,\dots,v_l,\dots,v_1,v_a)$ and~$(v_b,v_k,\dots,v_{l},\dots,v_k,v_b)$ for some~$l\in [k]_0$ (where we have~$v_0=a$ and~$v_{k+1}=v_b$),
	    or~(ii)~the sequence~$\seq_{\alphatsp}$ contains the two subpaths~$(v_a,v_1,\dots,v_l,\dots,v_1,v_a)$ and~$(v_b,v_k,\dots,v_{l+1},\dots,v_k,v_b)$ for some~$l\in [k]_0$.
	    It is easy to see that if case~(i) applies, we can delete both appearances of the edge~$(v_l,v_{l+1})$ in~$\seq_{\alphatsp}$ without increasing the~$\alpha$-cost of the entire sequence and thus case~(ii) applies.
	    For case~(ii), it is without loss of generality to assume that~$l=k$, i.e., the missing edge is~$e=(v_k,v_b)$.
	    To see this, note that all edges on the~$v_a$--$v_b$-path~$P$ have cost~$1$ and hence we can exchange the subpaths~$(v_a,v_1,\dots,v_l,\dots,v_1,v_a)$ and~$(v_b,v_k,\dots,v_{l},\dots,v_k,v_b)$ by the two subpaths~$(v_a,v_1,\dots,v_k,\dots,v_1,v_a)$ and~$(v_b)$ without changing the~$\alpha$-cost of the entire sequence.
	    However, since the number of vertices~$n$ in~$G_a$ and~$G_b$ is even, there must exist a second pair of vertices~$v_a'$ and~$v_b'$ such that the edges of the~$v_a'$--$v_b'$-path~$P'$ are also not consecutive in~$\seq_{\alphatsp}$, i.e., edge~$e'=(v_k',v_b')$ is also not contained in~$\seq_{\alphatsp}$.
	    We construct a new tour~$\seq$ as follows.
	    Let~$M_{\alpha}$ be the multiset of edges that appear in~$\seq_{\alphatsp}$.
	    We define a new multiset of edges~$M$ by deleting one copy of each edge on paths~$P$ and~$P'$ and adding the two missing edges~$(v_k,v_b)$ and~$(v_k',v_b')$.
	    We also add the two edges~$(v_a,v_a')$ and~$(v_b,v_b')$ in~$G_a$ and~$G_b$, respectively, and obtain
	    \begin{align*}
	    	M\coloneqq M_{\alpha} &\setminus \{(v_a,v_1),(v_1,v_2),\dots,(v_{k-1},v_k),(v_a',v_1'),(v_1',v_2'),\dots,(v_{k-1}',v_k')\} \\
	    	&\cup \{(v_a,v_a'),(v_b,v_b'),(v_k,v_b),(v_k',v_b')\}.
	    \end{align*}
	    All vertices in the induced subgraph~$G'[M]$ of~$G'$ have an even degree.
	    Thus, there exists an Euler tour~$\seq$ of~$G'[M]$ that visits all vertices of~$G'$ and has~$\alpha$-cost at most
	    \begin{align*}
	    	c_{\alpha}(\seq)
	    	\leq c_{\alpha}(\seq_{\alphatsp})-2k\alpha + 2 + c_e + c_{e'}
	    	\leq c_{\alpha}(\seq_{\alphatsp})-2\frac{3}{\alpha}\alpha + 6 = c_{\alpha}(\seq_{\alphatsp}).
	    \end{align*}
	    Thus, the tour~$\seq$ visits all vertices in~$G'$, and its~$\alpha$-cost is no larger than the~$\alpha$-cost of~$\seq_{\alphatsp}$.
	    We repeat this procedure at most~$\frac{n}{2}$ times to obtain the result.
	    \renewcommand{\qedsymbol}{$\meindreieck$}
	\end{proof}
	
	\begin{claim}\label{cl:paths-once}
	    Without loss of generality, we can assume that no~$v_a$--$v_b$-path is traversed twice by~$\seq_{\alphatsp}$.
	\end{claim}
	\begin{proof}[Proof of \Cref{cl:paths-once}.]
		By \Cref{cl:full-paths}, all edges of a~$v_a$--$v_b$-path~$P$ are traversed consecutively.
		Assume there is a~$v_a$--$v_b$-path~$P$ in~$G'$ that is traversed twice by~$\seq_{\alphatsp}$.
		Let~$v_a,v_1,v_2,\dots,v_k,v_b$ be the vertices on that~$v_a$--$v_b$-path.
		Since the number of vertices~$n$ in~$G_a$ and~$G_b$ is even, there exists a second pair of vertices~$v_a'$ and~$v_b'$ such that the~$v_a'$--$v_b'$-path~$P'$ is also traversed twice by~$\seq_{\alphatsp}$.
		Let~$M_{\alpha}$ be the multiset of edges that appear in~$\seq_{\alphatsp}$.
	    We define a new multiset of edges~$M$ by deleting one copy of each edge on paths~$P$ and~$P'$ and adding the two edges~$(v_a,v_a')$ and~$(v_b,v_b')$ in~$G_a$ and~$G_b$, respectively, i.e.,
	    \begin{align*}
	    	M\coloneqq M_{\alpha} &\setminus \{(v_a,v_1),(v_1,v_2),\dots,(v_{k-1},v_k),(v_k,v_b),(v_a',v_1'),(v_1',v_2'),\dots,(v_{k-1}',v_k'),(v_k',v_b')\} \\
	    	&\cup \{(v_a,v_a'),(v_b,v_b')\}.
	    \end{align*}
		All vertices in the induced subgraph~$G'[M]$ of~$G'$ have an even degree.
	    Thus, there exists an Euler tour~$\seq$ of~$G'[M]$ that visits all vertices of~$G'$ and has~$\alpha$-cost at most
	    \begin{align*}
	    	c_{\alpha}(\seq)
	    	\leq c_{\alpha}(\seq_{\alphatsp})-2(k+1)\alpha + c_{(v_a,v_a')} + c_{(v_b,v_b')}
	    	\leq c_{\alpha}(\seq_{\alphatsp})-2\frac{3+\alpha}{\alpha}\alpha + 4
	    	\leq c_{\alpha}(\seq_{\alphatsp}).
	    \end{align*}
	    Thus, the tour~$\seq$ visits all vertices in~$G'$, and its~$\alpha$-cost is no larger than the~$\alpha$-cost of~$\seq_{\alphatsp}$.
	     We repeat this procedure at most~$\frac{n}{2}$ times to obtain the result.
	    \renewcommand{\qedsymbol}{$\meindreieck$}
	\end{proof}
	
	\begin{claim}\label{cl:double-vertex}
	    Without loss of generality, we can assume that no vertex~$v\in V'$ is visited twice by~$\seq_{\alphatsp}$.
	\end{claim}
	\begin{proof}[Proof of \Cref{cl:double-vertex}.]
		By \Cref{cl:paths-once}, no intermediate vertex of a~$v_a$--$v_b$-path~$P$ is visited twice.
		Let~$v$ be a vertex visited twice by~$\seq_{\alphatsp}$ and assume, without loss of generality, that~$v\in G_a$.
		Let~$M_{\alpha}$ be the multiset of edges that appear in~$\seq_{\alphatsp}$.
		Then vertex~$v$ has an even degree of at least~$4$ in the induced subgraph~$G'[M_{\alpha}]$ of~$G'$.
		By Claims~\ref{cl:full-paths} and~\ref{cl:paths-once}, the edge~$(v,v_1)$ is contained once in~$M_v$.
		Hence, vertex~$v$ has degree at least~$3$ in the induced subgraph~$G_a[M_\alpha]$.
        Since all edges are traversed at most twice in~$\seq_{\alphatsp}$, the subgraph~$G_a[M_\alpha]$ contains at least two different edges adjacent to~$v$.
        Let $P=(e_1,\dots,e_r)$ be a simple inclusionwise-maximal path in~$G_a[M_\alpha]$ that contains two such edges and for every~$1<i<r$ the edge $e_i$ is contained twice in~$G_a[M_\alpha]$. 
        Note that any inner vertex of this path has degree at least~$3$ in~$G_a[M_\alpha]$.
        Thus, if~$P$ is a cycle, the graph~$G'[M_{\alpha}\setminus P]$ still contains an Euler tour with strictly fewer edges than~$\seq_{\alphatsp}$ and cost at most~$c_{\alpha}(\seq_{\alphatsp})$.
        Otherwise,~$P$ is an undirected path starting in some vertex~$u$ and ending in vertex~$w$.
        We shortcut~$P$ by deleting all its edges and instead adding the edge~$(u,w)$.
        Since all intermediate vertices had degree at least~$4$ in~$G'[M_\alpha]$, the remaining edge set~$M_\alpha \setminus P$ still contains an Euler tour~$\seq$.        
        Note that, edges~$e_1$ and~$e_r$ appeared only once in~$G_a[M_\alpha]$. Thus, their removal saves at least a cost of~$2$. 
        On the other hand, adding edge~$(u,w)$ increases the tour's cost by at most~$2$.
	    
	    We repeat this procedure at most~$n$ times to obtain the result.
		\renewcommand{\qedsymbol}{$\meindreieck$}
	\end{proof}
	
	Since the procedures in the proofs of Claims~\ref{cl:full-paths}--\ref{cl:double-vertex} can be executed in polynomial time in the size of~$G$, we may assume, without loss of generality, that the tour~$\seq_{\alphatsp}$ is of the form
	\begin{align}\label{eq:form-alpha-tsp}
		\seq_{\alphatsp} =((v_a,w_a),P_w,(w_b,u_b),P_u,\dots,P_x,(x_b,v_b),P_v),
	\end{align}
	where~$P_y$ is the unique~$y_a$--$y_b$-path in~$G'$.
	In particular, all vertices are visited exactly once by~$\seq_{\alphatsp}$ and, hence, there is no pair of edges~$(v_a,w_a)$ and~$(v_b,w_b)$ that both appear in~$\seq_{\alphatsp}$ (under the trivial assumption of~$n>2$).
	See \Cref{fig:construction-alpha-tsp-1} for an example of such a solution~$\seq_{\alphatsp}$.
	Let~$\bar{E}$ be the edges of~$\seq_{\alphatsp}$ that correspond to edges in~$G$.
	Then,~$\bar{E}$ is a simple set and all vertices have degree~$2$ in~$G[\bar{E}]$.
	Hence, there exists a unique tour~$\seq_{\textsc{tsp}}$ in~$G$ that uses exactly the edge set~$\bar{E}$.
	See \Cref{fig:construction-alpha-tsp-2} for the corresponding solution~$\seq_{\textsc{tsp}}$.
	We summarize the previous results in the following algorithm~$\algtsp$.
	
	\begin{figure}[t]
	\scriptsize
	\centering
	\begin{subfigure}[b]{.65\textwidth}
	\begin{center}
	\begin{tikzpicture}[xscale=0.5,yscale=0.5,shorten > = 0pt]
		\begin{scope}[xshift=10cm]
			\foreach \i in {1,...,6} {
		    \node[state] (m\i) at ({(360/6) * (\i - 1) +20}:3) {};
			}
			\foreach \i in {1,...,6} {
			    \foreach \j in {\i,...,6} {
			        \ifnum\i<\j
			            \draw[] (m\i) -- (m\j);
			        \fi
			    }
			}
			\node[label=right:{$G_b$}] (G1) at(0,5) {};
		\end{scope}
		\foreach \i in {1,...,6} {
	    \node[state] (n\i) at ({(360/6) * (\i - 1) +20}:3) {};
		}
		\foreach \i in {1,...,6} {
		    \foreach \j in {\i,...,6} {
		        \ifnum\i<\j
		            \draw[] (n\i) -- (n\j);
		        \fi
		    }
		}
		\node[label=right:{$G_a$}] (G1) at (0,5) {};

		\foreach \j in {1,...,6} {
		\draw[\myred,thick] (n\j)--(m\j);
		}
		\draw[thick, \myred] (n2)--(n3);
		\draw[thick, \myred] (m3)--(m4);
		\draw[thick, \myred] (n4)--(n6);
		\draw[thick, \myred] (m6)--(m1);
		\draw[thick, \myred] (n1)--(n5);
		\draw[thick, \myred] (m5)--(m2);

		\foreach \j in {1,...,6} {
		\foreach \i in {1,...,6} {
	    \node[circle,draw=gray,thick,fill=gray,minimum size=1mm,inner sep=0pt] (a\i) at ($ (n\j)!{\i/7}!(m\j)$) {};
	  	}
	  	}

	\end{tikzpicture}
	\end{center}
	\caption{}
	\label{fig:construction-alpha-tsp-1}
	\end{subfigure}
	\begin{subfigure}[b]{.3\textwidth}
	\begin{center}
	\begin{tikzpicture}[xscale=0.5,yscale=0.5,shorten > = 0pt]
		\foreach \i in {1,...,6} {
	    \node[state] (n\i) at ({(360/6) * (\i - 1) +20}:3) {};
		}
		\foreach \i in {1,...,6} {
		    \foreach \j in {\i,...,6} {
		        \ifnum\i<\j
		            \draw[] (n\i) -- (n\j);
		        \fi
		    }
		}
		\draw[thick, \myred] (n2)--(n3);
		\draw[thick, \myred] (n3)--(n4);
		\draw[thick, \myred] (n4)--(n6);
		\draw[thick, \myred] (n6)--(n1);
		\draw[thick, \myred] (n1)--(n5);
		\draw[thick, \myred] (n5)--(n2);
	\end{tikzpicture}
	\end{center}
	\caption{}
	\label{fig:construction-alpha-tsp-2}
	\end{subfigure}
	\caption{(a)~The solution~$\seq_{\alphatsp}$ on instance~$I_{\alphatsp}$ for~$\alpha$-TSP(1,2) in \textcolor{\myred}{red}.
	(b)~The corresponding solution~$\seq_{\textsc{tsp}}$ on instance~$I_{\textsc{tsp}}$ for \textsc{TSP}(1,2) in \textcolor{\myred}{red}.
	}
	\label{fig:construction-alpha-tsp}
	\end{figure}
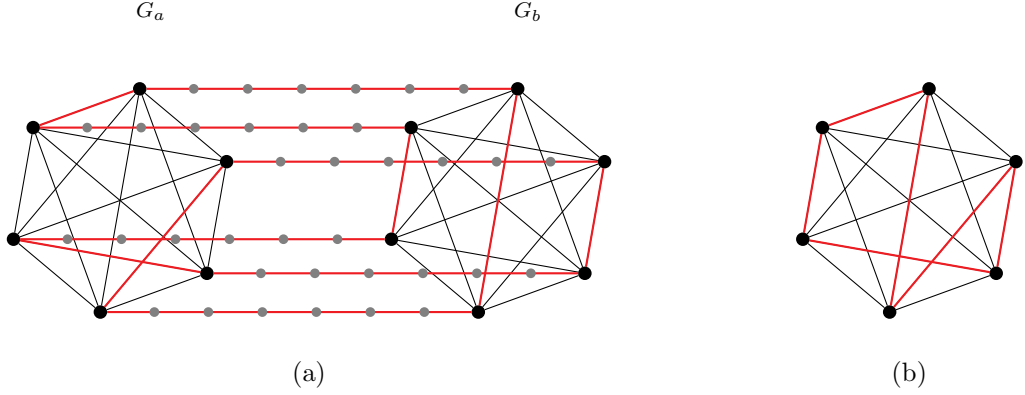

	\paragraph{Construction of the Algorithm for \textsc{TSP}(1,2).}
	For a given \textsc{TSP}(1,2) instance~$I_{\textsc{tsp}}$ on graph~$G$ as input, the algorithm~$\algtsp$ is defined as follows.
	\begin{enumerate}
		\item Construct the corresponding~$\alpha$-\textsc{TSP}(1,2) instance~$I_{\alphatsp}$ on graph~$G'$.
		\item Run the~$(1+\eps)$-approximation algorithm~$\algalpha$ on instance~$I_{\alphatsp}$ and obtain the sequence~$\seq_{\alphatsp}$.
		\item Transform~$\seq_{\alphatsp}$ to be of the form in~\eqref{eq:form-alpha-tsp}.
		\item Construct and return the simple tour~$\seq_{\textsc{tsp}}$ of~$G$.
	\end{enumerate}
	
	By assumption, the~$(1+\eps)$-approximation algorithm~$\algalpha$ has a polynomial running time in the size of~$G'$.
	Since the value~$\alpha\in(0,1]$ is constant, the number of vertices~$|V'|$ of~$G'$ can be bounded by~$|V'|=2n+nk\leq n\big(3+\frac{3}{\alpha}\big)$ which is polynomially bounded by the size of~$G$.
	The transformation of~$\seq_{\alphatsp}$ and the construction of~$I_{\alphatsp}$ and~$\seq_{\textsc{tsp}}$ can also be done in polynomial time.
	Hence, the algorithm~$\algtsp$ has a polynomial running time in the size of~$G$.
	
	It remains to analyze the approximation guarantee of~$\algalpha$.
	In particular, we will show that if~$\eps>0$ can be arbitrarily small, the algorithm~$\algalpha$ is a polynomial-time~$\gamma$-approximation algorithm for \textsc{TSP}(1,2) with~$\gamma<1+\varrho$.
	Since such an algorithm can only exist if~$\mathsf{P}=\mathsf{NP}$, this will finish the proof of \Cref{lem:hardness-alpha-tsp}.
	To this end, we use the following notation:~$\algalpha(I_{\alphatsp})\coloneqq c_{\alpha}(\seq_{\alphatsp})$ and~$\algtsp(I_{\textsc{tsp}})\coloneqq c_{\alpha}(\seq_{\textsc{tsp}})$.
	Furthermore, let~$\opt_{\alpha\textsc{-tsp}}(I_{\alphatsp})$ and~$\opt_\tsp(I_{\textsc{tsp}})$ denote the cost of the optimal solutions for~$\alpha$-\textsc{TSP}(1,2) on instance~$I_{\alphatsp}$ and for \textsc{TSP}(1,2) on instance~$I_{\textsc{tsp}}$, respectively.
	
	Using the previous observations, a feasible solution for~$\alpha$-\textsc{TSP}(1,2) on instance~$I_{\alphatsp}$ can be constructed by the traversal of the optimal \textsc{TSP}(1,2) solution for~$G$ applied on~$G_a$ and~$G_b$ where one switches to the other copy whenever a new vertex is visited.
	Since no edge is used more than once, this construction gives an upper bound on the~$\alpha$-cost of~$\opt_{\alpha\textsc{-tsp}}(I_{\alphatsp})$ by
	\begin{align*}
		\opt_{\alpha\textsc{-tsp}}(I_{\alphatsp})
		\leq (k+1)n + \opt_\tsp(I_{\textsc{tsp}}).
	\end{align*}
	Similarly, the tours~$\seq_{\textsc{tsp}}$ and~$\seq_{\alphatsp}$ do not use any edge more than once, i.e., the discount factor~$\alpha$ never applies.
	Thus, the~$\alpha$-cost of the solution~$\algtsp(I_{\textsc{tsp}})$ is the same as the~$\alpha$-cost of the solution~$\algalpha(I_{\textsc{tsp}})$ without the cost for the~$v_a$--$v_b$-paths, i.e.,
	\begin{align*}
		\algtsp(I_{\textsc{tsp}})
		= \algalpha(I_{\alphatsp}) -(k+1)n.
	\end{align*}
	Note that an optimal solution for the instance~$I_{\textsc{tsp}}$ contains at least~$n$ edges of cost at least~$1$, i.e.,~$n\leq \opt_{\textsc{tsp}}(I_{\textsc{tsp}})$.
	Since the algorithm~$\algalpha$ is a~$(1+\eps)$-approximation algorithm for~$\alpha$-\textsc{TSP}(1,2), we conclude that
	\begin{align*}
		\algtsp(I_{\textsc{tsp}})
		&= \algalpha(I_{\alphatsp}) -(k+1)n \\
		&\leq (1+\eps)\opt_{\alpha\textsc{-tsp}}(I_{\alphatsp}) -(k+1)n \\
		&\leq (1+\eps)(k+1)n + (1+\eps)\opt_\tsp(I_{\textsc{tsp}}) -(k+1)n\\
		&\leq (1+\eps(k+2)) \opt_\tsp(I_{\textsc{tsp}})).
	\end{align*}
	Hence, the algorithm~$\algtsp$ is a polynomial-time~$\gamma$-approximation algorithm for \textsc{TSP}(1,2) with $\gamma\leq 1+\eps(k+2)$.
	Setting~$\eps<\frac{\varrho}{k+2}$, we obtain a polynomial-time~$\gamma$-approximation algorithm for \textsc{TSP}(1,2) with~$\gamma<1+\varrho$.
	This would imply~$\mathsf{P}=\mathsf{NP}$ and thus, for every constant~$\alpha\in(0,1]$, there exists a constant~$\eps>0$ such that there is no polynomial-time~$(1+\eps)$-approximation algorithm for~$\alpha$-\textsc{TSP(1,2)}, unless~$\mathsf{P}=\mathsf{NP}$.
	This finishes this proof.	
\end{proof}

We want to use the inapproximability result for~$\alpha$-TSP(1,2) to prove \Cref{thm:hardness-psp-not-0} as follows.
For a given~$\alpha$-TSP(1,2) instance~$I_{\alphatsp}$ with~$\alpha\in(0,1]$, we construct an instance~$I_{\alphapsp}$ for the \psp with discount factor~$\alpha$ ($\alpha$-PSP).
We then assume, for contradiction, that there exists a polynomial-time~$(1+\eps)$-approximation algorithm~$\alg_{\alphapsp}$ for~$\alpha$-PSP for arbitrary small~$\eps>0$.
We apply this algorithm to~$I_{\alphapsp}$.
Based on this solution, we construct a solution for the~$\alpha$-TSP(1,2) instance~$I_{\alphatsp}$.
Finally, we show that this yields a polynomial-time approximation algorithm for~$\alpha$-TSP(1,2) with arbitrarily small approximation guarantee, contradicting \Cref{lem:hardness-alpha-tsp}, unless~$\mathsf{P}=\mathsf{NP}$.
The construction of the~$\alpha$-PSP instance and the proof of the hardness of approximation is similar to the proof of \Cref{thm:hardness-esp}.

\paragraph{Construction of the~$\alpha$-PSP Instance.}
Let~$I_{\alphatsp}=(G,(c_e)_{e\in E},\alpha)$ be an instance of~$\alpha$-\textsc{TSP}(1,2) on the undirected complete graph~$G=(V, E)$ with edge costs~$c_e\in \{1,2\}$ for all~$e\in E$. 
We construct the instance~$I_{\alphapsp}=(G',(w_v)_{v\in V'},(c'_e)_{e\in E'},\alpha)$ for~$\alpha$-PSP as follows.
First, The graph~$G'=(V',E')$ consists of~$k\in\N$ copies~$G_1,\dots,G_k$ of~$G$ and an additional vertex~$s$, the start vertex.
The constant number~$k$ of copies will be determined later.
Let~$v\in V$ be an arbitrary but fixed vertex of~$G$.
All copies of~$v$ in~$G_1,\dots,G_k$ are connected to~$s$ by an edge of cost~$a$ with
\begin{align*}
	a\coloneqq 4(n-1),
\end{align*} 
where~$n\coloneqq |V|$.
All edges within some copy~$G_i$ are assigned the same cost as in the original graph~$G$. 
Each vertex~$v\in V'\setminus \{s\}$ has weight~$w_v=\frac{1}{n}$ and the start vertex~$s$ has weight~$0$. 
Hence, each copy~$G_i$ has a total weight of 1.
Finally, the discount factor is set to~$\alpha_e\coloneqq \alpha$ for all edges~$e\in E'$.
This finishes the construction of the PSP instance~$I_{\alphapsp}$.
We refer to \Cref{fig:construction-psp} for an illustration of the construction.

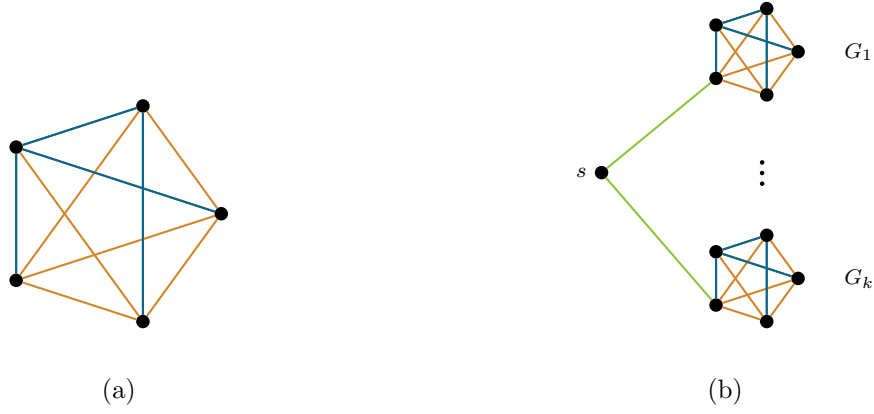
\begin{figure}[t]
\scriptsize
\centering
\begin{subfigure}[b]{.48\textwidth}
\begin{center}
\begin{tikzpicture}[xscale=0.5,yscale=0.5,shorten > = 0pt]
	\foreach \i in {1,...,5} {
    \node[state] (n\i) at ({360/5 * (\i - 1)}:3) {};
	}
	\foreach \i in {1,...,5} {
	    \foreach \j in {\i,...,5} {
	        \ifnum\i<\j
	            \draw[thick, \myyellow] (n\i) -- (n\j);
	        \fi
	    }
	}
	\draw[thick, \myblue] (n1) -- (n3);
	\draw[thick, \myblue] (n2) -- (n3);
	\draw[thick, \myblue] (n4) -- (n3);
	\draw[thick, \myblue] (n5) -- (n2);
\end{tikzpicture}
\end{center}
\caption{}
\label{fig:construction-psp-1}
\end{subfigure}
\begin{subfigure}[b]{.48\textwidth}
\begin{center}
\begin{tikzpicture}[xscale=0.2,yscale=0.2,shorten > = 0pt]
	\begin{scope}[yshift=15cm]
			\foreach \i in {1,...,5} {
	    \node[state] (m\i) at ({360/5 * (\i - 1)}:3) {};
		}
		\foreach \i in {1,...,5} {
		    \foreach \j in {\i,...,5} {
		        \ifnum\i<\j
		            \draw[thick, \myyellow] (m\i) -- (m\j);
		        \fi
		    }
		}
		\draw[thick, \myblue] (m1) -- (m3);
		\draw[thick, \myblue] (m2) -- (m3);
		\draw[thick, \myblue] (m4) -- (m3);
		\draw[thick, \myblue] (m5) -- (m2);
		\node[label=right:{$G_1$}] (G1) at (5,0) {};
	\end{scope}
	\foreach \i in {1,...,5} {
    \node[state] (n\i) at ({360/5 * (\i - 1)}:3) {};
	}
	\foreach \i in {1,...,5} {
	    \foreach \j in {\i,...,5} {
	        \ifnum\i<\j
	            \draw[thick, \myyellow] (n\i) -- (n\j);
	        \fi
	    }
	}
	\draw[thick, \myblue] (n1) -- (n3);
	\draw[thick, \myblue] (n2) -- (n3);
	\draw[thick, \myblue] (n4) -- (n3);
	\draw[thick, \myblue] (n5) -- (n2);
	\node[label=right:{$G_k$}] (G1) at (5,0) {};
	\node[label=right:{\huge~$\vdots$}] (G1) at (-2,7.5) {};
	\node[state,label=left:{$s$}] (s) at (-10,7) {};
	\draw[thick, \mygreen] (s) -- (m4);
	\draw[thick, \mygreen] (s) -- (n4);

\end{tikzpicture}
\end{center}
\caption{}
\label{fig:construction-psp-2}
\end{subfigure}
\caption{(a)~Instance~$I_{\alphatsp}$ for~$\alpha$-\textsc{TSP}(1,2).
\textcolor{\myblue}{Blue} edges have cost~$1$ and \textcolor{\myyellow}{orange} edges have cost~$2$.
(b)~Instance~$I_{\alphapsp}$ for~$\alpha$-PSP constructed from~$I_{\alphatsp}$.
All vertices have weight~$\smash{\frac{1}{5}}$.
\textcolor{\myblue}{Blue} edges have cost~$1$, \textcolor{\myyellow}{orange} edges have cost~$2$, and \textcolor{\mygreen}{green} edges have cost~$a=4(n-1)$.}
\label{fig:construction-psp}
\end{figure}

To prove \Cref{thm:hardness-psp-not-0} we make some assumptions on the sequence~$\seq_\alg$ obtained from the polynomial-time~$(1+\eps)$-approximation algorithm~$\alg_{\alphapsp}$ on instance~$I_{\alphapsp}$.
In this direction, we call~$\seq_\alg$ \emph{structured} if all edges of a copy~$G_i$ and the edge connecting~$G_i$ to~$s$ are consecutive in~$\seq_\alg$.
In other words, the sequence~$\seq_\alg$ visits each copy only once and collects all its weight before returning to the start vertex to visit the next copy.
We show that we can transform the solution~$\seq_\alg$ to be structured without increasing its total~$\alpha$-latency.

\begin{lemma}\label{lem:ah-consecutive}
    Given a solution~$\seq_\alg$ to the~$\alpha$-PSP instance~$I_{\alphapsp}$, we can construct a structured solution~$\seq'_\alg$ in polynomial time such that~$\lat_{\alpha}(\seq'_\alg) \leq \lat_{\alpha}(\seq_\alg)$.
\end{lemma}
\begin{proof}
	If~$\seq_\alg$ is structured, there is nothing left to show.
	Thus, we assume that~$\seq_\alg$ is not structured, i.e., there exists some copy~$G_i$ such that not all edges belonging to~$G_i$ or connecting~$G_i$ to~$s$ are consecutive in~$\seq_\alg$.
	We write~$\seq_\alg$ as a concatenation of (consecutive) subsequences
	\begin{align*}
	\seq_\alg = \seq_1 + \seq_2 + \seq_3 + \dots + \seq_{2p} + \seq_{2p+1}, 
	\end{align*}
	for some~$p > 1$
	such that the subsequences with even index~$\seq_2,\seq_4,\dots,\seq_{2p}$ are the inclusion-wise maximal subsequences of~$\seq_\alg$ consisting only of edges belonging to~$G_i$ or connecting~$G_i$ to~$s$.
	In particular, each subsequence~$\seq_i$ is a tour, i.e., it starts and ends in vertex~$s$.
	The tours~$\seq_1,\seq_3,\dots,\seq_{2p+1}$ with odd index are the inclusion-wise maximal subsequences of the remaining edges in~$\seq_\alg$.
	Note that~$\seq_1$ and~$\seq_{2p+1}$ may be empty, but all other tours with odd indices are non-empty. 
	For some subsequence~$\hat{\seq}$ of~$\seq_\alg$, we denote by~$c'(\hat{\seq})$ its contribution to the total~$\alpha$-cost of the entire sequence~$\seq_\alg$.
	In particular, we have~$c'(\hat{\seq})\leq c_{\alpha}(\hat{\seq})$ for all subsequences~$\hat{\seq}$.
	Further, we denote by~$t(\hat{\seq})$ the number of vertices that~$\hat{\seq}$ visits and that have not been visited before.
	For every subsequence~$\hat{\seq}$ of~$\seq_\alg$ with~$t(\hat{\seq})>0$, we define the \emph{ratio} of~$\hat{\seq}$ as 
	\begin{align*}
		r(\hat{\seq})\coloneqq \frac{c'(\hat{\seq})}{t(\hat{\seq})}.
	\end{align*}
	
	\begin{claim}\label{cl:t>1-psp}
	    Without loss of generality, we can assume that~$t(\seq_j)\geq 1$ for all~$j\in\{2,3,\dots,2p\}$.
	\end{claim}
	\begin{proof}[Proof of \Cref{cl:t>1-psp}.]
	    Assume that there exists some~$j\in\{2,3,\dots,2p\}$ such that~$t(\seq_j)=0$.
	    Then, we can swap the positions of~$\seq_j$ and~$\seq_{j+1}$ and continue with the newly obtained sequence, which has fewer subsequences.
	    The obtained sequence is still a feasible solution to the \psp, as each subsequence~$\seq_i$ starts and ends in~$s$.
	    By doing so, we only improve the total~$\alpha$-latency since no exploration of any vertex is postponed.
	    \renewcommand{\qedsymbol}{$\meindreieck$}
	\end{proof}
	
	With \Cref{cl:t>1-psp}, the ratio~$r(\seq_j)$ is well-defined for all~$j\in\{2,3,\dots,2p\}$. 
	
	\begin{claim}\label{cl:r>=2-psp}
	    Without loss of generality, we can assume that~$r(\seq_2)\geq 4(1+\alpha)$.
	\end{claim}
	\begin{proof}[Proof of \Cref{cl:r>=2-psp}.]
		The assumption of~$p>1$ together with~\Cref{cl:t>1-psp} yields the existence of the subsequence~$\seq_4$ with~$t(\seq_4)\geq 1$.
		This implies that~$t(\seq_2)\leq n-1$.
		Since each edge in~$G_i$ has cost at least~$1$, the~$\alpha$-cost of traversing edges of~$\seq_2$ that lie in~$G_i$ is at least~$t(\seq_2)-1$.
		Furthermore,~$\seq_2$ is the first subsequence that visits~$G_i$.
		Hence, it contains the edge that connects~$s$ and~$G_i$ that has not been traversed before.
		In total, this yields a lower bound of~$c'(\seq_2)\geq (1+\alpha)a+ t(\seq_2)-1$.
		We conclude
		\begin{align*}
			r(\seq_2)
			=\frac{c'(\seq_2)}{t(\seq_2)}
			\geq \frac{(1+\alpha)a+t(\seq_2)-1}{t(\seq_2)}
			\geq \frac{(1+\alpha)4(n-1) + n-2}{n-1}
			\geq 4(1+\alpha),
		\end{align*}
		where we used that~$n\geq 2$.
		\renewcommand{\qedsymbol}{$\meindreieck$}
	\end{proof}
	
	Assume we exchanged the order of~$\seq_{2p}$ and~$\seq_{2p-1}$.
	Then we can delete the first edge of~$\seq_{2p}$, and the last edge of~$\seq_{2p-2}$ as both edges are the one edge that connects vertex~$s$ to copy~$G_i$.
	We capture this observation by introducing a reduced ratio~$\tilde{r}$ for~$\seq_{2p}$ by
	\begin{align*}
		\tilde{r}(\seq_{2p})
		\coloneqq \frac{c'(\seq_{2p})-2\alpha a}{t(\seq_{2p})}.
	\end{align*}
	
	\begin{claim}\label{cl:r<=2-psp}
	    Without loss of generality, we can assume that~$\tilde{r}(\seq_{2p})\leq 4(1+\alpha)$.
	\end{claim}
	\begin{proof}[Proof of \Cref{cl:r<=2-psp}]
		Assume that~$\tilde{r}(\seq_{2p})> 4(1+\alpha)$.
		We show how we can transform~$\seq_{2p}$ to obtain~$\tilde{r}(\seq_{2p})\leq 4(1+\alpha)$ without increasing the total~$\alpha$-latency of~$\seq_\alg$.
		To this end, let~$v$ be the last new vertex visited by~$\seq_{2p}$ and let~$v_0$ be the vertex of~$G_i$ connected to~$s$.
		Let~$\seq_v$ be the subsequence of~$\seq_{2p}$ that starts in~$v$ and ends in~$v_0$.
		Note that we can exchange~$\seq_v$ by either edge~$e=(v,v_0)$ of cost at most~$2$ or by a path~$P$ that only consists of edges that have already been used in~$G_i$. Thus this subsequence contributes at most~$c'(\seq_v)\leq \min\{2,2(n-1)\alpha\}$ to the total~${\alpha}$-cost of~$\seq$.
		Denote by~$\seq_{2p}'$ the sequence obtained by exchanging~$\seq_v$ for either~$e$ or~$P$ based on which one contributes less to the total~$\alpha$-cost.
		If~$\tilde{r}(\seq_{2p}')\leq 4(1+\alpha)$, we are done.
		Thus, assume~$\tilde{r}(\seq_{2p}')> 4(1+\alpha)$ and denote by~$(e_1,\dots,e_z)$ the maximal subsequence of~$\seq_{2p}$ that only consists of edges in~$G_i$ and ends with edge~$e_z$, where~$e_z$ is the edge that connects the last vertex~$v$.
		We obtain
		\begin{align*}
			\tilde{r}(\seq_{2p})
			=\frac{c'((e_1,\dots,e_z))+c'(\seq_v)}{t(\seq_{2p})}
			> 4(1+\alpha).
		\end{align*}
		With~$c'(\seq_v)\leq 2$,~$t(\seq_{2p})\geq 1$, and~$t((e_1,\dots,e_z))\leq 1$, we obtain a lower bound on the ratio of the subsequence~$(e_1,\dots,e_z)$ by
		\begin{align*}
			r((e_1,\dots,e_z))
			\geq \frac{c'((e_1,\dots,e_z))}{t((e_1,\dots,e_z))}
			\geq \frac{c'((e_1,\dots,e_z))+c'(\seq_v)}{t(\seq_{2p})} -\frac{c'(\seq_v)}{t(\seq_{2p})}
			> 4(1+\alpha)-2
			> 2.
		\end{align*}
		Let~$\Bar{\seq}=(e_x,\dots,e_z)$ be the shortest subsequence of~$(e_1,\dots,e_z)$ such that~$r(\Bar{\seq})>2$.
		Let~$\{e^*_1,\dots,e^*_{y}=e_z\}$ be the set of edges that connect a new vertex in the order as they appear in~$\Bar{\seq}$.
	    We claim that for any subsequence~$\Bar{\seq}_j=(e_x,e_{x+1},\dots,e^*_j)$ with~$j\in [y]$ it holds that~$r(\Bar{\seq}_j)>2$.
	    Assume for contradiction that there exists some~$j\in [y]$ such that~$r(\Bar{\seq}_j)\leq 2$.
	    Then let~$\Bar{\seq}_{-j}$ be such that~$\Bar{\seq}$ is a concatenation of~$\Bar{\seq}_{j}$ and~$\Bar{\seq}_{-j}$.
	    Since~$\Bar{\seq}$ is the shortest contiguous subsequence of~$\seq_{2p}$ that ends in~$e_z$ such that~$r(\Bar{\seq})>2$ holds, it follows that~$r(\Bar{\seq}_{-j})\leq 2$, otherwise~$\Bar{\seq}$ would not be minimal.
	    In total, this yields
	    \begin{align*}
	        2
	        < r(\Bar{\seq})
	        =\frac{c'(\Bar{\seq})}{t(\Bar{\seq})}
	        =\frac{c'(\Bar{\seq}_{j})+c'(\Bar{\seq}_{-j})}{t(\Bar{\seq}_{j})+t(\Bar{\seq}_{-j})}
	        \leq \frac{2t(\Bar{\seq}_{j})+2t(\Bar{\seq}_{-j})}{t(\Bar{\seq}_{j})+t(\Bar{\seq}_{-j})}
	        =2,
	    \end{align*}
	    a contradiction.
	    Hence, for any subsequence~$\Bar{\seq}_j=(e_x,\dots,e^*_j)$ with~$j\in [y]$, it holds that~$r(\Bar{\seq}_j)>2$.
	    Let~$V'(\Bar{\seq}_j)\subseteq V'$ denote the set of the~$t(\Bar{\seq}_j)$ new vertices which are connected by~$\Bar{\seq}_j$ for some~$j\in [y]$.
	    We can remove all edges in~$\Bar{\seq}_j$ and replace them with a path that connects all vertices in~$V'(\Bar{\seq}_j)$ with a single edge of cost at most~$2$ each in the same order as they had been connected in~$\Bar{\seq}_j$.
	    This strictly decreases the~$\alpha$-latency of all vertices in~$V'(\Bar{\seq}_j)$ and in particular the~$\alpha$-latency of the last vertex~$v$.
	    We repeat this procedure at most~$t(\seq_{2p})<n$ times until we have~$r((e_1',\dots,e_z'))\leq 2$ for the new sequence~$(e_1',\dots,e_z')$.
	    Finally, this yields
	    \begin{align*}
			\tilde{r}(\seq_{2p}')
			=\frac{c'((e_1',\dots,e_z'))+c'(\seq_v)}{t(\seq_{2p})}
			= \frac{r((e_1',\dots,e_z')) t(\seq_{2p}) +c'(\seq_v)}{t(\seq_{2p})}
			\leq \frac{2 t(\seq_{2p}) +2}{t(\seq_{2p})}
			\leq 4.
		\end{align*}
	    With~$\alpha>0$, we conclude~$\tilde{r}(\seq_{2p}')\leq 4(1+\alpha)$ as desired.	    
	    Furthermore, the~$\alpha$-latencies of all vertices visited after the traversal of~$\seq_{2p}$ do not increase since we strictly reduced the total~$\alpha$-cost of~$\seq_{2p}$ and thus all~$\alpha$-latencies of vertices visited later.
		\renewcommand{\qedsymbol}{$\meindreieck$}
	\end{proof}

	With these three claims, we now prove \Cref{lem:ah-consecutive}.
    Consider the ratios~$r(\seq_{2p-1})$ and~$\tilde{r}(\seq_{2p})$ and distinguish two cases.
    First, assume~$r(\seq_{2p-1})\geq \tilde{r}(\seq_{2p})$.
    We swap those two subsequences and delete the last edge of~$\seq_{2p-2}$ and the first edge of~$\seq_{2p}$.
    Note that both edges connect vertex~$s$ to~$G_i$, and thus, they contribute~$2\alpha a$ to the~$\alpha$-cost of~$\seq_\alg$.
    We claim that this swap does not increase the total~$\alpha$-latency.
    To this end, note that~$r(\seq_{2p-1}) \geq \tilde{r}(\seq_{2p})$ yields
    \begin{align}
        \frac{c'(\seq_{2p-1})}{t(\seq_{2p-1})}
        &\geq \frac{c'(\seq_{2p})-2\alpha a}{t(\seq_{2p})} &&\Leftrightarrow  &
        c'(\seq_{2p-1}) t(\seq_{2p})
        \geq (c'(\seq_{2p})-2\alpha a) t(\seq_{2p-1}). \label{eq:ah-ratio-psp}
    \end{align}
    The swap causes the~$\alpha$-latency of~$t(\seq_{2p})$ many vertices to decrease by~$c'(\seq_{2p-1})+2\alpha a$ while the~$\alpha$-latency of~$t(\seq_{2p-1})$ many vertices increases by~$c'(\seq_{2p})-2\alpha a$.
    The~$\alpha$-latencies of all vertices visited by~$\seq_{2p+1}$ also decrease by~$2\alpha a$.
    Hence, by \eqref{eq:ah-ratio-psp}, the total~$\alpha$-latency of the new sequence can only decrease.
    
    For the second case, we assume~$r(\seq_{2p-1})< \tilde{r}(\seq_{2p})$ which yields~$r(\seq_{2p-1})<4(1+\alpha)$.
    We then compare~$r(\seq_{2p-2})$ to~$r(\seq_{2p-1})$ and continue recursively with adjacent subsequences until we find the first pair~$\seq_{j}$ and~$\seq_{j-1}$ with~$j\in\{2,3,\dots,2p-2\}$ such that~$r(\seq_{j})\geq r(\seq_{j-1})$.
    This pair exists since~$r(\seq_{2p-1})<4(1+\alpha)$ and~$r(\seq_2)\geq 4(1+\alpha)$.
    Swapping the order of those two subsequences causes the~$\alpha$-latency of~$t(\seq_{j+1})$ many vertices to decrease by at least~$c'(\seq_{j})$ while the~$\alpha$-latency of~$t(\seq_{j})$ many vertices increases by at most~$c'(\seq_{j+1})$.
    Again, the~$\alpha$-latencies of vertices visited later can only decrease.
    Similar to the first case,~$r(\seq_{j}) \geq r(\seq_{j+1})$ yields
    \begin{align*}
        \frac{c'(\seq_{j})}{t(\seq_{j})}
        &\geq \frac{c'(\seq_{j+1})}{t(\seq_{j+1})} &&\Leftrightarrow  &
        c'(\seq_{j}) t(\seq_{j+1})
        \geq c'(\seq_{j+1}) t(\seq_{j}).
    \end{align*}
    Thus, swapping those two subsequences does not increase the total~$\alpha$-latency of~$\seq_\alg$.
    After swapping at most~$p\leq |V|$ pairs of subsequences, the desired property is established for~$G_i$.
    This entire process can be repeated for each copy of~$G$ until the obtained sequence is structured.
    Computing the ratios and performing the swaps of the subsequences takes time polynomial in the length of the sequence; hence, this procedure runs in polynomial time.
\end{proof}

We assume from now on that if~$\seq_\alg$ is structured, it visits the~$k$ copies of~$G$ in the order~$G_1,\dots,G_k$.
We are now ready to prove \Cref{thm:hardness-psp-not-0}.

\begin{proof}[Proof of \Cref{thm:hardness-psp-not-0}]
	By \Cref{lem:hardness-alpha-tsp}, there exists a constant~$\varrho>0$ such that there is no polynomial-time~$(1+\varrho)$-approximation algorithm for~$\alpha$-TSP(1,2), unless~$\mathsf{P}=\mathsf{NP}$.
	We assume, for contradiction, that there exists a polynomial-time~$(1+\eps)$-approximation-algorithm~$\alg_{\alphapsp}$ for the~$\alpha$-PSP and arbitrary small~$\eps>0$.
	Finally, we show how this algorithm implies a polynomial-time~$\gamma$-approximation algorithm~$\algalpha$ for~$\alpha$-TSP(1,2) with~$\gamma<1+\varrho$.
	The algorithm is defined as follows:
	
	\paragraph{Construction of the Algorithm for~$\alpha$-\textsc{TSP}(1,2).}
	For a given~$\alpha$-\textsc{TSP}(1,2) instance~$I_{\alphatsp}$ as input, the algorithm~$\alg_{\alphatsp}$ is defined as follows.
	\begin{enumerate}
		\item Construct the corresponding~$\alpha$-PSP instance~$I_{\alphapsp}$.
		\item \label{it:construct-psp} Run the~$(1+\eps)$-approximation algorithm~$\alg_{\alphapsp}$ on instance~$I_{\alphapsp}$ and obtain the sequence~$\seq_\alg'$.
		\item Compute a corresponding structured sequence~$\seq_\alg$.
		\item Let~$\seq_i$ be the subsequence of~$\seq_\alg'$ in copy~$G_i$ and denote by~$\Seq\coloneqq\{\seq_i : i\in [k]\}$ the set containing the~$k$ subsequences. Each of these tours~$\seq_i$ yields a feasible solution to~$\alpha$-\textsc{TSP}(1,2) on~$I_{\alphatsp}$.
		\item Return~$\seq^*=\arg\min \bigl\{c_{\alpha}(\seq_i) : i\in [k]\bigr\}$.
	\end{enumerate}
	
	First, we argue that the running time of~$\alg_{\alphatsp}$ is polynomially bounded in the size of~$G$.
	By assumption, the running time of~$\alg_{\alphapsp}$ on~$G'$ is polynomially bounded by the size of~$G'$.
	Since the size of $G'$ is polynomially bounded by the size of~$G$ (because~$k$ is a constant that will be determined later), Step~\ref{it:construct-psp} runs in polynomial time.
	The construction of~$I_{\alphapsp}$ and~$\seq^*$ can also be done in polynomial time.
	Thus, with \Cref{lem:ah-consecutive},~$\alg_{\alphatsp}$ is a polynomial-time algorithm.
	
	Next, we analyze the approximation guarantee obtained by~$\alg_{\alphatsp}$.
	We denote the cost of the algorithms~$\alg_{\alphatsp}$ and~$\alg_{\alphapsp}$ on the instances~$I_{\alphatsp}$ and~$I_{\alphapsp}$ by~$\alg_{\alphatsp}(I_{\alphatsp})$ and~$\alg_{\alphapsp}(I_{\alphapsp})$, respectively.
	Therefore, $\alg_{\alphatsp}(I_{\alphatsp})=c_{\alpha}(\seq^*)$ and $\alg_{\alphapsp}(I_{\alphapsp})=C_{\alpha}(\seq_\alg)$.
	Further, we denote the cost of the optimal solutions for~$\alpha$-\textsc{TSP}(1,2) and~$\alpha$-PSP on instances~$I_{\alphatsp}$ and~$I_{\alphapsp}$ by~$\opt_{\alpha\textsc{-tsp}}(I_{\alphatsp})$ and~$\opt_{\alpha\textsc{-psp}}(I_{\alphapsp})$, respectively.
	To prove the claimed approximation ratio of~$\gamma<1+\varrho$, we start by giving an upper bound on~$\alg_{\alphatsp}(I_{\alphatsp})$.
	To this end, let~$\seq_1,\dots,\seq_k$ be the~$\alpha$-TSP(1,2) solutions that~$\alg_{\alphapsp}$ obtains as a byproduct on the~$k$ copies~$G_1,\dots,G_k$ of instance~$I_{\alphapsp}$.
	The upper bound is obtained by assuming that the structured sequence~$\seq_\alg$ collects the total weight~$1$ of each copy~$G_i$ when it visits the first vertex of that copy.
	This yields
	\begin{align*}
	    \alg_{\alphapsp}(I_{\alphapsp})
	    &\geq \sum_{i=1}^k \left(ia+ (i-1)\alpha a + \sum_{j=1}^{i-1} c_{\alpha}(\seq_j) \right)\\
	    &\geq \sum_{i=1}^k \left(ia+ (i-1)\alpha a + \sum_{j=1}^{i-1} c_{\alpha}(\seq^*) \right)\\
	    &= \frac{k(k+1)}{2} a + \frac{k(k-1)}{2}\alpha a + \frac{(k-1)k}{2} c_{\alpha}(\seq^*),
	\end{align*}
	which is equivalent to
	\begin{align}
	    \alg_{\alphatsp}(I_{\alphatsp})=c_{\alpha}(\seq^*)
	    &\leq \frac{2}{(k-1)k}\left( \alg_{\alphapsp}(I_{\alphapsp}) - \frac{k(k+1)}{2} a -\frac{k(k-1)}{2}\alpha a \right). \label{eq:ah-upper-bound-alg-psp}
	\end{align}
	
	Next, we give an upper bound on~$\opt_{\alpha\textsc{-psp}}(I_{\alphapsp})$. 
	To do so, assume we have an optimal solution for~$\alpha$-TSP(1,2) on instance~$I_{\alphatsp}$. 
	Using this optimal solution, we can construct a solution to~$\alpha$-PSP on instance~$I_{\alphapsp}$ as follows.
	First, enter~$G_1$.
	Then traverse~$G_1$ using an optimal tour for~$\alpha$-TSP(1,2) on instance~$I_{\alphatsp}$, i.e.,~$G$.
	Finally, return to vertex~$s$ and continue in the same manner for the remaining copies~$G_2,\dots,G_k$.
	By assuming that the total weight of each copy is collected only when visiting the last vertex, we obtain the following upper bound on~$\opt_{\alpha\textsc{-psp}}(I_{\alphapsp})$
	\begin{align}
	    \opt_{\alpha\textsc{-psp}}(I_{\alphapsp})
	    &\leq \sum_{i=1}^k i(a+\opt_{\alpha\textsc{-tsp}}(I_{\alphatsp}))+ (i-1)\alpha a \notag \\
	    &=\frac{k(k+1)}{2} (a+\opt_{\alpha\textsc{-tsp}}(I_{\alphatsp})) + \frac{k(k-1)}{2}\alpha a \ . \label{eq:ah-upper-bound-opt-psp}
	\end{align}
	Now, combining \eqref{eq:ah-upper-bound-alg-psp} and \eqref{eq:ah-upper-bound-opt-psp} with~$\alg_{\alphapsp}(I_{\alphapsp})\leq (1+\eps) \opt_{\alpha\textsc{-psp}}(I_{\alphapsp})$ yields
	\begin{align*}
	    \alg_{\alphatsp}(I_{\alphatsp})
	    &\leq \frac{2}{(k-1)k}\left( \alg_{\alphapsp}(I_{\alphapsp}) - \frac{k(k+1)}{2} a -\frac{k(k-1)}{2}\alpha a \right) \notag \\
	    &\leq \frac{2}{(k-1)k}\left( (1+\eps) \opt_{\alpha\textsc{-psp}}(I_{\alphapsp}) - \frac{k(k+1)}{2} a -\frac{k(k-1)}{2}\alpha a \right) \notag \\
	    &\leq \frac{2}{(k-1)k}\biggl[ (1+\eps) \left( \frac{k(k+1)}{2} (a+\opt_{\alpha\textsc{-tsp}}(I_{\alphatsp})) + \frac{k(k-1)}{2}\alpha a \right) \notag \\
	    & \phantom{\leq \frac{2}{(k-1)k}\biggl[} - \frac{k(k+1)}{2} a -\frac{k(k-1)}{2}\alpha a \biggr] \notag \\
	    &\leq \frac{k+1}{k-1}\left(a \eps + (1+\eps)\opt_{\alpha\textsc{-tsp}}(I_{\alphatsp}))\right) +\eps \alpha a \notag \\
	    &\leq \frac{k+1}{k-1}\left(1+\eps(5+4\alpha) \right) \opt_{\alpha\textsc{-tsp}}(I_{\alphatsp}).    
	\end{align*}

	Thus, the~$(1+\eps)$-approximation algorithm for~$\alpha$-PSP yields a~$\gamma$-approximation algorithm for~$\alpha$-TSP(1,2) with~$\gamma = \frac{k+1}{k-1}(1+\eps(5+4\alpha))$.
	However, by choosing~$\eps$ and~$k$ such that~$0<\eps<\frac{\varrho k-(\varrho+2)}{(k+1)(5+4\alpha)}$ and~$k > 1+\frac{2}{\varrho}$, we have~$\gamma <1+\varrho$, a contradiction to the approximation hardness of~$\alpha$-TSP(1,2), unless~$\mathsf{P}=\mathsf{NP}$.
	This proves that there exists a constant~$\eps>0$ such that there is no polynomial-time~$(1+\eps)$-approximation algorithm for the \psp with discount factor~$\alpha$, unless~$\mathsf{P}=\mathsf{NP}$.
	This finishes the proof of \Cref{thm:hardness-psp}.	
\end{proof}

\end{document}